\documentclass[graybox,envcountsame,envcountchap]{svmono}
\usepackage{longtable}
\usepackage[small,bf,singlelinecheck=off]{caption}
\usepackage{pifont}

\usepackage{mathptmx}
\usepackage{helvet}
\usepackage{courier}

\usepackage[mathscr]{euscript}

\usepackage{type1cm}         

\usepackage{makeidx}         
\usepackage{graphicx}        
\usepackage{multicol}        
\usepackage[bottom]{footmisc}

\usepackage{overpic}  

\usepackage{bbm}
\usepackage{nicefrac}
\usepackage{cite}

\usepackage[ 
             colorlinks=true,
             linkcolor = blue,
             urlcolor  = blue,
             citecolor = red,
             anchorcolor = green
]{hyperref}

\usepackage[usenames,dvipsnames]{xcolor}
\usepackage{color}
\usepackage{amsmath}
\usepackage{amsfonts}
\usepackage{amssymb}
\usepackage{bbm}

\usepackage{graphicx}	
\usepackage[caption=false]{subfig}
\usepackage{tikz}
\usetikzlibrary{chains}
\usetikzlibrary{fit}
\usepackage{pgflibraryarrows}		
\usepackage{pgflibrarysnakes}		
\usepackage{epsfig}
\usetikzlibrary{shapes.symbols,patterns} 
\usepackage{pgfplots}

\usepackage{makeidx}
\usepackage{dsfont}

\usepackage{afterpage}

\usepackage{wasysym}
\usepackage{wrapfig}

\DeclareFontFamily{U}{mathx}{\hyphenchar\font45}
\DeclareFontShape{U}{mathx}{m}{n}{<-> mathx10}{}
\DeclareSymbolFont{mathx}{U}{mathx}{m}{n}
\DeclareMathAccent{\widebar}{0}{mathx}{"73}

\usepackage[margin=4cm]{geometry}

\newcommand*{\PP}{\mathbb{P}}
\newcommand*{\ee}{\mathrm{e}}

\newcommand*{\cA}{\mathcal{A}}
\newcommand*{\cB}{\mathcal{B}}
\newcommand*{\cC}{\mathcal{C}}

\newcommand*{\cE}{\mathcal{E}}
\newcommand*{\cF}{\mathcal{F}}

\newcommand*{\cI}{\mathcal{I}}

\newcommand*{\cM}{\mathcal{M}}
\newcommand*{\cP}{\mathcal{P}}

\newcommand*{\cR}{\mathcal{R}}
\newcommand*{\cS}{\mathcal{S}}
\newcommand*{\cT}{\mathcal{T}}
\newcommand*{\cU}{\mathcal{U}}

\newcommand*{\cX}{\mathcal{X}}
\newcommand*{\cY}{\mathcal{Y}}

\newcommand*{\N}{\mathbb{N}}
\newcommand*{\R}{\mathbb{R}}
\newcommand*{\C}{\mathbb{C}}

\newcommand*{\indic}{\mathds{1}}

\newcommand*{\eps}{\varepsilon}
\newcommand*{\rank}{\mathrm{rank}}

\newcommand*{\id}{\mathrm{id}}
\newcommand*{\poly}{\mathrm{poly}}

\newcommand*{\spec}{\mathrm{spec}}
\newcommand*{\supp}{\mathrm{supp}}
\newcommand*{\tr}{\mathrm{tr}}
\newcommand*{\ket}[1]{| #1 \rangle}
\newcommand*{\bra}[1]{\langle #1 |}

\newcommand{\proj}[1]{|#1\rangle\!\langle #1|}

\newcommand*{\TPCP}{\mathsf{TPCP}}
\newcommand*{\LL}{\mathsf{L}}
\newcommand*{\St}{\mathsf{S}}
\newcommand*{\Her}{\mathsf{H}}
\newcommand*{\Pos}{\mathsf{P}}
\newcommand*{\Poss}{\mathsf{P}_{\!\!\!+}}
\newcommand*{\U}{\mathsf{U}}
\newcommand*{\V}{\mathsf{V}}
\newcommand*{\QMC}{\mathsf{MC}}

\newcommand*{\Real}{\mathrm{Re}}
\newcommand*{\Immag}{\mathrm{Im}}
\newcommand*{\interior}{\mathrm{int}}

\newcommand*{\conv}{\mathrm{conv}}

\newcommand*{\MD}{D_{\mathbb{M}}}
\newcommand*{\PD}{D_{\mathbb{P}}}

\newcommand*{\ci}{\mathrm{i}} 
\newcommand*{\di}{\mathrm{d}} 

\newcommand{\norm}[1]{\left\lVert#1\right\rVert}
\newcommand{\normU}[1]{\left|\hspace{-0.4mm}\left|\hspace{-0.4mm}\left|#1\right|\hspace{-0.4mm}\right|\hspace{-0.4mm}\right|}

\allowdisplaybreaks    

\makeindex

\smartqed

\begin{document}

\author{David Sutter}
\title{Approximate quantum Markov chains}
\date{\today}


\hypersetup{pageanchor=false}

\maketitle

\frontmatter


\extrachap{Acknowledgements}

First and foremost, I would like to thank my advisor Renato Renner for his encouragement and support. His never ending enthusiasm, optimism and persistency in doing research as well as his precision in thought and communication were highly inspiring and clearly sharpened my mind. His guidance during the last years was outstanding. I was entirely free to work on what I like most, at the same time always knowing that he would immediately interrupt in case I drift off to analyzing meaningless problems. 

I am especially grateful to J\"urg Fr\"ohlich for introducing me to the exciting field of mathematical physics and for carefully listening to my oftentimes vague new ideas. J\"urg's immense knowledge about physics is extraordinary and I enormously enjoyed our regular meetings. Furthermore, my sincere thanks go to Emre Telatar for investing his time in studying my work and for being my co-examiner. I also would like to thank Manfred Sigrist for interesting discussions in the early mornings at the institute and for representing the physics department at my defense.

During the last couple of years I was extremely lucky to collaborate with various brilliant people including Mario Berta, Fr\'ed\'eric Dupuis, Omar Fawzi, Aram W.~Harrow, Hamed Hassani, Raban Iten, Marius Junge, John Lygeros, Peyman Mohajerin Esfahani, Renato Renner, Joseph M.~Renes, Volkher B. Scholz, Tobias Sutter, Marco Tomamichel, Mark M.~Wilde, and Andreas Winter. I would like to thank all of them for their patience and effort to work with me.   

I am very grateful to all members of the Institute for Theoretical Physics at ETH Zurich and in particular the Quantum Information Theory Group people for their support and the friendly atmosphere. It was an exciting and truly wonderful time.

Finally and most importantly, I would like to thank my twin brother, my sister, and my parents for their constant support and encouragement.

\vspace{10mm}
\begin{flushright}  

 Zurich, \today \\
 David Sutter

\end{flushright}  

\hypersetup{pageanchor=true}

\tableofcontents

\mainmatter


\chapter{Introduction}
\label{chapter_intro} 

\abstract{This chapter reviews the concept of a Markov chain for random variables and then introduces a generalization to quantum mechanical systems. We discuss the subtle differences between classical and quantum Markov chains and point out difficulties that show up in the quantum case. Unlike the classical case that has been analyzed and understood well in the past, certain aspects of quantum Markov chains remained mysterious. Only very recently, it has been recognized how an entropic characterization of states that approximately satisfy the Markov condition looks like. This insight justifies the definition of \emph{approximate quantum Markov chains} which this book is about.}
\vspace{8mm}

\noindent Markov chains are named after the Russian mathematician Andrei Markov (1856-1922), who introduced them in 1907. Suppose we have a sequence of random variables $(X_n)_{n \geq 1}$. The simplest model is the case where the random variables are assumed to be pairwise independent. For this scenario many nice results, such as the law of large number or the central limit theorem, are known. At the same time the pairwise independence assumption makes the model rather restrictive.

Markov's idea was to consider a more general dependence structure that however is still simple enough that it can be analyzed rigorously. Informally, his idea was to assume that the random variables $(X_n)_{n \geq 1}$ are ordered in a very specific way.\footnote{We then say $(X_n)_{n\geq 1}$ forms a \emph{Markov chain} in order $X_1 \leftrightarrow X_2\leftrightarrow X_3 \leftrightarrow \ldots$ \,.} This ordering implies that all the information that the random variables $(X_1,\cdots,X_{k-1})$ could have about $X_k$ for any $k>1$ is contained in $X_{k-1}$. More precisely, we require that the collective entire past $(X_1,\ldots,X_{k-2})$ is independent of the collective entire future $(X_{k}, \ldots)$ conditioned on the present $X_{k-1}$. This model has the advantage that in order to describe $X_k$ we only need to remember $X_{k-1}$ and can forget about the past $(X_1,\ldots,X_{k-2})$. This makes the model simple enough that we can prove precise properties and describe its behavior for large values of $n$. At the same time, the model is considerably more general than the pairwise independence assumption which makes it suitable for many situations (see,  e.g.,~\cite{meyn,norris,levin2009markov,aldous-fill-2014}).

Markov chains are intensively studied and have been generalized to the quantum mechanical setup~\cite{accardi83} where random variables are replaced by density operators on a Hilbert space.\footnote{In Section~\ref{sec_introQMC} and in particular in Section~\ref{sec_QMC} we introduce the concept of a quantum Markov chain.} Natural questions that arise are:
\begin{svgraybox}
What are the main differences between classical and quantum Markov chains? What do we know about sequences of random variables that approximately form a Markov chain? Do they approximately behave as (exact) Markov chains?
\end{svgraybox}

This book will answer these questions. We will first introduce the reader to quantum Markov chains and explain how to define a robust version of this concept that will be called \emph{approximate quantum Markov chains}.

In the literature there exists the term ``\emph{short Markov chains}'' which should distinguish the Markov chain between three random variables from infinite chains. Since we only consider Markov chains defined for three random variables in this book we drop the term ``short''.\index{short Markov chain}

\section{Classical Markov chains} 
 Three random variables $X,Y,Z$ with joint distribution $P_{XYZ}$ form a \emph{Markov chain}\index{Markov chain!classical} in order $X\leftrightarrow Y \leftrightarrow Z$ if $X$ and $Z$ are independent conditioned on $Y$. In mathematical terms this can be expressed as
\begin{align} \label{eq_cl_def}
P_{XYZ} \text{ is a Markov chain} \quad \iff \quad P_{XZ|Y}=P_{X|Y}  P_{Z|Y} \, ,
\end{align}
where $P_{X|Y}$ denotes the probability distribution of $X$ conditioned on $Y$. Bayes' theorem directly implies that the right-hand side of~\eqref{eq_cl_def} can be rewritten as $P_{XYZ}=P_{XY} P_{Z|Y}$.
Operationally, the Markov chain condition tells us that all the information the pair $(X,Y)$ has about $Z$ is contained in $Y$. In other words, there is no need to remember $X$ in order to determine $Z$ if we already know $Y$. Suppose we loose the random variable $Z$. The Markov chain condition ensures that it is possible to reconstruct $Z$ by only acting on $Y$ with a stochastic map.\footnote{The reconstruction refers to a stochastically indistinguishable copy which means that if we denote the reconstructed random variable by $Z'$ we require that the probability law of $(X,Y,Z')$ is the same as $(X,Y,Z)$.} More precisely,  
\begin{align} \label{eq_cl_recovery}
P_{XYZ} \text{ is a Markov chain} \quad \iff \quad \exists \, \text{ stochastic matrix } W_{Z|Y} \text{ such that } P_{XYZ}=P_{XY}  W_{Z|Y} \, .
\end{align}
Bayes' theorem directly implies that $W_{Z|Y}$ can be always chosen as $W_{Z|Y}=P_{Z|Y}$. A third characterization of $P_{XYZ}$ being a Markov chain is that the conditional mutual information vanishes, i.e.,
\begin{align} \label{eq_cl_CMI}
P_{XYZ} \text{ is a Markov chain} \quad \iff \quad I(X:Z|Y)_P=0  \, ,
\end{align}
where 
\begin{align} \label{eq_classicalCMI} 
I(X:Z|Y)_P:=H(XY)_P+H(YZ)_P-H(XYZ)_P-H(Y)_P
\end{align}
denotes the \emph{conditional mutual information}\index{conditional mutual information} and $H(X)_P:=-\sum_{x \in \cX} P_{X}(x) \log P_{X}(x)$ is the \emph{Shannon entropy}\index{entropy!Shannon entropy}. 
\begin{exercise} \label{ex_charMarkov}
Verify the three characterizations~\eqref{eq_cl_def},~\eqref{eq_cl_recovery}, and~\eqref{eq_cl_CMI} for a tripartite distribution $P_{XYZ}$ being a Markov chain.
\end{exercise}

We saw above that~\eqref{eq_cl_def},~\eqref{eq_cl_recovery}, and~\eqref{eq_cl_CMI} are equivalent characterizations for a tripartite distribution to be a Markov chain. The conditional mutual information can be written in terms of a relative entropy, i.e.,
\begin{align} \label{eq_classIneq}
I(X:Z|Y)_P = D(P_{XYZ}\|P_{XY}  P_{Z|Y}) \, ,
\end{align}
where 
\index{entropy!relative entropy}
\begin{align} \label{eq_classicalRelEnt}
D(Q\|R):= \left \lbrace \begin{array}{l l } \sum_{x \in \cX} Q(x) \log \frac{Q(x)}{R(x)} & \quad \text{if } Q \ll R \\
  + \infty & \quad\text{otherwise ,} \end{array} \right.
\end{align}
denotes the \emph{relative entropy} (also known as \emph{Kullback-Leibler divergence})\index{Kullback-Leibler divergence} between two arbitrary probability distributions $Q$ and $R$ on a discrete set $\cX$ and $Q \ll R$ means that $Q$ is absolutely continuous with respect to $R$.
Interestingly, there is an exact correspondence between the conditional mutual information and the relative entropy distance to the set of Markov
chains, also known as a variational formula for the conditional mutual information of the form
\begin{align} \label{eq_ClassicalCMIRelEnt}
I(X:Z|Y)_P = \min_Q \{ D(P_{XYZ}\| Q_{XYZ}) \, : \, Q_{XYZ} \text{ is a Markov chain} \} \, .
\end{align}
A simple calculation reveals that $Q_{XYZ}=P_{XY}P_{Z|Y}$ is the optimizer to~\eqref{eq_ClassicalCMIRelEnt}.
\begin{exercise}\label{ex_varFormulaCMI}
Prove~\eqref{eq_ClassicalCMIRelEnt} and show that the optimizer is always given by $Q_{XYZ}=P_{XY}P_{Z|Y}$. 
\end{exercise}

\subsection{Robustness of classical Markov chains} \label{sec_introRobCl}
Above we have seen three equivalent characterizations~\eqref{eq_cl_def},~\eqref{eq_cl_recovery}, and~\eqref{eq_cl_CMI} for a tripartite distribution $P_{XYZ}$ being a Markov chain. An interesting question is whether these characterizations remain equivalent if they are satisfied approximately. This is indeed the case. To see this, let us recall the variational formula for the mutual information~\eqref{eq_ClassicalCMIRelEnt} which implies that for any distribution $P_{XYZ}$
\begin{align} \label{eq_approxEquiv}
I(X:Z|Y)_P = \eps \quad \iff \quad D(P_{XYZ}||P_{XY} P_{Z|Y}) = \eps \, .
\end{align}
This shows that every distribution $P_{XYZ}$ such that the conditional mutual information is small (but not necessarily vanishing), i.e., $I(X:Z|Y)_P=\eps$, where $\epsilon>0$ is small also approximately satisfies~\eqref{eq_cl_def} and~\eqref{eq_cl_recovery} and vice versa, since by Pinsker's inequality\footnote{Pinsker's inequality states that $\norm{P-Q}_1 \leq \sqrt{2 D(P\|Q)}$ where $\norm{\cdot}_1$ denotes the total variation norm~\cite{csiszarkorner81}.} we have
\begin{align}
\norm{P_{XZ|Y}-P_{X|Y}P_{Z|Y}}_1 = \norm{P_{XYZ}-P_{XY}P_{Z|Y}}_1 \leq \sqrt{2 D(P_{XYZ}\| P_{XY}P_{Z|Y})} \, .
\end{align}

Combining~\eqref{eq_ClassicalCMIRelEnt} with~\eqref{eq_approxEquiv} gives
\begin{align}
I(X:Z|Y)_P=D(P_{XYZ}||P_{XY} P_{Z|Y}) = \min_Q \{ D(P_{XYZ}\| Q_{XYZ}) \, : \, Q_{XYZ} \text{ is a Markov chain} \} \, ,
\end{align}
which shows that distributions with a small conditional mutual information are always close (in terms of the relative entropy distance) to Markov chains and vice versa. 
 As a result, we may define a \emph{(classical) approximate Markov} chain as a tripartite distribution $P_{XYZ}$ such that the conditional mutual information $I(X:Z|Y)_P$ is small.



\section{Quantum Markov chains} \label{sec_introQMC}
So far we considered Markov chains for classical systems that are modeled by random variables. To describe the more general quantum mechanical setup the random variables are replaced by density operators on a Hilbert space.

A tripartite state $\rho_{ABC}$ on {$A\otimes B \otimes C$}, where $A$, $B$, and $C$ denote Hilbert spaces, forms a quantum Markov chain if the $A$ and $C$ part can be viewed independent conditioned on the $B$ part --- for a meaningful notion of conditioning. Generalizing the classical definition of a Markov chain to the quantum mechanical setup turns out to be delicate since conditioning on a quantum system is delicate. Out of the three equivalent characterizations~\eqref{eq_cl_def},~\eqref{eq_cl_recovery}, and~\eqref{eq_cl_CMI} for classical Markov chains we have seen above, it turns out that~\eqref{eq_cl_recovery} servers best for the definition of a quantum Markov chain.\index{Markov chain!quantum}

A tripartite state $\rho_{ABC}$ on $A\otimes B \otimes C$ is called a \emph{(quantum) Markov chain} in order $A\leftrightarrow B\leftrightarrow C$  if there exists a recovery map $\cR_{B \to BC}$ from $B$ to $B\otimes C$ such that 
\begin{align} \label{eq_Markov}
\rho_{ABC} =  (\cI_A \otimes \cR_{B \to BC})(\rho_{AB}) \, ,
\end{align}
where $\cI_A(\cdot)$ denotes the identity map on $A$.
A recovery map is an arbitrary trace-preserving completely positive map. 
The condition~\eqref{eq_Markov} says that the $C$ part can be reconstructed by only acting on the $B$ part.

Petz proved an entropic characterization for the set of quantum Markov chains~\cite{Pet86,Pet03} by showing that
\begin{align} \label{eq_PetzCharacterization}
\rho_{ABC} \text{ is a quantum Markov chain} \quad \iff \quad I(A:C|B)_{\rho}=0 \, ,
\end{align}
where 
\begin{align} \label{eq_QCMI}
I(A:C|B)_{\rho}:=H(AB)_{\rho}+H(BC)_{\rho}-H(ABC)_{\rho}-H(B)_{\rho} 
\end{align}
denotes the \emph{quantum conditional mutual information} and $H(A)_{\rho}:=-\tr\, \rho_A \log \rho_A$ is the \emph{von Neumann entropy}.\index{von Neumann entropy} Furthermore, Petz showed that in case $I(A:C|B)_{\rho}=0$ the recovery map
\begin{align} \label{eq_PetzRecMap}
\cT_{B\to BC} \, : \, X_B \mapsto \rho_{BC}^{\frac{1}{2}}\big( \rho_B^{-\frac{1}{2}} X_B \rho_{B}^{-\frac{1}{2}} \otimes \id_C \big)\rho_{BC}^{\frac{1}{2}}
\end{align} 
always satisfies~\eqref{eq_Markov} (we refer to Theorem~\ref{thm_PetzCMI} for a more precise statement). The recovery map $\cT_{B\to BC}$ is called \emph{Petz recovery map} or \emph{transpose map}.\index{Petz recovery map} It is trace-preserving and completely positive (see Remark~\ref{rmk_PetzTPCP}).

 The result~\eqref{eq_PetzCharacterization} gives an entropic characterization for the set of quantum Markov chains. Furthermore,~\eqref{eq_PetzCharacterization} displays a criterion to verify easily if a certain tripartite state is a quantum Markov chain, as evaluating the conditional mutual information is simple. We further note that the algebraic structure of quantum Markov states has been studied extensively~\cite{HJPW04} (see Theorem~\ref{thm_algebraicQMC} for a precise statement). Quantum Markov chains and their properties are discussed in more detail in Section~\ref{sec_QMC}.

\subsection{Robustness of quantum Markov chains}
A natural question that is relevant for applications is whether the above statements are robust. Specifically, one would like to have a characterization for the set of tripartite states that have a small (but not necessarily vanishing) conditional mutual information, i.e., $I(A : C|B)_{\rho} \leq \eps$ for $\eps > 0$. First results revealed that such states can have a large trace distance to Markov chains that is independent of $\eps$~\cite{CSW12,ILW08} (see Proposition~\ref{prop_notCloseToMarkov} for a precise statement), which has been taken as an indication that their characterization may be difficult.\footnote{As explained in Section~\ref{sec_introRobCl} above, classical tripartite distributions with a small conditional mutual information are always close to classical Markov chains.} This is discussed in more detail in Section~\ref{sec_smallCMInotMarkov}.

As discussed above, states $\rho_{ABC}$ such that $I(A : C|B)_{\rho}$ is small are not necessarily close to any Markov chain, however such states approximately satisfy~\eqref{eq_Markov}. More precisely, it was shown~\cite{FR14,BHOS14,SFR15,wilde15,STH15,JRSWW15,SBT16} that for any state $\rho_{ABC}$ there exists a recovery map $\cR_{B \to BC}$ such that
\begin{align} \label{eq_FRintroo}
I(A:C|B)_{\rho} \geq \MD\big(\rho_{ABC} \| (\cI_A \otimes \cR_{B \to BC})(\rho_{AB}) \big) \, ,
\end{align}
where $\MD$ denotes the measured relative entropy (see Definition~\ref{def_measRelEnt}). The measured relative entropy $\MD(\omega\|\tau)$ is a quantity that determines how close $\omega$ and $\tau$ are. It is nonnegative and vanishes if and only if $\omega = \tau$. The measured relative entropy and its properties are discussed in Section~\ref{sec_measRelEnt}.
We refer to Theorem~\ref{thm_FR} for a more precise statement. Inequality~\eqref{eq_FRintroo} justifies the definition of \emph{approximate quantum Markov chains} as states that have a small conditional mutual information, since according to~\eqref{eq_FRintroo} these states approximately satisfy~\eqref{eq_Markov}. In Section~\ref{sec_AQMC} we discuss in detail the properties of approximate quantum Markov chains.

Unlike in the classical case where the robustness of Markov chains directly follows from~\eqref{eq_ClassicalCMIRelEnt} which is simple to prove (see Exercise~\ref{ex_varFormulaCMI}), Inequality~\eqref{eq_FRintroo} far from trivial. A large part of this book (mainly Chapters~\ref{chapter_nonCommute} and~\ref{chapter_traceIneq}) are dedicated to the task of developing mathematical techniques that can be applied afterwards in Chapter~\ref{chapter_recoverability} to prove~\eqref{eq_FRintroo}.

\section{Outline}
The aim of this book is to introduce its readers to the concept of approximate quantum Markov chains, i.e., a robust version of Markov chains for quantum mechanical systems. 
Our exposition does not assume any prior knowledge about Markov chains nor quantum mechanics. We derive all relevant technical statements from the very beginning such that the reader only needs to be familiar with basic linear algebra, analysis, and probability theory.
We believe that the mathematical techniques described in the book, with an emphasis on their applications to understand the behavior of approximate Markov chains, are of independent interest beyond the scope of this book. 
\bigskip

\noindent The following is a brief summary of the main results obtained in each chapter:

\begin{description}

\item
\textbf{Chapter~\ref{chapter_prelim}}
introduces the mathematical preliminaries that are necessary to follow the book. The advanced reader may easily skip this chapter. We first explain the notation that is summarized in Table~\ref{table_notation} before introducing basic properties of norms (Section~\ref{sec_schattenNorms}), quantum mechanical evolutions (Section~\ref{sec_quantumChannels}), and entropy measures (Section~\ref{sec_entropyMeas}). Section~\ref{sec_posOperators} discusses well-known properties of functions on Hermitian operators. 

\vspace{2mm}

\item
\textbf{Chapter~\ref{chapter_nonCommute}} 
presents two different mathematical techniques that can be used to overcome difficulties arising from the noncommutative nature of linear operators. 
Suppose we are given two operators. Is it possible to modify one of the two operators such that it commutes with the other one without changing it by too much? 

\hspace{4mm}In Section~\ref{sec_pinching} we present a first answer to the above question by introducing the spectral pinching method. For any Hermitian operator $H$ with spectral decomposition $H = \sum_{\lambda} \lambda \Pi_{\lambda}$ we can define the pinching map with respect to $H$ as
\begin{align}
\cP_{H} \, : \, X \mapsto \sum_{\lambda} \Pi_{\lambda} X \Pi_{\lambda} \, .
\end{align}
The pinching map satisfies various nice properties that are summarized in Lemma~\ref{lem_propertiesPinching}. For example, $\cP_{H}(X)$ always commutes with $H$ for any nonnegative operator $X$. Furthermore, there is an operator inequality that relates $\cP_{H}(X)$ with $X$. We demonstrate how to use the spectral pinching method in practice by presenting an intuitive proof for the Golden-Thompson inequality that is only based on properties of pinching maps.

\hspace{4mm}Section~\ref{sec_interpolation} discusses complex interpolation theory which oftentimes can be used as an alternative to the pinching technique. The basic idea is the following: consider an operator-valued holomorphic function defined on the strip $S:=\{ z \in \C : 0 \leq \Real\, z \leq 1\}$.
Complex interpolation theory allows us to control the behavior of the norm of the function at $(0,1)$ by its norm on the boundary, i.e., at $\Real \, z = 0$ and $\Real \, z = 1$. This is made precise in Theorem~\ref{thm_hirschman}, which is the main result of this section.
Interpolation theory is less intuitive than pinching, however can lead to stronger results as we will demonstrate in Chapter~\ref{chapter_traceIneq}.

\vspace{2mm}

\item
\textbf{Chapter~\ref{chapter_traceIneq}} 
shows how to employ the techniques presented in Chapter~\ref{chapter_nonCommute} to prove novel real-valued inequalities involving several linear operators --- so-called \emph{trace inequalities}. Trace inequalities are a powerful tool that oftentimes helps us to understand the behavior of functions of operators. 

\hspace{4mm}Arguably one of the most famous trace inequalities is the Golden-Thompson inequality stating that for any Hermitian operators $H_1$ and $H_2$ we have
\begin{align}\label{eq_GToverview}
\tr \, \ee^{H_1 + H_2} \leq \tr  \, \ee^{H_1} \ee^{H_2} \, .
\end{align}
The main result of this chapter is an extension of~\eqref{eq_GToverview} to arbitrarily many matrices (see Theorem~\ref{thm_GT_steinHirschman}). As we will show, the intuition for this extension can be seen from the pinching method whereas the precise result is proven using interpolation theory, i.e., with the help of Theorem~\ref{thm_hirschman}.

\hspace{4mm}Besides the Golden-Thompson inequality there exists a variety of other interesting trace inequalities. For example the Araki-Lieb-Thiring inequality states that for any nonnegative operators $B_1$, $B_2$, and any $q >0$ we have
\begin{align} \label{eq_ALToverview}
\tr\, \big(B_1^{\frac{r}{2}} B_2^r B_1^{\frac{r}{2}} \big)^{\frac{q}{r}} \leq \tr\, \big(B_1^{\frac{1}{2}} B_2 B_1^{\frac{1}{2}}\big)^q \quad \text{if} \quad r \in (0,1] \, .
\end{align}
In Section~\ref{sec_ALT} we prove an extension of~\eqref{eq_ALToverview} to arbitrarily many matrices (see Theorem~\ref{thm_ALT_Hirschman}).

\hspace{4mm}Finally, we consider a logarithmic trace inequality stating that for any nonnegative operators $B_1$, $B_2$, and any $p >0$ we have
\begin{align}\label{eq_logTrace_overview}
 \frac{1}{p}\tr \, B_1 \log B^{\frac{p}{2}}_2 B_1^p B^{\frac{p}{2}}_2 
 \leq  \tr \, B_1 ( \log B_1 + \log B_2)
\leq \frac{1}{p}\tr \, B_1 \log B_1^{\frac{p}{2}} B_2^p B_1^{\frac{p}{2}} \, .
\end{align}
In Section~\ref{sec_logTrace} we prove an extension of the first inequality of~\eqref{eq_logTrace_overview} to arbitrarily many matrices (see Theorem~\ref{thm_logTraceMulti}).

\vspace{2mm}

\item
\textbf{Chapter~\ref{chapter_recoverability}} 
properly defines the concept of a quantum Markov chain (see Section~\ref{sec_QMC}) as tripartite states $\rho_{ABC}$ such that there exists a recovery map $\cR_{B \to BC}$ from $B$ to $B\otimes C$ that satisfies
\begin{align} \label{eq_Markov_overview}
\rho_{ABC}=(\cI_{A}\otimes \cR_{B \to BC})(\rho_{AB}) \, ,
\end{align}
where $\cI_A$ denotes the identity map on $A$. Alternatively, quantum Markov chains are characterized as states $\rho_{ABC}$ such that the conditional mutual information vanishes, i.e., $I(A:C|B)_{\rho}=0$ (see Theorem~\ref{thm_PetzCMI}).

\hspace{4mm}With the help of the extension of the Golden-Thompson inequality to four matrices (derived in Chapter~\ref{chapter_traceIneq}) we show that for any density operator $\rho_{ABC}$ there exists an explicit recovery map $\cR_{B \to BC}$ that only depends on $\rho_{BC}$ such that
\begin{align} \label{eq_FR_overview}
I(A:C|B)_{\rho} \geq \MD\big(\rho_{ABC} \| (\cI_A \otimes \cR_{B \to BC})(\rho_{AB}) \big) \geq 0 \, .
\end{align}
We refer to Theorem~\ref{thm_FR} for a more precise statement. Inequality~\eqref{eq_FR_overview} shows that states with a small conditional mutual information approximately satisfy the Markov condition~\eqref{eq_Markov_overview}. This therefore justifies the definition of approximate quantum Markov chains as states that have a small conditional mutual information. Proposition~\ref{prop_notCloseToMarkov} shows that approximate quantum Markov chains, however, can be far from any Markov chain (with respect to the trace distance).

\hspace{4mm}Inequality~\eqref{eq_FR_overview} shows that states $\rho_{ABC}$ with a small conditional mutual information can be approximately recovered from $\rho_{AB}$ by only acting on the $B$-system, i.e., a small conditional mutual information is a sufficient condition that a state reconstruction in the sense of~\eqref{eq_Markov_overview} is approximately possible.
Theorem~\ref{thm_CMI_UB} proves an entropic necessary condition involving the conditional mutual information that such an approximate state reconstruction is possible. In particular, we will see that there exist states with a large conditional mutual information such that~\eqref{eq_Markov_overview} still approximately holds. 

\hspace{4mm}Another reason why~\eqref{eq_FR_overview} is interesting is that it strengthens the celebrated strong subadditivity of quantum entropy which ensures that $I(A:C|B)_{\rho}:=H(AB)_{\rho}+H(BC)_{\rho}-H(ABC)_{\rho}-H(B)_{\rho} \geq 0$. This entropy inequality is well-studied and known to be equivalent to various other famous entropy inequalities such as the data processing inequality, concavity of the conditional entropy and joint convexity of the relative entropy.
In Section~\ref{sec_EntropyIneq} we show how~\eqref{eq_FR_overview} can be used to prove strengthenings of the other entropy inequalities. 

\vspace{2mm}

\item
\textbf{Appendix~\ref{cha_constructionExample}} presents an example showing that there exist states $\rho_{ABC}$ with an arbitrarily large quantum conditional mutual information (i.e., $I(A:C|B)_{\rho}$ is large) that, however, can be reconstructed well in the sense that there exits a recovery map $\cR_{B \to BC}$ such that $\rho_{ABC}$ is close to $(\cI_A \otimes \cR_{B \to BC})(\rho_{AB})$.

\vspace{2mm}

\item
\textbf{Appendix~\ref{app_Ex_tight}} discusses examples showing that Theorem~\ref{thm_CMI_UB} is essentially tight and therefore cannot be further improved.

\vspace{2mm}

\item
\textbf{Appendix~\ref{app_solutions}} provides solutions to the exercises stated throughout the book. The exercises are chosen such that they can be solved rather straightforwardly. The main purpose of the exercises is to give the reader a possibility to check if she has understood the presented subject.

\end{description}


\chapter{Preliminaries}
\label{chapter_prelim} 

\abstract{ The formalism of quantum mechanics distinguishes itself from classical physics by being more general. Its main characters are linear operators on a Hilbert space. To understand these objects, this chapter introduces the basic concepts and notation that will be used throughout the book. We further introduce various entropic quantities that are used to describe the behavior of quantum mechanical systems. }
\vspace{8mm}
\noindent Our notation is summarized in Table~\ref{table_notation}. The expert reader may directly proceed to Chapter~\ref{chapter_nonCommute}. 
In this book we restrict ourselves to finite-dimensional Hilbert spaces, even though most of the results covered remain valid for separable Hilbert spaces. As a result, linear operators on these Hilbert spaces can be viewed as matrices. We decided to still call them operators, i.e., for example a positive semidefinite matrix will be called a nonnegative operator on a (finite-dimensional) Hilbert space.
\section{Notation} \label{sec_notation}
\index{notation}
The notational conventions used in this book are summarized in Table~\ref{table_notation}.
To simplify notation we try to avoid brackets whenever possible. For example, $\tr \, A^p$ has to be read as $\tr(A^p)$. We will usually drop identity operators from the notation when they are clear from the context. We would thus write for example $\rho_B \rho_{AB} \rho_{B}$ instead of $(\id_A \otimes \rho_B) \, \rho_{AB} \, (\id_A \otimes \rho_B)$.

A Hermitian operator $H$ is called nonnegative (denoted by $H\geq 0$) if all its eigenvalues are nonnegative. It is called strictly positive (denoted by $H>0$) if all its eigenvalues are strictly positive. We partially order the set of Hermitian operators (\emph{L\"owner ordering}) by defining $H_1 \geq H_2$ to mean $H_1 - H_2 \geq 0$ for two Hermitian operators $H_1$ and $H_2$.

For $f:\R \to \C$ we denote its Fourier tranform by \index{Fourier transform}
\begin{align}
\hat f(\omega) := \int_{-\infty}^{\infty} \di t \, f(t) \ee^{-\ci \omega t}  \, .
\end{align}
We use $\indic\{$statement$\}$ to denote the indicator of the statement, i.e.,
\begin{align}
\indic\{\text{statement}\} = \left \lbrace \begin{array}{l l}
1 & \text{if statement is true}\\
0 & \text{if statement is false} \, .
\end{array} \right.
\end{align}

\begin{longtable}[!htb]{l l}
\caption{Overview of notational conventions}\\
\label{table_notation}
\textbf{General} \\ 
\hline
$\C,\,\R,\,\R_+,\,\N$ & complex, real, nonnegative real, and natural numbers\\
$[n]$ & the set $\{1,2,\ldots,n\}$ for $n\in \N$\\
$\log$ & natural logarithm \\
$\bra{\cdot}$, $\ket{\cdot}$ & bra and ket \\
$A,B,C,\dots$ & Hilbert spaces are denoted by capital letters \\
$\dim(A)$ & dimension of the Hilbert space $A$\\
$\cA,\cB,\cC,\dots$ & mappings are denoted by calligraphic capital letters\\
$\id_A$, $\cI_A$ & identity operator and identity map on $A$\\ 
$\tr$, $\tr_A$ & trace and partial trace \\
$\poly(n)$ &   arbitrary polynomial in $n$\\
$\conv(X)$ & convex hull of the set $X$ \\
$\interior(X)$ & interior of the set $X$ \\
$\partial X$ & boundary of the set $X$\\
$\hat f$ & Fourier transform of $f$\\
$f \star g $ & convolution of $f$ and $g$\\
$\mathrm{tri}_{\kappa}$ & triangular function of width $\kappa$\\
$|X|$ & cardinality of the set $X$ \\
$\indic$ & indicator function\\
\hline \\
\textbf{Operators} \\ 
\hline
$\LL(A), \LL(A,B)$ & set of bounded linear operators on $A$ and from $A$ to $B$\\
$\Her(A)$ & set of Hermitian operators on $A$\\
$\Pos(A)$, $\Poss(A)$ & set of nonnegative and strictly positive operators on $A$ \\ 
$\St(A)$ & set of density operators on $A$  \\
$\U(A)$ & set of unitaries on $A$ \\
$\V(A,B)$ & set of isometries from $A$ to $B$  \\
$\TPCP(A,B)$ & set of trace-preserving completely positive maps from $A$ to $B$ \\
$\QMC(A\otimes B \otimes C)$ & set of (quantum) Markov chains on $A\otimes B \otimes C$\\
$\spec(A)$ & set of distinct singular values of the operator $A$\\
$\supp(A)$ & support of the operator $A$\\
$\rank(A)$ & rank of the operator $A$\\
$ A \ll B$ & support of $A$ is contained in the support of $B$\\
$[A,B]$ & commutator between $A$ and $B$, i.e., $[A,B]:=AB-BA$ \\
$\Delta_{H}$ & spectral gap of the Hermitian operator $H$\\
$|A|$ & modulus of the operator $A$ \\
$A^\dagger$ & conjugate transpose of the operator $A$ \\
$\bar A $ & conjugate of the operator $A$ \\
$A^{\mathrm{T}}$ & transpose of the operator $A$\\
$A \otimes B$ & tensor product between operator $A$ and $B$\\
$A \oplus B $ & direct sum between operator $A$ and $B$\\
\hline \\
\textbf{Distance measures} \\ 
\hline
$\norm{\cdot}_p$ & Schatten $p$-norm \\
$\normU{\cdot}$ & arbitrary unitarily invariant norm \\
$F(\rho,\sigma)$ & fidelity between $\rho$ and $\sigma$\\
$\Delta(\rho,\sigma)$ & trace distance between $\rho$ and $\sigma$\\
\hline \\
\textbf{Entropies} \\ 
\hline
$H(\rho)$ & von Neumann entropy of the density operator $\rho$ \\
$H(A|B)$ & conditional entropy of $A$ given $B$\\
$D(\rho \| \sigma)$ & relative entropy between $\rho$ and $\sigma$\\
$D_{\mathbb{M}}(\rho\| \sigma)$ & measured relative entropy between $\rho$ and $\sigma$\\
$D_{\alpha}(\rho \| \sigma)$ & minimal $\alpha$-R\'enyi relative entropy between $\rho$ and $\sigma$\\
$I(A:C|B)$ & conditional mutual information of $A$ and $C$ given $B$\\
$x \mapsto h(x)$ & binary entropy function \\
\hline \\
\textbf{Abbreviations} \\ 
\hline
POVMs & positive operator valued measures \\
DPI & data processing inequality\\
SSA & strong subadditivity of quantum entropy \\
GT & Golden-Thompson \\
ALT & Araki-Lieb-Thirring\\
\end{longtable} 

\section{Schatten norms} \label{sec_schattenNorms}
To deal with linear operators on a Hilbert space, the concept of a norm is useful.\index{norm}
\begin{svgraybox}
\vspace{-4mm}
\begin{definition}
A \emph{norm} of a linear operator $L\in \LL(A)$ is a map $\norm{\cdot}: \LL(A) \to [0,\infty)$ that satisfies:\vspace{2mm}
\\
\setlength{\tabcolsep}{0.5em}
\begin{tabular}{ll}
1.~Nonnegativity: & $\norm{L}\geq 0$ for all $L\in \LL(A)$ and $\norm{L}=0$ if and only if $L=0$.  \vspace{1mm}\\ \vspace{1mm}
2.~Absolute homogeneity: & $\norm{\alpha L} = |\alpha| \norm{L}$ for all $\alphaÊ\in \C$, $L \in \LL(A)$.\\
3.~Triangle inequality: & $\norm{L_1 + L_2} \leq \norm{L_1} + \norm{L_2}$ for all $L_1,L_2 \in \LL(A)$.
\end{tabular}
\end{definition}
\vspace{-3mm}
\end{svgraybox}
A norm $\normU{\cdot}$ is called \emph{unitarily invariant} if $\normU{U L V^\dagger} = \normU{L}$ for any isometries $U,V \in \V(A,B)$.\index{norm!unitary invariant} In the following we will consider a particular family of unitarily invariant norms the so-called \emph{Schatten $p$-norms}.\index{norm!Schatten norm} The \emph{modulus} of a of a linear operator $L \in \LL(A)$ is the positive semi-definite operator $|L|:=\sqrt{L^\dagger L}$. \index{modulus of an operator}
\begin{svgraybox}
\vspace{-4mm}
\begin{definition}
For any $L \in \LL(A)$ and $p\geq 1$, the \emph{Schatten $p$-norm} is defined as
\begin{align} \label{eq_schatten}
\norm{L}_p:=\left( \tr\, |L|^p \right)^{\frac{1}{p}} \, .
\end{align}
\end{definition}
\vspace{-4mm}
\end{svgraybox}
We extend this definition to all $p > 0$, but note that $\norm{L}_p$ is not a norm for $p \in (0,1)$ since it does not satisfy the triangle inequality.\footnote{For $p\in (0,1)$ the Schatten $p$-norm is thus only a \emph{quasi-norm}.} In the limit $p\to \infty$ we recover the \emph{operator norm} or \emph{spectral norm}, for $p = 1$ we obtain the \emph{trace norm}, and for $p=2$ the \emph{Frobenius} or \emph{Hilbert-Schmidt norm}. Schatten norms are functions of the singular values and thus unitarily invariant. Furthermore, by definition we have
\begin{align}
\norm{L}_p = \norm{L^\dagger}_p \quad \text{and} \quad \norm{L}_{2p}^2 = \norm{L L^\dagger}_p = \norm{L^\dagger L}_p \, .
\end{align}
Schatten norms are ordered in the sense that
\begin{align}
\norm{L}_p \leq \norm{L}_q \quad \text{for} \quad 1\leq q \leq p \, .
\end{align}
Schatten norms are multiplicative under tensor products, i.e., 
\begin{align}
\norm{L_1 \otimes L_2 \otimes \cdots \otimes L_n}_p =\prod_{k=1}^n \norm{L_k}_p \, .
\end{align}
Interestingly, among all possible norms only the Schatten $p$-norms with $p \geq 1$ are unitarily invariant and at the same time multiplicative under tensor products~\cite[Theorem~4.2]{nechita11}.\footnote{These two properties are crucial for the pinching method discussed in Section~\ref{sec_pinching}.}
\begin{exercise} \label{ex_pNorm}
Show that the Schatten $p$-norm defined in~\eqref{eq_schatten} is a norm for $p\geq 1$ and verify that it satisfies the properties mentioned above.
\end{exercise}

Schatten norms can be expressed in terms of a variational formula, i.e., we can write it as the following optimization problem~\cite[Section~IV.2]{bhatia_book}. 
\begin{lemma}[Variational formula Schatten norm] \label{lem_varSchatten}\index{variational formula!Schatten norm}
Let $L\in \LL(A)$ and $p\geq 1$. Then
\begin{align}
\norm{L}_p = \sup_{K\in \LL(A)} \{ |\tr\, L^\dagger K| \, : \, \norm{K}_q =1 \} \quad \textnormal{for} \quad \frac{1}{p} + \frac{1}{q} =1 \, .
\end{align}
\end{lemma}

Schatten norms are submultiplicative, i.e., for $L_1,L_2 \in \LL(A)$ we have
\begin{align}
\norm{L_1 L_2}_p \leq \norm{L_1}_p \norm{L_2}_p  \quad \textnormal{for all } p \geq 1\, .
\end{align}
A stronger result is obtained by the \emph{generalized H\"older inequality} for Schatten (quasi) norms~\cite[Exercise~IV.2.7]{bhatia_book} (see~\cite[Section~3.3]{raban_MA} for a precise proof).
\begin{svgraybox}
\vspace{-4mm}
\begin{proposition}[H\"older's inequality] \label{prop_Holder} \index{H\"older inequality}
Let $n \in \N$, $p,p_1,\dots,p_n \in \R_+$ and a finite sequence $(L_k)_{k=1}^n$ of linear operators. Then
\begin{align}
\norm{\prod_{k=1}^n L_k}_p \leq \prod_{k=1}^n \norm{L_k}_{p_k}  \quad \text{for} \quad \sum_{k=1}^n \frac{1}{p_k} = \frac{1}{p} \, .
\end{align}
\end{proposition}
\vspace{-4mm}
\end{svgraybox}

The function $L \mapsto \norm{L}_p$ for $p \geq 1$ is convex as the Schatten $p$-norm satisfies the triangle inequality.
This means that for any probability measure $\mu$ on a measurable space $(X,\Sigma)$ and a sequence $(L_x)_{x \in X}$ of linear operators, we have
\begin{align} \label{eq_convNorm}
  \norm{\int_X \mu(\di x) \, L_x}_p \leq 
   \int_X \mu(\di x) \norm{L_x}_p \quad \text{for} \quad p\geq 1 \, .
\end{align}
Quasi-norms with $p \in (0,1)$ are no longer convex.
However, we show that these quasi-norms still satisfy an asymptotic convexity property for tensor products of operators in the following sense~\cite{SBT16}.
\begin{lemma} \label{lem_normConvexNEW}
Let $p \in (0,1)$, $\mu$ be a probability measure on $(X,\Sigma)$ and consider a sequence $(B_x)_{x \in X}$ of nonnegative operators. Then
\begin{align} 
\frac{1}{m} \log \norm{\int_X \mu(\di x) \, B_x^{\otimes m}}_p \leq \frac{1}{m} \log \int_X \mu(\di x) \norm{B_x^{\otimes m}}_p + \frac{\log \poly (m)}{m}   \, .
\end{align}
\end{lemma}
\begin{proof}
Let $A$ denote the Hilbert space of dimension $d$ where the nonnegative operators $B_x$ act on.
For any $x \in X$, consider the spectral decomposition $B_x = \sum_{k} \lambda_k \ket{k}\!\bra{k}$. 
Let $\ket{v_x} = \sum_k \sqrt{\lambda_k} \ket{k} \otimes \ket{k} \in A \otimes A'$ be a purification of $B_x$, i.e., $\tr_{A'} \proj{v_x} = B_x$.
Now note that the projectors $( \ket{v_x}\bra{v_x} )^{\otimes m}$ lie in the symmetric subspace of $(A \otimes A')^{\otimes m}$ whose dimension grows as $\poly(m)$.\footnote{This follows from the fact that the dimension of the symmetric subspace of $A^{\otimes m}$ is equal to the number of types of sequences of $d$ symbols of length $m$, which is polynomial in $m$ (as shown in~\eqref{eq_types}).} 
Moreover, we have
\begin{align}
\int_X \mu(\di x) B_k^{\otimes m} =  \int_X \mu(\di x) \, \tr_{A^{'\otimes m}} \left(  \ket{v_x}\bra{v_x} \right)^{\otimes m} \,.
\end{align}
Carath\'eodory's theorem (see, e.g.,~\cite[Theorem 18]{eggleston_book}) ensures the existence of a discrete probability measure~$P$ on $I \subset X$ with $|I| = \poly(m)$ such that
\begin{align}
\int_X \mu(\di x) B_x^{\otimes m} = \sum_{x \in I} P(x) B_x^{\otimes m} 
\quad \textrm{and} \quad
\int_X \mu(\di x) \norm{B_x^{\otimes m}}_p = \sum_{x \in I} P(x) \norm{B_x^{\otimes m}}_p \,.\label{eq_sd1}
\end{align}
We thus have
\begin{align} \label{eq_stepConc1}
\frac{1}{m} \log \norm{\int_X \mu(\di x) B_x^{\otimes m}}_p 
&=\frac{1}{m} \log \norm{\sum_{x \in I} P(x)B_x^{\otimes m}}_p \, .
\end{align}

For $p\in (0,1)$ the Schatten $p$-norms only satisfy a weakened version of the triangle inequality (see, e.g.,~\cite[Equation~20]{Kittaneh97}) of the form
\begin{align}
  \norm{\sum_{x=1}^n B_x }_p^p \leq \sum_{x=1}^n \norm{ B_x }_p^p \,.
\end{align}
Combining this with~\eqref{eq_stepConc1} gives
\begin{align}
\frac{1}{m} \log \norm{\int_X \mu(\di x) B_x^{\otimes m}}_p 
&\leq \frac{1}{m} \log \left( \sum_{x \in I} \norm{P(x) B_x^{\otimes m}}_p^p \right)^{\frac{1}{p}} \\
&=\frac{1}{m} \log \left( |I|^\frac{1}{p} \Big( \frac{1}{|I|} \sum_{x \in I} \norm{P(x) B_x^{\otimes m}}_p^p \Big)^{\frac{1}{p}} \right)\, .
\end{align}
As the map $t \mapsto t^{\frac{1}{p}}$ is convex for $p\in (0,1)$ (see Table~\ref{table_opConvex}) we obtain
\begin{align}
\frac{1}{m} \log \norm{\int_X \mu(\di x) B_x^{\otimes m}}_p 
&\leq  \frac{1}{m} \log\left(|I|^{\frac{1}{p}-1} \sum_{x\in I} \norm{P(x) B_x^{\otimes m}}_p \right) \\
&= \frac{1}{m} \log \left( \sum_{x \in I} P(x) \norm{B_x^{\otimes m}}_p \right) + \frac{1}{m} \frac{1-p}{p} \log| I| \\
&= \frac{1}{m} \log \left( \int_X \mu(\di x) \norm{B_x^{\otimes m}}_p \right) + \frac{\log \poly(m)}{m} \ ,
\end{align}
where the final step uses that $|I| = \poly(m)$. \qed
\end{proof}

Combining Lemma~\ref{lem_normConvexNEW} with~\eqref{eq_convNorm} shows that for all $p > 0$ we have the following quasi-convexity property
\begin{align}
\frac{1}{m} \log \norm{\int_X \mu(\di x) \, B_x^{\otimes m}}_p \leq \log \sup_{x \in X} \norm{B_x}_p + \frac{\log \poly (m)}{m} \,.  \label{eq_TCpropnew}
\end{align}
Lemma~\ref{lem_normConvexNEW} will be particularly useful in combination with the pinching technique presented in Section~\ref{sec_pinching}.

\section{Functions on Hermitian operators} \label{sec_posOperators}
The set of Hermitian operators is equipped with a natural partial order, i.e., a consistent way of saying that one operator is larger than another one or that two operators are actually incomparable. For $H_1, H_2 \in \Her(A)$ we say $H_1$ is larger than $H_2$, denoted by $H_1 \geq H_2$ if and only if $H_1 - H_2$ is nonnegative, i.e., $H_1 - H_2 \geq 0$, or equivalently $H_1-H_2 \in \Pos(A)$.
This defines a partial order (called \emph{L\"ownerÕs partial order}) in the sense that two Hermitian operators may be incomparable.\index{L\"owner's partial order}

For every Hermitian operator $H \in \Her(A)$ we can write down its spectral decomposition, i.e.,\index{spectral decomposition}
\begin{align}
H = \sum_{\lambda \in \spec(H)} \lambda \Pi_{\lambda} \, ,
\end{align}
where $\Pi_{\lambda}$ denotes the projector onto the eigenspace of $\lambda$. For any continuous function $f:\R \to \R$ we define the operator $f(H) \in \Her(A)$ as
\begin{align}
f(H) := \sum_{\lambda \in \spec(H)} f(\lambda) \Pi_{\lambda} \, .
\end{align}
By definition we thus have $f(U H U^\dagger) = U f(H) U^\dagger$ for any unitary $U\in \U(A)$. If we consider a function $f:\R_+\to \R$, its operator-valued version maps nonnegative operators to Hermitian operators.
\begin{svgraybox}
\vspace{-4mm}
\begin{definition}
Let $\mathbb{I} \subseteq \R$. A function $f:\mathbb{I} \to \R$ is called \emph{operator monotone} if
\begin{align}
H_1 \leq H_2 \,\, \implies \,\, f(H_1) \leq f(H_2) \, ,
\end{align}
for all  $H_1,H_2 \in \Her(A)$ with $\spec(H_1),\spec(H_2) \in \mathbb{I}^{|\spec(H_k)|}$.
The function $f$ is \emph{operator anti-monotone} if $-f$ is operator monotone. 
\end{definition}
\vspace{-4mm}
\end{svgraybox}
\begin{svgraybox}
\vspace{-4mm}
\begin{definition}
Let $\mathbb{I} \subseteq \R$. A function $f:\mathbb{I} \to \R$ is called \emph{operator convex} if
\begin{align}
f(t H_1 + (1-t)H_2) \leq t f(H_1) + (1-t) f(H_2)\, ,
\end{align}
for all $t\in [0,1]$ and for all $H_1,H_2 \in \Her(A)$ with $\spec(H_1),\spec(H_2) \in \mathbb{I}^{|\spec(H_k)|}$.
The function $f$ is \emph{operator concave} if $-f$ is operator convex. 
\end{definition}
\vspace{-4mm}
\end{svgraybox}
A two-parameter function is called \emph{jointly convex} (\emph{jointly concave}) if it is convex (concave) when taking convex combinations of the input tuples. For many functions it has been determined if they are operator convex or operator monotone. Table~\ref{table_opConvex} summarizes a few prominent examples. 
\setlength{\tabcolsep}{0.80em}
\begin{table}[!htb]
\caption{Examples of operator monotone, convex and concave functions.}
 \label{table_opConvex}
\begin{tabular}{l c c c c c}
\textbf{Function} & \textbf{Domain} & \textbf{Op.~monotone} & \textbf{Op.~anti-monotone} & \textbf{Op.~convex} & \textbf{Op.~concave}  \\
\hline
$t \mapsto t^\alpha$ & $(0,\infty)$ & $\alpha \in [0,1]$ & $\alpha \in [-1,0]$ & $\alpha \in [-1,0) \cup [1,2]$  & $\alpha \in (0,1]$ \\
$t\mapsto \log \, t$ & $(0,\infty)$ &\ding{51}  & \ding{55} &\ding{55} & \ding{51} \\
$t\mapsto t\log t$ & $[0,\infty)$ & \ding{55}  & \ding{55}  & \ding{51} & \ding{55}  \\ 
$t\mapsto \ee^{t}$ & $\mathbb{I} \subseteq \R$ & \ding{55}  & \ding{55}  & \ding{55}  & \ding{55} \\
\hline 
\end{tabular}
\phantom{d}\hspace{1.5mm}Note that $t \mapsto t^\alpha$ is neither operator monotone, convex, nor concave for $\alpha <-1$ and $\alpha >2$.
\end{table}

The following two propositions which can be found in~\cite[Theorem~V.2.5]{bhatia_book} and~\cite[Theorem~2.10]{carlen_book} summarize some generic facts about the convexity and monotonicity of certain functions on Hermitian operators.

\begin{proposition}
Let $f:\R_+ \to \R_+$ be continuous. Then, $f$ is operator monotone if and only if it is operator concave.
\end{proposition}

\begin{proposition}[Convexity and monotonicity of trace functions] \label{prop_traceFunctions}~\index{trace functions}
Let $f:\R \to \R$ be continuous. If $t \mapsto f(t)$ is monotone, so is $\Her(A) \ni H \mapsto \tr f(H)$. Likewise, if $t \mapsto f(t)$ is (strictly) convex, so is $\Her(A) \ni H \mapsto \tr f(H)$.
\end{proposition}

To show that a certain function is operator convex can be difficult and sometimes leads to deep and powerful results. We next discuss two such statements.
\begin{svgraybox}
\vspace{-4mm}
\begin{theorem}[Peierls-Bogoliubov] \label{thm_Peierls} \index{Peierls-Bogoliubov inequality}
The map
\begin{align}
\Her(A) \ni H \mapsto  \log \tr\, \ee^H 
\end{align}
is convex.
\end{theorem}
\vspace{-4mm}
\end{svgraybox}
\begin{proof}
The variational formula for the relative entropy (see~\eqref{eq_step1}) shows that for $t\in [0,1]$ and $H_1,H_2 \in \Her(A)$ we have
\begin{align}
\log \tr\, \ee^{t H_1 + (1-t) H_2} 
&= \max_{\rho \in \St(A)} \{ \tr \big(t H_1 + (1-t) H_2\big) \rho  - D(\rho\| \id_A) \} \\
&\leq t \max_{\rho \in \St(A)} \{ \tr\, H_1 \rho  - D(\rho\| \id_A) \} + (1-t)\max_{\rho \in \St(A)} \{ \tr\, H_2 \rho  - D(\rho\| \id_A) \} \\
&=t \log \tr\, \ee^{ H_1} + (1-t)   \log \tr\, \ee^{ H_2} \, ,
\end{align}
where the final step uses~\eqref{eq_step1}. \qed
\end{proof}

For $H_1, H_2 \in \Her(A)$ Theorem~\ref{thm_Peierls} implies that the function
\begin{align}
(0,1] \ni t \mapsto f(t)= \log \tr \, \ee^{H_1 + t H_2}
\end{align}
is convex and hence
\begin{align}
f(1) - f(0) \geq \frac{f(t) - f(0)}{t} \, .
\end{align}
Taking the limit $t \to 0$ gives the following result which is called~\emph{Peierls-Bogoliubov inequality} in the literature.
\begin{corollary} \label{cor_Peierls}
Let $H_1, H_2 \in \Her(A)$. Then
\begin{align}
\log \frac{\tr \,\ee^{H_1 + H_2}}{\tr\, \ee^{H_1}} \geq \frac{\tr\, H_2 \ee^{H_1}}{\tr\, \ee^{H_1}} \, .
\end{align}
\end{corollary}
The next result is a concavity theorem~\cite{Lieb73}. As we will see later this result is deeply connected with Lieb's triple operator inequality that is discussed in Theorem~\ref{thm_Lieb3} in Chapter~\ref{chapter_traceIneq}.
\begin{svgraybox}
\vspace{-4mm}
\begin{theorem}[Lieb's concavity theorem] \index{Lieb's theorem} \label{thm_liebConc}
Let $H \in \Her(A)$. The map
\begin{align}
\Poss(A) \ni  B \mapsto \tr \, \ee^{H + \log B}
\end{align}
is concave.
\end{theorem}
\vspace{-4mm}
\end{svgraybox}
\begin{proof}
The variational formula for the relative entropy (see~\eqref{eq_varTropp}) shows that for $t\in [0,1]$ and $B_1,B_2 \in \Poss(A)$ we have
\begin{align}
 &\tr \, \ee^{H + \log (t B_1 + (1-t) B_2)} \nonumber\\
 &\hspace{3mm}= \max_{\omega \in \Poss(A)} \{ \tr\, \omega H - D\big(\omega \| t B_1 +(1-t) B_2 \big) + \tr\,\omega \} \\
 &\hspace{3mm} \geq t \max_{\omega \in \Poss(A)} \{ \tr\, \omega H - D(\omega \|  B_1) + \tr\,\omega \} + (1-t) \max_{\omega \in \Poss(A)} \{ \tr\, \omega H - D(\omega \| B_2) + \tr\,\omega \}  \\
 &\hspace{3mm}= t \, \tr\,\ee^{H + \log B_1} + (1-t)  \tr\,\ee^{H + \log B_2} \, ,
\end{align}
where penultimate step uses the joint convexity property of the relative entropy (see Proposition~\ref{prop_relEntProperties}). The final step follows from~\eqref{eq_varTropp}. \qed
\end{proof}

Another celebrated inequality for differentiable functions on nonnegative operators is due to Klein.
\begin{svgraybox}
\vspace{-4mm}
\begin{theorem}[Klein's inequality] \label{thm_kleinsInequality}~\index{Klein's inequality}
Let $B_1,B_2 \in \Pos(A)$ and $f:(0,\infty) \to \R$ be differentiable and convex. Then
\begin{align}
 \tr \, f(B_1) - \tr \, f(B_2) \geq \tr \, (B_1-B_2)f'(B_2) \, .
\end{align}
If $f$ is strictly convex, there is equality if and only if $B_1 = B_2$.
\end{theorem}
\vspace{-4mm}
\end{svgraybox}
\begin{proof}
Define the function $(0,1] \ni t \mapsto g(t)=\tr\, f(t A_1 +(1-t) A_2)$ which according to Proposition~\ref{prop_traceFunctions} is convex. This implies that
\begin{align}
g(1) - g(0) \geq \frac{g(t)-g(0)}{t}  \, .
\end{align}
Taking the limit $t \to 0$ shows that
\begin{align}
 \tr f(A_1) - \tr f(A_2)  \geq  \lim_{t \to \infty} \frac{g(t)-g(0)}{t} = \frac{\di }{\di t} g(t) |_{t=0} = \tr (A_1-A_2)f'(A_2) \, .
\end{align}
\qed
\end{proof}

We close the discussion about functions on Hermitian operators by discussing an operator version of Jensen's inequality~\cite{hansen82}.
\begin{svgraybox}
\vspace{-4mm}
\begin{theorem}[Jensen's operator inequality] \label{thm_Jensen} \index{Jensen's operator inequality} \index{Jensen's operator inequality}
Let $\mathbb{I} \subseteq \R$ and $f:\mathbb{I} \to \R$ be continuous. Then, the following are equivalent
\begin{enumerate}
\item $f$ is operator convex. \vspace{1mm}
\item For every $n\in \N$ we have
\begin{align}
f\left( \sum_{k=1}^n L_k H_k L^\dagger_k \right) \leq \sum_{k=1}^n L_k f(H_k) L^\dagger_k \, ,
\end{align}
for all $H_k \in \Her(A)$ with $\spec(H_k) \in \mathbb{I}$ and all $L_k \in \LL(A,B)$ such that $\sum_{k=1}^n L_k L^\dagger_k =\id_A$.  \vspace{1mm}
\item $f(V H V^\dagger) \leq V f(H) V^\dagger$ for all $V \in \V(A,B)$, $H \in \Her(A)$ such that $\spec(H) \in \mathbb{I}$.  \vspace{1mm}
\item $\Pi f(\Pi H \Pi + t (1-\Pi)) \Pi \leq \Pi f(H) \Pi$ for all projectors $\Pi$ onto $A$, $t \in \mathbb{I}$, $H \in \Her(A)$ such that $\spec(H) \in \mathbb{I}$.
\end{enumerate}
\end{theorem}
\vspace{-4mm}
\end{svgraybox}
\section{Quantum channels} \label{sec_quantumChannels}
In this section we discuss how to model time evolutions of quantum mechanical systems. One postulate of quantum mechanics\footnote{The interested reader can find more information about these postulates in~\cite[Section 2]{nielsenChuang_book}} states that any \emph{isolated} evolution of a subsystem of a composite system over a fixed time interval $[t_0,t_1]$ corresponds to a unitary operator on the state space of the subsystem. For a composite system with state space $A\otimes B$ and isolated evolutions on both subsystems described by $U_A \in \U(A)$ and $U_B \in \U(B)$, respectively, any state $\rho_{AB} \in \St(A\otimes B)$ at time $t_0$ is transformed into the state
\begin{align}
\rho'_{AB} = (U_A \otimes U_B) \rho_{AB} (U_A^\dagger \otimes U_B^\dagger) 
\end{align}
at time $t_1$. Since unitaries are reversible we see that isolated evolutions are reversible, too. \index{physical evolution}

It is helpful to describe the behavior of subsystems in the general case where there is interaction between $A$ and $B$. Such evolutions are no longer isolated and are irreversible. We note that it is always possible to embed the irreversible evolution into a larger system such that it becomes reversible. For the moment we will, however, not follow this viewpoint and rather discuss the mathematical framework to describe general physical evolutions.
There are two equivalent ways to describe the evolution of a quantum mechanical system, called \emph{Schr\"odinger} and \emph{Heisenberg picture}. We will mainly work in the Schr\"odinger picture, the interested reader may consider~\cite{Wolf_skript} for more information about the Heisenberg picture.

A map $\cE: \LL(A) \to \LL(B)$ describes a physical evolution in a meaningful way if it is \emph{linear}, \emph{trace-preserving}, and \emph{completely positive}. Such maps are called \emph{quantum channels} and describe in a most general way a physical evolution.\index{quantum channel} The set of quantum channels from $A$ to $B$, i.e., trace-preserving completely positive maps from $A$ to $B$, is denoted by $\TPCP(A,B)$.

\begin{svgraybox}
\vspace{-4mm}
\begin{definition} \label{def_TP} \index{trace-preserving map}
 A linear map $\cE: \LL(A) \to \LL(B)$ is called \emph{trace-preserving} if $\tr\, \cE(\omega) = \tr\, \omega$ for all $\omega \in \LL(A)$.
\end{definition}
\vspace{-4mm}
\end{svgraybox}

\begin{svgraybox}
\vspace{-4mm}
 \begin{definition}\label{def_CP} \index{completely positive map} \index{positive map}
 A linear map $\cE: \LL(A) \to \LL(B)$ is called \emph{positive} if $\cE(\omega) \in \Pos(B)$ for all $\omega \in \Pos(A)$. The map $\cE$ is called \emph{completely positive} if for any Hilbert space $R$ the map $\cE \otimes \cI_R$ is positive.
 \end{definition}
\vspace{-4mm}
\end{svgraybox}
\begin{exercise} \label{ex_CPvsP}
Construct a linear map $\cE:\LL(A) \to \LL(B)$ that is positive but not completely positive.
\end{exercise}

There exist different representations of trace-preserving completely positive maps. We briefly discuss the three most common ones: the \emph{Choi-Jamiolkowski representation}~\cite{choi82,jami72}, the \emph{Stinespring dilation}~\cite{stinespring55}, and the \emph{operator-sum representation} (also known as \emph{Kraus representation})~\cite{kraus1983}.

For any linear map $\cE:\LL(A) \to \LL(B)$ the corresponding \emph{Jamiolkowski state} is defined by
\begin{align} \index{Jamiolkowski state}
\tau_{\cE}:= (\cE \otimes \cI_{A'})(\proj{\Omega}_{AA'}) \, ,
\end{align}
where 
\begin{align}
\ket{\Omega}_{AA'}:= \frac{1}{\sqrt{\dim(A)}} \sum_{k=1}^{\dim(A)} \ket{kk}_{AA'}
\end{align}
denotes a maximally entangled state. The Jamiolkowski state fully characterizes the map $\cE$.
\begin{svgraybox}
\vspace{-4mm}
\begin{proposition}[Choi-Jamiolkowski representation] \index{Choi-Jamiolkowski representation}
The following provides a one-to-one correspondence between linear maps $\cE:\LL(A) \to \LL(B)$ and operators $\tau \in \LL(B\otimes A)$
\begin{align} \label{eq_Choi}
\tau_{\cE}= (\cE \otimes \cI_{A'})(\proj{\Omega}_{AA'}),  \qquad \tr \,\omega \cE(\sigma) = \dim(A) \tr \, \tau_{\cE}\, \omega\otimes \sigma^{\mathrm{T}} \, ,
\end{align}
for all $\omega \in \LL(B)$, $\sigma \in \LL(A)$ and where the transpose is taken with regards to the Schmidt basis of $\Omega$. The mappings $\cE \mapsto \tau_{\cE}$ and $\tau_{\cE} \mapsto \cE$ defined by~\eqref{eq_Choi} are mutual inverses. 
\end{proposition}
\vspace{-4mm}
\end{svgraybox}
The Jamiolkowski state has a few nice properties. For example it allows us to easily verify if a linear map is trace-preserving and completely positive, since
\begin{align}
\cE \text{ is trace-preserving} \quad \iff \quad \tr_B \, \tau_{\cE} = \frac{\id_A}{\dim(A)} \, ,
\end{align}
and 
\begin{align}
\cE \text{ is completely positive} \quad \iff \quad \tau_{\eps} \in \Pos(B\otimes A) \, .
\end{align}
We can express the map $\cE$ in terms of its Jamiolkowski state as
\begin{align}
\cE \,: \, X \mapsto \dim(A) \tr_A \, \tau_{\cE} (\id_B \otimes X^{\mathrm{T}}) \, .
\end{align}

Another representation of quantum channels shows that they can be viewed as unitary evolutions by enlarging our space.
\begin{svgraybox}
\vspace{-4mm}
\begin{proposition}[Stinespring dilation] \index{Stinespring dilation} \label{prop_stinespring}
Let $\cE:\LL(A) \to \LL(B)$ be linear and completely positive. Then there exists an isometry $V \in \V(A,B \otimes R)$ such that
\begin{align}
\cE \, : \, X \mapsto \tr_R \, V X V^{\dagger} \, .
\end{align}
\end{proposition}
\vspace{-4mm}
\end{svgraybox}
This shows that any possible quantum channel corresponds to a unitary evolution of a larger system.

We finally discuss another representation that shows that a channel can be characterized by a sequence of operators.
\begin{svgraybox}
\vspace{-4mm}
\begin{proposition}[Operator-sum representation] \index{Operator-sum representation} \index{Kraus representation} \label{prop_operatorSum}
Let $\cE:\LL(A) \to \LL(B)$ be linear and completely positive. Then, there exists $r\leq \dim(A) \dim(B)$ and a finite sequence $(E_k)_{k\in[r]}$ of operators $E_k \in \LL(A,B)$ such that
\begin{align}
\cE \, : \, X \mapsto \sum_{k=1}^r E_k X E_k^{\dagger} \,.
\end{align}
The mapping $\cE$ is trace-preserving if and only if $\sum_{k=1}^r E^\dagger_k E_k = \id_A$.
\end{proposition}
\vspace{-4mm}
\end{svgraybox}
We note that $r=\rank(\tau_{\cE})$, where $\tau_{\cE}$ is the Jamiolkowski state of $\cE$, is the \emph{Kraus rank}. The operators $E_k$ are sometimes called \emph{Kraus operators}.

\begin{exercise} \label{ex_Kraus}
Is the finite sequence $(E_k)_{k\in[r]}$ of Kraus operators uniquely determined by $\cE$?
\end{exercise}

\section{Entropy measures} \label{sec_entropyMeas}
Entropy measures are indispensable tools in classical and quantum information theory. They characterize ultimate limits of various operational tasks such as data compression or channel coding~\cite{shannon48,schumacher95}. In this book, we mainly use entropy measures as mathematical objects whose properties are well studied~\cite{petz_statbook,cover,marco_book}. We will not discuss the operational relevance of these measures. The interested reader may consider~\cite{marco_book,wilde_book,holevo_book} for more information.

We next define the entropic quantities that are relevant for this book.  
For a density operator $\rho_A \in \St(A) = \{ X \in \Pos(A) : \tr\, X = 1 \}$ the \emph{von Neumann entropy} is defined as\index{entropy!von Neumann entropy}
\begin{align}
H(A)_{\rho}=H(\rho_A):=-\tr \, \rho_A \log \rho_A \, . 
\end{align}
For a bipartite density operator $\rho_{AB} \in \St(A\otimes B)$ the \emph{conditional entropy} of $A$ given $B$ is\index{entropy!conditional entropy}
\begin{align}
H(A|B)_{\rho}:=H(AB)_{\rho} - H(B)_{\rho} \, .
\end{align}
Finally, for a tripartite density operator $\rho_{ABC} \in \St(A\otimes B \otimes C)$ we define the \emph{conditional mutual information} between $A$ and $C$ given $B$ as\index{conditional mutual information}
\begin{align}
I(A:C|B)_{\rho}:=H(AB)_{\rho}+H(BC)_{\rho}-H(ABC)_{\rho}-H(B)_{\rho} \, .
\end{align}
All these entropy measures can be expressed in terms of the \emph{relative entropy}. Before defining the relative entropy we discuss another measure called \emph{fidelity} that can be used to determine how close two nonnegative operators are.
\subsection{Fidelity}
The fidelity is measure of distance between two nonnegative operators that is ubiquitous in quantum information theory. Oftentimes it is defined for density operators only, however here we define it for general nonnegative operators and discuss certain properties.
\begin{svgraybox}
\vspace{-4mm}
\begin{definition}
For $\rho, \sigma \in \Pos(A)$ the \emph{fidelity} between $\rho$ and $\sigma$ is defined by~\index{fidelity}
\begin{align} \label{eq_fidelity}
F(\rho,\sigma):=\norm{\sqrt{\rho} \sqrt{\sigma}}^2_1 \, .
\end{align}
\end{definition}
\vspace{-4mm}
\end{svgraybox}
The fidelity has various different characterizations.\footnote{We would like to draw the readers attention to the fact that in certain textbooks the fidelity is defined without the square.} One that is particularly useful is due to Uhlmann and relates the fidelity to the notion of purifications~\cite{Uhl76}.
\begin{svgraybox}
\vspace{-4mm}
\begin{theorem}[Uhlmann] \label{thm_Uhlmann} \index{Uhlmann's theorem}
Let $\rho_{AR}=\proj{\psi}_{AR}$ and $\sigma_{AR}=\proj{\phi}_{AR}$ be purifications of $\rho_A \in \Pos(A)$ and $\sigma_A \in \Pos(A)$, respectively. Then
\begin{align}
F(\rho_A , \sigma_A)= \sup_{U_R \in \U(R)} \left| \bra{\psi} (\id_A \otimes U_R ) \ket{\phi} \right|^2 \, .
\end{align}
\end{theorem}
\vspace{-4mm}
\end{svgraybox}

Another characterization of the fidelity is due to Alberti~\cite{Alberti1983}.
\begin{svgraybox}
\vspace{-4mm}
\begin{theorem}[Alberti] \label{thm_Alberti} \index{Alberti's theorem}
Let $\rho, \sigma \in \Pos(A)$. Then
\begin{align}
F(\rho,\sigma)=\inf_{\omega \in \Poss(A)}  (\tr\, \rho \omega) (\tr\, \sigma \omega^{-1}) \, .
\end{align}
\end{theorem}
\vspace{-4mm}
\end{svgraybox}

One reason the fidelity plays an important role in quantum information theory is due to the fact that it has nice properties. In the following we list some of them.
\begin{proposition} \label{prop_propFid}
The fidelity defined in~\eqref{eq_fidelity} satisfies:\vspace{2mm}\\
\setlength{\tabcolsep}{0.3em}
\begin{tabular}{ll}
1.~Multiplicativity & $F(\rho_1 \otimes \rho_2 , \sigma_1 \otimes \sigma_2) = F(\rho_1,\sigma_1) F(\rho_2,\sigma_2)$ for all $\rho_1,\rho_2,\sigma_1,\sigma_2 \in \Pos(A)$.\vspace{2mm}\\ 
2.~Nonnegativity & $F(\rho,\sigma) \in [0,1]$ for all $\rho,\sigma \in \St(A)$. Moreover $F(\rho,\sigma)=1$ if and only \\
&  if $\rho=\sigma$, and $F(\rho,\sigma)=0$ if and only if $\rho\sigma=0$.\vspace{2mm} \\
3.~Isometric invariance& $F(V\rho V^\dagger, V \sigma V^\dagger)=F(\rho,\sigma)$ for all $V\in \V(A,B)$, $\rho,\sigma \in \Pos(A)$ \vspace{2mm} \\
4.~DPI & $F(\rho,\sigma) \leq F\big(\cE(\rho),\cE(\sigma)\big)$ for all $\rho,\sigma\! \in \!\Pos(A)$ and all $\cE \!\in\! \TPCP(A,B)$. \vspace{2mm} \\
5.~Joint concavity & $(\rho,\sigma) \mapsto F(\rho,\sigma)$ is jointly concave on $\Pos(A) \times \Pos(A)$.  \vspace{2mm}  \\
6.~Orthogonal states:& $F(t\rho_1 +(1-t) \rho_2, t \sigma_1 +(1-t) \sigma_2) = t F(\rho_1, \sigma_1) + (1-t)F(\rho_2,\sigma_2)$  \\
&  for $t \in [0,1]$, $\rho_1 \in \St(A)$, $\rho_2 \in \St(B)$, $\sigma_1 \in \Pos(A)$, $\sigma_2 \in \Pos(B)$ such that \\
&  both $\rho_1$ and $\sigma_1$ are orthogonal to both $\rho_2$ and $\sigma_2$.\
\end{tabular}
\end{proposition}
\begin{proof}
The multiplicativity property follows from the fact that Schatten norms are multiplicative under the tensor product
\begin{align}
F(\rho_1 \otimes \rho_2 , \sigma_1 \otimes \sigma_2)
&= \norm{\sqrt{\rho_1 \otimes \rho_2} \sqrt{\sigma_1 \otimes \sigma_2}}^2_1
= \norm{\sqrt{\rho_1}\sqrt{\sigma_1} \otimes \sqrt{\rho_2}\sqrt{\sigma_2}}^2_1\\
&= \norm{\sqrt{\rho_1}\sqrt{\sigma_1} }^2_1\norm{\sqrt{\rho_2}\sqrt{\sigma_2} }^2_1
=F(\rho_1,\sigma_1) F(\rho_2,\sigma_2) \, .
\end{align}
The nonnegativity follows directly from Uhlmann's theorem. By defintion we see that $F(\rho,\sigma)=0$ if and only if $\sqrt{\rho} \sqrt{\sigma}=0$ which is equivalent to $\rho \sigma =0$.
Since Schatten norms are unitarily invariant we find
\begin{align}
F(V \rho V^\dagger, V \sigma V^\dagger) 
=\norm{\sqrt{V \rho V^\dagger} \sqrt{V \sigma V^\dagger}}^2_1
= \norm{V \sqrt{\rho} V^\dagger V \sqrt{\sigma} V^\dagger}^2_1
=F(\rho,\sigma) \, ,
\end{align}
which proves that the fidelity is isometric invariant.

We first show that data-processing inequality for the partial trace, i.e., we show that
\begin{align} \label{eq_partialTraceFid}
F(\rho_{AB}, \sigma_{AB})\leq F(\rho_A,\sigma_A) \quad \text{for all} \quad \rho_{AB},\sigma_{AB} \in \Pos(A\otimes B) \, . 
\end{align}
Let $\ket{\psi}_{ABR}$ and $\ket{\phi}_{ABR}$ be purifications of $\rho_{AB}$ and $\sigma_{AB}$, respectively. Uhlmann's theorem shows that
\begin{align}
F(\rho_{AB},\sigma_{AB}) = \sup_{U_R \in \U(R)} |\bra{\psi} \id_{AB} \otimes U_R \ket{\phi}|^2
\end{align}
and
\begin{align}
 F(\rho_{A},\sigma_{A}) = \sup_{U_{BR} \in \U(B \otimes R)} |\bra{\psi} \id_{A} \otimes U_{BR} \ket{\phi}|^2 \, .
\end{align}
This proves~\eqref{eq_partialTraceFid}. By the Stinespring dilation (see Proposition~\ref{prop_stinespring}) there exists an isometry $V \in \V(A,B\otimes R)$ such that 
\begin{align}
F\big(\cE(\rho),\cE(\sigma)) 
= F(\tr_R\, V \rho V^\dagger, \tr_R \, V \rho V^\dagger) 
\geq F( V \rho V^\dagger, V \rho V^\dagger) 
= F(\rho,\sigma) \, ,
\end{align}
where the inequality step uses the DPI for the partial trace (as shown in~\eqref{eq_partialTraceFid}). The final step follows from the isometric invariance of the fidelity.

The joint concavity property of the fidelity follows from Alberti's theorem. For $t\in [0,1]$ and $\rho_1,\rho_2,\sigma_1,\sigma_2 \in \Pos(A)$ we have
\begin{align}
&F\big(t\rho_1 + (1-t)\rho_2, t \sigma_1 +(1-t) \sigma_2 \big)  \nonumber \\
&\hspace{20mm}= \inf_{\omega \in \Poss(A)} \left \lbrace t (\tr\, \rho_1 \omega) (\tr\, \sigma_1 \omega^{-1}) + (1-t)  (\tr\, \rho_2 \omega) (\tr\, \sigma_2 \omega^{-1}) \right \rbrace  \\
&\hspace{20mm} \geq t  \inf_{\omega \in \Poss(A)} \left \lbrace (\tr\, \rho_1 \omega) (\tr\, \sigma_1 \omega^{-1}) \right \rbrace  + (1-t) \inf_{\omega \in \Poss(A)} \left \lbrace (\tr\, \rho_2 \omega) (\tr\, \sigma_2 \omega^{-1}) \right \rbrace  \\
&\hspace{20mm}= t F(\rho_1,\sigma_1) + (1-t) F(\rho_2,\sigma_2) \, .
\end{align}

It thus remains to prove the final statement of the proposition. The joint concavity of the fidelity implies that
\begin{align} \label{eq_concS1}
F\big(t \rho_1 + (1-t) \rho_2 , t \sigma_1 +(1-t)\sigma_2\big) \geq t F(\rho_1,\sigma_1) + (1-t) F(\rho_2,\sigma_2) \, .
\end{align}
For the other direction, let $\Pi_1$ and $\Pi_2$ denote the projectors onto the joint support of $\rho_1, \sigma_1$ and $\rho_2,\sigma_2$, respectively. Furthermore, let $\bar \rho = t \rho_1 + (1-t)\rho_2$ and $\bar \sigma = t \sigma_1 + (1-t)\sigma_2$. The orthogonality assumption implies that $\Pi_1$ and $\Pi_2$ are orthogonal and
\begin{align} \label{eq_proj}
t \rho_1 = \Pi_1 \bar \rho \Pi_1 \qquad \text{and} \qquad (1-t) \rho_2 = \Pi_2 \bar \rho \Pi_2 \, .
\end{align}
Let $\ket{\bar \psi}$ and $\ket{\bar \phi}$ be purifications of $\bar \rho$ and $\bar \sigma$, respectively, such that $F(\bar \rho, \bar \sigma)=| \bra{\bar\psi} \ket{\bar \phi}|^2$. Equation~\eqref{eq_proj} thus implies that $\Pi_1 \ket{\bar \psi}$ and $\Pi_2 \ket{\bar \psi}$ are purifications of $t\rho_1$ and $(1-t) \rho_2$, respectively. Similarly, $\Pi_1 \ket{\bar \phi}$ and $\Pi_2 \ket{\bar \phi}$ are purifications of $t \sigma_1$ and $(1-t) \sigma_2$.
By Uhlmann's theorem (see Theorem~\ref{thm_Uhlmann}) we thus have
\begin{align}
F(\bar \rho, \bar \sigma)
= \left|\bra{\bar\psi} \ket{\bar \phi} \right|^2
= \left| \bra{\bar \psi} \Pi_1 \ket{\bar \phi} + \bra{\bar \psi} \Pi_2 \ket{\bar \phi}  \right|^2
\leq tF( \rho_1, \sigma_1) + (1-t) F(\rho_2,\sigma_2)\, .
\end{align}
Combining this with~\eqref{eq_concS1} proves the assertion. 
\qed
\end{proof}

\subsection{Relative entropy}
Many entropy measures can be expressed in terms of the \emph{relative entropy}. \index{entropy!relative entropy}
\begin{svgraybox}
\vspace{-4mm}
\begin{definition} \label{def_relEnt}
For $\rho \in \St(A)$ and $\sigma \in \Pos(A)$ the \emph{relative entropy} between $\rho$ and $\sigma$ is defined as
\begin{align} \label{eq_relEnt}
D(\rho \| \sigma) := \left \lbrace \begin{array}{ll}
\tr \,\rho (\log \rho - \log \sigma) &  \quad \text{if } \rho \ll \sigma \\
+ \infty & \quad \text{otherwise} \, .
\end{array} \right.
\end{align}
\end{definition}
\vspace{-4mm}
\end{svgraybox}
It is immediate to verify that $H(A)_{\rho} = - D(\rho_A \| \id_A)$, $H(A|B)_{\rho}=-D(\rho_{AB}\| \id_A \otimes \rho_B)$ and 
\begin{align}
I(A:C|B)= D\big(\rho_{ABC} \| \exp (\log \rho_{AB} + \log \rho_{BC} - \log \rho_B)\big) \, .
\end{align}
As a result, in order to understand the mathematical properties of these several different entropy measures it suffices to analyze the relative entropy.
\begin{proposition}[Properties of relative entropy] \label{prop_relEntProperties} \index{data processing inequality}
The relative entropy defined in~\eqref{eq_relEnt} satisfies\vspace{2mm}\\
\setlength{\tabcolsep}{0.23em}
\begin{tabular}{ll}
1.~Additivity: & $D(\rho_1 \otimes \rho_2 \| \sigma_1 \otimes \sigma_2) = D(\rho_1 \| \sigma_1) + D(\rho_2 \| \sigma_2)$ for all $\rho_1 \in \St(A)$,\\ 
&   $\sigma_1 \in \Pos(A)$, $\rho_2 \in \St(B)$, $\sigma_2 \in \Pos(B)$. \vspace{2mm}\\
2.~Nonnegativity: & $D(\rho \| \sigma) \geq 0$ for all $\rho,\sigma \in \St(A)$ with equality if and only if $\rho = \sigma$. \vspace{2mm}\\
3.~Isometric invariance: & $D(V\rho V^\dagger  \| V \sigma V^\dagger)=D(\rho \| \sigma)$ for all $V \in \V(A,B)$, $\rho \in \St(A)$, $\sigma \in \Pos(A)$.\vspace{2mm}\\
4.~DPI: & $D(\rho \| \sigma)\! \geq\! D\big( \cE(\rho) \| \cE(\sigma)\big)$ for all $\rho\! \in\! \St(A)$,\! $\sigma\! \in\! \Pos(A)$,\! $\cE\!\!\in\!\!  \TPCP(A,B)$.\vspace{2mm}\\
5.~Joint convexity: &$(\rho,\sigma) \mapsto D(\rho \| \sigma)$ is jointly convex on $\Pos(A) \times \Pos(A)$.\vspace{2mm}\\
6.~Orthogonal states:& $D(t\rho_1 +(1-t) \rho_2 \| t \sigma_1 +(1-t) \sigma_2) = t D(\rho_1\|\sigma_1) + (1-t)D(\rho_2 \|\sigma_2)$  \\
&  for $t \in [0,1]$, $\rho_1 \in \St(A)$, $\rho_2 \in \St(B)$, $\sigma_1 \in \Pos(A)$, $\sigma_2 \in \Pos(B)$ such that \\
&  both $\rho_1$ and $\sigma_1$ are orthogonal to both $\rho_2$ and $\sigma_2$.\
\end{tabular}
\end{proposition}
\begin{proof}
The properties of the tensor product explained in Exercise~\ref{exercise_tensorProduct} show that 
\begin{align}
D(\rho_1 \otimes \rho_2 \| \sigma_1 \otimes \sigma_2) 
&= \tr\,  \rho_1 \log \rho_1 + \tr \, \rho_2 \log \rho_2 - \tr\, \rho_1 \log \sigma_1 - \tr\, \rho_2 \log \sigma_2 \\
&= D(\rho_1 \| \sigma_1) + D(\rho_2 \| \sigma_2) \,,
\end{align}
which proves the first property.
The positive definiteness property of the relative entropy follows directly from Klein's inequality (see Theorem~\ref{thm_kleinsInequality} with $f(t)=t \log t$ which is strictly convex for $t \in (0,\infty)$).
The relative entropy is invariant under isometries since $\log V \rho V^\dagger = V (\log \rho) V^\dagger$ for every isometry $V$ and since the trace is cyclic. 

The proofs of the data processing inequality and the joint convexity of the relative entropy require more effort. We postpone the proof of these two properties to Section~\ref{sec_EntropyIneq}. There we prove strengthened versions of the DPI (see Theorem~\ref{thm_strengthened_mono}) and the joint convexity property (see Corollary~\ref{cor_jointConvRelEnt}) that immediately imply the two statements of the Lemma.

It thus remains to prove the last assertion of the proposition. By the orthogonality assumption we have
\begin{align}
\log\big(t \rho_1 +(1-t) \rho_2\big) = \log t \rho_1 + \log (1-t)\rho_2 = \log t + \log(1-t) + \log \rho_1 + \log \rho_2 \, ,
\end{align}
which thus implies the desired statement.
\qed
\end{proof}

The relative entropy features a variational formula, i.e., it can be expressed a the following convex optimization problem~\cite{Petz_variational88,FBT15}, which will be important in Chapter~\ref{chapter_recoverability}. 
\begin{svgraybox}
\vspace{-4mm}
\begin{lemma}[Variational formula for relative entropy]\index{variational formula!relative entropy} \label{lem_varFormulaRelEnt}
Let $\rho \in \St(A)$ and $\sigma \in \Pos(A)$. Then
\begin{align} 
D(\rho \| \sigma)  
&= \sup_{\omega \in \Poss(A)}\left \lbrace \tr\, \rho \log \omega - \log \tr\, \ee^{\log \sigma + \log \omega}\right \rbrace \label{eq_varFormulaRelEnt_1}  \\
&=\sup_{\omega \in \Poss(A)}\left \lbrace\tr\, \rho \log \omega +1 - \tr\, \ee^{\log \sigma + \log \omega} \right \rbrace\, . \label{eq_varFormulaRelEnt_2} 
\end{align}
\end{lemma}
\vspace{-4mm}
\end{svgraybox}
\begin{proof}
We first show that for $H \in \Her(A)$ and $\sigma \in \Poss(A)$ we have
\begin{align} \label{eq_step1}
\log \tr \,\ee^{H + \log \sigma} = \max_{\rho \in \St(A)} \{ \tr\, \rho H - D(\rho\| \sigma)Ê\} \, .
\end{align}
To see this define 
\begin{align}
f(\rho)= \tr\, \rho H - D(\rho \| \sigma) \, .
\end{align}
Let $\rho = \sum_{\lambda \in \spec(\rho)}\lambda \Pi_{\lambda}$ denote the spectral decomposition of $\rho$. Since $\rho \in \St(A)$ we have $\sum_{\lambda \in \spec(\rho)} \lambda \leq1$ and $\lambda \geq 0$. We therefore can write
\begin{align}
f\left( \sum_{\lambda \in \spec(\rho)} \lambda \Pi_{\lambda} \right) = \sum_{\lambda \in \spec(\rho)} \left( \lambda \tr \, \Pi_{\lambda} H + \lambda \tr\, \Pi_{\lambda} \log \sigma - \lambda \log \lambda \right) \, .
\end{align}
Since 
\begin{align}
\frac{\partial}{\partial \lambda} f\left( \sum_{\lambda \in \spec(\rho)} \lambda \Pi_{\lambda} \right) \Bigg |_{\lambda=0}  = + \infty \, ,
\end{align}
we can conclude that the minimizer of~\eqref{eq_step1} is a strictly positive operator $\tilde \rho$ with $\tr\, \tilde \rho =1$. For any $K \in \Her(A)$ with $\tr\, K=0$ we have
\begin{align}
0 = \frac{\di }{\di t} f(\tilde \rho + t K) |_{t=0} = \tr\, K(H + \log \sigma - \log \tilde \rho) \, .
\end{align}
This shows that $H + \log \sigma - \log \tilde \rho$ is proportional to the identity operator and hence
\begin{align}
\tilde \rho = \frac{\ee^{H + \log \sigma}}{\tr\, \ee^{H + \log \sigma}} \quad \text{and} \quad f(\tilde \rho) = \log \tr\, \ee^{H + \log \sigma} \, ,
\end{align}
which proves~\eqref{eq_step1}.

We are now ready to prove~\eqref{eq_varFormulaRelEnt_1}. Equation~\eqref{eq_step1} implies that for $\omega \in \Poss(A)$ the functional
\begin{align}
\Her(A) \ni H \mapsto \log \tr \, \ee^{H + \log \omega} 
\end{align}
is convex.\footnote{This can be seen as follows. Let $\cX \ni x\mapsto f(x,y)$ be an affine function. Then, $g(x)=\max_{y\in \cY} f(x,y)$ is convex since for $t \in [0,1]$ we have $g(t x_1 + (1-t)x_2)=\max_{y \in \cY} \{f(t x_1 + (1-t)x_2,y)\}=\max_{y \in \cY} \{t f(x_1,y) + (1-t) f(x_2,y)\} \leq t g(x_1) + (1-t) g(x_2)$.} 
Let $\tilde H= \log \rho - \log \sigma$ and consider the function
\begin{align}
\Her(A) \ni H \mapsto g(H) := \tr\, \rho H - \log \tr\, \ee^{H + \log \sigma} \, ,
\end{align}
which is concave as explained before. For any $K \in \Her(A)$ we have 
\begin{align}
\frac{\di }{\di t} g(\tilde H + t K) |_{t=0} = 0 \, ,
\end{align}
since $\tr \, \rho =1$ and $\frac{\di }{\di t} \tr \, \ee^{\log \rho + t K}|_{t=0}= \tr\, \rho K$. As a result, $\tilde H$ is the maximizer of $g$ and 
\begin{align}
g(\tilde H) = \tr\, \rho (\log \rho - \log \sigma) = D(\rho \| \sigma) \, .
\end{align}
Recalling that every $H \in \Her(A)$ can be written as $H= \log \omega$ for some $\omega \in \Poss(A)$ then proves~\eqref{eq_varFormulaRelEnt_1}.

It thus remains to show~\eqref{eq_varFormulaRelEnt_2}. Note that $\log x \leq x-1$ for $x\in \R_+$ and hence $\log \tr \, \ee^{\log \sigma + \log \omega} \leq  \tr\, \ee^{\log \sigma + \log \omega}-1$. Consequently, we have
\begin{align} \label{eq_part1ofVarPf}
\sup_{\omega \in \Poss(A)} \left \lbrace \tr\, \rho \log \omega - \log \tr\, \ee^{\log \sigma + \log \omega} \right \rbrace
\geq\sup_{\omega \in \Poss(A)}  \left \lbrace \tr\, \rho \log \omega +1 - \tr\, \ee^{\log \sigma + \log \omega} \right \rbrace \, .
\end{align}
Since $\tr\, \rho \log \omega - \log \tr\, \ee^{\log \sigma + \log \omega}$ is invariant under the substitution $\omega \to \alpha \omega$ for $\alpha \in \R_+$ we can assume without loss of generality that $\omega$ is such that $\tr \, \ee^{\log \sigma + \log \omega}=1$. That is, we have
\begin{align}
&\sup_{\omega \in \Poss(A)}  \left \lbrace \tr\, \rho \log \omega - \log \tr\, \ee^{\log \sigma + \log \omega}  \right \rbrace \nonumber \\
&\hspace{30mm}=\sup_{\omega \in \Poss(A)}  \left \lbrace \tr\, \rho \log \omega - \log \tr\, \ee^{\log \sigma + \log \omega} \, : \, \tr \, \ee^{\log \sigma + \log \omega}=1 \right \rbrace \\
&\hspace{30mm}\leq \sup_{\omega \in \Poss(A)}  \left \lbrace \tr\, \rho \log \omega - 1 + \tr \,  \ee^{\log \sigma + \log \omega} \right \rbrace  \, .
\end{align}
Combining this with~\eqref{eq_part1ofVarPf} proves~\eqref{eq_varFormulaRelEnt_2}. \qed
\end{proof}

\begin{exercise} \label{ex_convexOpt}
Verify that the optimization problem in Lemma~\ref{lem_varFormulaRelEnt} is convex optimization problem (i.e., maximizing a concave function over a convex set~\cite{boyd_book}).
\end{exercise}

\begin{remark}
Another variational formula for the relative entropy that is similar to~\eqref{eq_step1} has been derived in~\cite{tropp12}. It states that for any $H \in \Her(A)$ and $\sigma \in \Poss(A)$ we have
\begin{align} \label{eq_varTropp}
\tr \,\ee^{H + \log \sigma} = \max_{\omega \in \Poss(A)} \{ \tr\, \omega H - D(\omega \| \sigma) + \tr\, \omegaÊ\} \, .
\end{align}
\end{remark}

\begin{exercise} \label{ex_tropp}
For any $B \in \Poss(A)$ the trace features the following variational formula~\cite{tropp12}
\begin{align} \label{eq_varTrace}
\tr \, B = \max_{X \in \Poss(A)}\{Ê\tr\, X -D(X\|B)\} \, .
\end{align}
Use Klein's inequality (see Theorem~\ref{thm_kleinsInequality}) to prove~\eqref{eq_varTrace} and show how~\eqref{eq_varTrace} can be used to verify~\eqref{eq_varTropp}.
\end{exercise}

\subsection{Measured relative entropy} \label{sec_measRelEnt}
Another quantity that will be important in this book is the \emph{measured relative entropy} which is defined as a maximization of the classical relative entropy over all measurement statistics that are attainable from two quantum states.~\index{entropy!measured relative entropy} \index{measured relative entropy}
\begin{svgraybox}
\vspace{-4mm}
\begin{definition} \label{def_measRelEnt}
For $\rho \in \St(A)$ and $\sigma \in \Pos(A)$ the \emph{measured relative entropy} between $\rho$ and $\sigma$ is defined as
\begin{align} \label{eq_MeasRelEnt} 
\MD(\rho \| \sigma):=\sup_{(X, M)  } D\big( P_{\rho,M} \big\| P_{\sigma,M} \big) \, ,
\end{align}
with POVMs $M$ on the power-set of a finite set $X$, and $P_{\rho,M}(x) := \tr \, \rho M(x)$.
\end{definition}
\vspace{-4mm}
\end{svgraybox}
At first sight this definition seems cumbersome because we cannot restrict the size of $\cX$ that we optimize over. Alternatively, the measured relative entropy can be expressed as the supremum of the relative entropy with measured inputs over all POVMs, i.e.,
\begin{align} \label{eq_defMeasRelAlternative}
\MD(\rho \| \sigma) = \sup \limits_{M \in \cM}  D\big( M(\rho) \| M(\sigma) \big) \ ,
\end{align}
where $\cM$ is the set of all classical-quantum channels $M(\omega)= \sum_x (\tr\, M_x \omega) \proj{x}$ with $(M_x)$ a POVM and $(\ket{x})$ an orthonormal basis.

As we will see, the measured relative entropy has interesting properties. Furthermore it has a variational characterization, i.e., it can be expressed as the following convex optimization problem~\cite{petz_Entropybook,FBT15}.
\begin{svgraybox}
\vspace{-4mm}
\begin{lemma}[Variational formula for measured relative entropy]\index{variational formula!measured relative entropy} \label{lem_varFormulaMeasRelEnt}
Let $\rho \in \St(A)$ and $\sigma \in \Pos(A)$. Then
\begin{align} \label{eq_varFormulaMeasRelEnt}
 \MD(\rho \| \sigma)  
= \sup_{\omega \in \Poss(A)}  \left \lbrace \tr\, \rho \log \omega - \log \tr\, \sigma \omega \right \rbrace
=\sup_{\omega \in \Poss(A)}  \left \lbrace \tr\, \rho \log \omega +1 - \tr\, \sigma \omega \right \rbrace \, .
\end{align}
\end{lemma}
\vspace{-4mm}
\end{svgraybox}
\begin{proof}
We start by defining the \emph{projective measured relative entropy}, where the measurements are assumed to be projective, i.e.,
\begin{align}
\PD(\rho \|Ê\sigma):= \sup_{\{ \Pi_k\}_{k \in [\dim(A)]}} \left \lbrace \sum_{k=1}^{\dim(A)} \tr\, \Pi_k \rho \log \frac{\tr\, \Pi_k \rho}{\tr\, \Pi_k \sigma} \right \rbrace \, ,
\end{align}
 where $\{\Pi_k \}_{k=1}^{\dim(A)}$ is a set of mutually orthogonal projectors. Without loss of generality it can be assumed that these projectors are rank-one as any course graining of the measurement outcomes can only reduce the relative entropy due to its data-processing inequality (see Proposition~\ref{prop_relEntProperties}).
 We now first show that 
 \begin{align} \label{eq_varFormulaMeasRelEnt_step1}
 \PD(\rho \| \sigma)  
= \sup_{\omega \in \Poss(A)}  \left \lbrace \tr\, \rho \log \omega - \log \tr\, \sigma \omega \right \rbrace
=\sup_{\omega \in \Poss(A)}  \left \lbrace \tr\, \rho \log \omega +1 - \tr\, \sigma \omega \right \rbrace \, .
\end{align}
If $\rho \not \ll \sigma$, all expressions in~\eqref{eq_varFormulaMeasRelEnt_step1} are unbounded. We therefore assume that $\rho \ll \sigma$. We can write
\begin{align} 
&\sup_{\omega \in \Poss(A)}  \left \lbrace \tr\, \rho \log \omega +1 - \tr\, \sigma \omega \right \rbrace   \nonumber \\
&\hspace{25mm}= \sup_{\{ \Pi_k\}_{k \in [\dim(A)]}} \sup_{ \{\lambda_k\}_{k \in [\dim(A)]}} \left \lbrace \sum_{k=1}^{\dim(A)} \left( (\tr\, \Pi_k \rho) \left(\log \lambda_k +1 \right) - \lambda_k \tr\, \Pi_k \sigma \right)   \right \rbrace \, , \label{eq_stepIntermed}
\end{align}  
where $\lambda_k>0$ are the eigenvalues of $\omega$ corresponding to the eigenvectors given by the rank-one projectors $\Pi_k$ and we used that $\tr\, \rho =1$. Since $\rho \ll \sigma$, for all $k \in [\dim(A)]$ such that $\tr\, \Pi_k \sigma = 0$ we also have $\tr\, \Pi_k \rho =0$. If $\tr\, \Pi_k \sigma >0$ and $\tr\, \Pi_k \rho =0$, then the supremum of the $k$-th term is $\sup_{\lambda_k >0} - \lambda_k \tr\, \Pi_k \omega =0$ which is achieved for $\lambda_k \to 0$. As a result, the only relevant case is $\tr\, \Pi_k \sigma > 0$ and $\tr\, \Pi_k \rho >0$. Since, $\lambda_k \mapsto (\tr\, \Pi_k \rho) \left(\log \lambda_k +1 \right) - \lambda_k \tr\, \Pi_k \sigma $ is concave with maximizer $\lambda^\star_k = \frac{\tr\, \Pi_k \rho}{\tr\, \Pi_k \sigma}$. Combining this with~\eqref{eq_stepIntermed} shows that 
\begin{align}
\sup_{\omega \in \Poss(A)}  \left \lbrace \tr\, \rho \log \omega +1 - \tr\, \sigma \omega \right \rbrace 
=\sup_{\{\Pi_k\}_{k \in [\dim(A)]}} \left \lbrace \sum_{k=1}^{\dim(A)} \tr\, \Pi_k \rho \log \frac{\tr\, \Pi_k \rho}{\tr\, \Pi_k \sigma} \right \rbrace =  \PD(\rho \| \sigma)  \, .  \label{eq_da1}
\end{align}
We note that $\log x \leq x-1$ for $x\in \R_+$ and hence $-\log \tr \, \sigma \omega \geq 1 - \tr\, \sigma \omega$. This shows that 
\begin{align}
\sup_{\omega \in \Poss(A)}  \left \lbrace \tr\, \rho \log \omega - \log \tr\, \sigma \omega \right \rbrace
\geq\sup_{\omega \in \Poss(A)}  \left \lbrace \tr\, \rho \log \omega +1 - \tr\, \sigma \omega \right \rbrace \, .  \label{eq_da2}
\end{align}
Since $ \tr\, \rho \log \omega - \log \tr\, \sigma \omega$ is invariant under the substitution $\omega \to \alpha \omega$ for $\alpha \in \R_+$ we can assume without loss of generality that $\omega$ is such that $\tr \,  \sigma \omega = 1$. That is, we have
\begin{align}
\sup_{\omega \in \Poss(A)}  \left \lbrace \tr\, \rho \log \omega - \log \tr\, \sigma \omega \right \rbrace
&= \sup_{\omega \in \Poss(A)}  \left \lbrace \tr\, \rho \log \omega - \log \tr\, \sigma \omega : \tr \,  \sigma \omega = 1 \right \rbrace \\
&\leq \sup_{\omega \in \Poss(A)}  \left \lbrace \tr\, \rho \log \omega +1 - \tr\, \sigma \omega \right \rbrace \, .  \label{eq_da3}
\end{align}
Combining~\eqref{eq_da1},~\eqref{eq_da2}, and~\eqref{eq_da3} proves~\eqref{eq_varFormulaMeasRelEnt_step1}.

It thus remains to show that $ \PD(\rho \| \sigma)   =  \MD(\rho \| \sigma)$. We note that $\PD(\rho \| \sigma)   \leq  \MD(\rho \| \sigma)$ holds by definition and if $\rho \not \ll \sigma$ we have $ \PD(\rho \| \sigma)   =  \MD(\rho \| \sigma)=+\infty$. It thus suffices to prove $\PD(\rho \| \sigma)   \geq  \MD(\rho \| \sigma)$ for $\rho  \ll \sigma$. Let $(\cX,M)$ be a POVM that achieves the measured relative entropy and recall that $P_{\rho,M}(x):=\tr\, M(x) \rho$. For $\cX':=\{x\in \cX: P_{\rho,M}(x) P_{\sigma,M}(x)>0\}$ we find
\begin{align}
\MD(\rho\| \sigma)
&=D(P_{\rho,M}\| P_{\sigma,M}) \\
&= \sum_{x \in \cX'} P_{\rho,M}(x) \log \frac{P_{\rho,M}(x)}{P_{\sigma,M}(x)}\\
&= \tr \, \rho \sum_{x \in \cX'}  M(x) \log \frac{P_{\rho,M}(x)}{P_{\sigma,M}(x)}\\
&=\tr \, \rho \sum_{x \in \cX'}  \sqrt{M(x)} \log \left( \frac{P_{\rho,M}(x)}{P_{\sigma,M}(x)} \id_A \right)\sqrt{M(x)} \, .
\end{align}
The operator Jensen inequality (see Theorem~\ref{thm_Jensen}) then shows that
\begin{align} \label{eq_doneMeasRelEnt}
\MD(\rho\| \sigma)  \leq \tr\, \rho \log \left( \sum_{x \in \cX'} M(x) \frac{P_{\rho,M}(x)}{P_{\sigma,M}(x)} \right)
= \tr\, \rho \log \omega + 1 - \tr \, \sigma \omega
\leq \PD(\rho\|\sigma) \, ,
\end{align}
for $\omega = \sum_{x \in \cX'} M(x) \frac{P_{\rho,M}(x)}{P_{\sigma,M}(x)}$, since 
\begin{align}
\tr\, \sigma \omega = \sum_{x \in \cX'} P_{\sigma,M}(x)  \frac{P_{\rho,M}(x)}{P_{\sigma,M}(x)} = \sum_{x \in \cX'} P_{\rho,M}(x) =1\, .
\end{align}
The final step in~\eqref{eq_doneMeasRelEnt} follows from~\eqref{eq_varFormulaMeasRelEnt_step1}. This proves the assertion.
 \qed
\end{proof}

The measured relative entropy has remarkable properties. Several of them are directly inherited from the relative entropy.
\begin{proposition}[Properties of measured relative entropy]\label{prop_MeasRelEnt}
The measured relative entropy defined in~\eqref{eq_MeasRelEnt} satisfies\vspace{2mm}\\
\setlength{\tabcolsep}{0.32em}
\begin{tabular}{ll}
1.~Submultiplicativity: &  $\MD(\rho_1 \otimes \rho_2 \| \sigma_1 \otimes \sigma_2) \geq \MD(\rho_1 \| \sigma_1) + \MD(\rho_2 \| \sigma_2)$ for all $\rho_1\in \St(A)$,\\
& $\sigma_1 \in \Pos(A)$, $\rho_2 \in \St(B)$, $\sigma_2 \in \Pos(B)$.\vspace{2mm}\\
2.~Nonnegativity: & $\MD(\rho \| \sigma) \geq 0$ for all $\rho,\sigma \in \St(A)$ with equality if and only if $\rho = \sigma$. \vspace{2mm}\\
3.~Isometric invariance: & $\MD(V\rho V^\dagger \| V \sigma V^\dagger)=\MD(\rho \| \sigma)$ for all $V \in \V(A,B)$, $\rho \in \St(A)$, \\
& $\sigma\in \Pos(A)$.\vspace{2mm}\\
4.~DPI: & $\MD(\rho \| \sigma) \geq \MD\big( \cE(\rho) \| \cE(\sigma)\big)$ for all $\rho \inÊ\St(A), \sigma \in \Pos(A)$,\\
&$\cE \in  \TPCP(A,B)$.\vspace{2mm}\\
5.~Joint convexity: &$(\rho,\sigma) \mapsto \MD(\rho \| \sigma)$ is jointly convex on $\Pos(A) \times \Pos(A)$.\vspace{2mm}\\
6.~Orthogonal states:& $\MD(t\rho_1\! +\!(1\!-\!t) \rho_2 \| t \sigma_1 \!+\!(1\!-\!t) \sigma_2)\! =\! t \MD(\rho_1\|\sigma_1) \!+\! (1\!-\!t)\MD(\rho_2 \|\sigma_2)$  \\
&  for $t \in [0,1]$, $\rho_1 \in \St(A)$, $\rho_2 \in \St(B)$, $\sigma_1 \in \Pos(A)$, $\sigma_2 \in \Pos(B)$ such that \\
&  both $\rho_1$ and $\sigma_1$ are orthogonal to both $\rho_2$ and $\sigma_2$.\
\end{tabular}
\end{proposition}
\begin{proof}
The submultiplicativity follows by definition of the measured relative entropy. 
The nonnegativity property is directly inherited from the classical relative entropy. The isometric invariance can be easily derived from the variational formula~\eqref{eq_varFormulaMeasRelEnt}. Let $\omega \in \Pos(B)$ be the optimizer for $\MD(V\rho V^\dagger \| V \sigma V^\dagger) $. Then,
\begin{align}
\MD(V\rho V^\dagger \| V \sigma V^\dagger)  
&= \tr \, V \rho V^{\dagger} \log \omega - \log \tr \, V \sigma V^{\dagger} \omega \\
&= \tr \, \rho \log(V^\dagger \omega V) - \log \tr \, \sigma V^{\dagger} \omega V \\
&\leq \MD(\rho \| \sigma) \, ,
\end{align}
where the final inequality step uses that $V^\dagger \omega V \in \Pos(A)$. Conversely, for $\omega \in \Pos(A)$ being the optimizer for $\MD(\rho \| \sigma)$ we find
\begin{align}
\MD(V\rho V^\dagger \| V \sigma V^\dagger)  
&\geq  \tr \, V \rho V^{\dagger} \log  V \omega V^\dagger - \log \tr \, V \sigma V^{\dagger} V \omega V^\dagger  \\
&= \tr \, \rho \log \omega - \log \tr \, \sigma  \omega \\
&= \MD(\rho \| \sigma) \, .
\end{align}

The joint convexity follows from the joint convexity of the relative entropy. For $t \in [0,1]$, $\rho_1,\rho_2 \in \St(A)$, $\sigma_1,\sigma_2 \in \Pos(A)$ we have
\begin{align}
\MD\big(t \rho_1 \!+\! (1\!-\! t)\rho_2 \| t \sigma_1 \!+\! (1\!-\! t) \sigma_2 \big)
&= D\big( t P_{\rho_1,M} + (1-t) P_{\rho_2,M} \| t P_{\sigma_1,M} + (1-t) P_{\sigma_2,M} \big) \\
& \leq t D(P_{\rho_1,M} \| P_{\sigma_1,M}) + (1-t) D(P_{\rho_2,M} \| P_{\sigma_2,M}) \\
&\leq t \MD(\rho_1 \| \sigma_1) + (1-t) \MD(\sigma_1 \| \sigma_2) \, .
\end{align}
It is well-known (see, e.g.~\cite[Proposition~4.2]{marco_book}) that the joint convexity property (together with the unitary invariance and the submultiplicativity property) implies the data-processing inequality.

It thus remains to verify the final statement of the proposition. Recall that the measured relative entropy can be expressed as~\eqref{eq_defMeasRelAlternative}.
Let $(M_x)$ and $(M'_y)$ be POVMs such that 
\begin{align}
M(t \rho_1 +(1-t)\rho_2) = t \sum_x \tr\, M_x \rho_1 \proj{x} + (1-t) \sum_y \tr\, M'_y \rho_2 \proj{y} \, .
\end{align}
We thus find
\begin{align}
&\MD\big(t \rho_1 +(1-t) \rho_2 \| t \sigma_1 +(1-t) \sigma_2 \big)  \nonumber \\
&\geq D\big( M(t \rho_1 +(1-t) \rho_2) \| M(t \sigma_1 +(1-t) \sigma_2) \big) \\ 
&= t D\left( \sum_x \tr\, M_x \rho_1 \proj{x} \Big\| \sum_x \tr \, M_x \sigma_1 \proj{x} \right) + (1-t) D\left( \sum_y \tr\, M'_y \rho_2 \proj{y} \Big\| \sum_y \tr \, M_y \sigma_2 \proj{y} \right) \nonumber 
\end{align}
where final penultimate step uses Proposition~\ref{prop_relEntProperties}. As this is valid for all POVMs $(M_x)$ and $(M'_y)$, we can take the supremum over those and thus obtain
\begin{align}
\MD\big(t \rho_1 +(1-t) \rho_2 \| t \sigma_1 +(1-t) \sigma_2 \big)  \geq t \MD(\rho_1 \| \sigma_1) + (1-t) \MD(\rho_2 \| \sigma_2) \, . \label{eq_sss1}
\end{align}
The other direction follows by the joint convexity of the relative entropy (see Proposition~\ref{prop_relEntProperties}). By~\eqref{eq_defMeasRelAlternative} there exists a POVM $(\bar M_x)$ such that 
\begin{align}
&\MD\big(t \rho_1 +(1-t) \rho_2 \| t \sigma_1 +(1-t) \sigma_2 \big)  \nonumber \\
&= D\left( \sum_x (t \tr\, \bar M_x \rho_1 +(1-t) \tr\, \bar M_x \rho_2 ) \proj{x} \Big\|  \sum_x (t \tr\, \bar M_x \sigma_1 +(1-t) \tr\, \bar M_x \sigma_2 )  \proj{x} \right) \\
&\leq t D\left( \sum_x  \tr\, \bar M_x \rho_1\proj{x} \Big\|  \sum_x  \tr\, \bar M_x \sigma_1   \proj{x} \right) + (1-t)D\left( \sum_x  \tr\, \bar M_x \rho_2\proj{x} \Big\|  \sum_x  \tr\, \bar M_x \sigma_2   \proj{x} \right)  \nonumber \\
&\leq t \MD(\rho_1 \| \sigma_1) + (1-t) \MD(\rho_2 \| \sigma_2) \, .
\end{align}
Combining this with~\eqref{eq_sss1} proves the assertion.  \qed 
\end{proof}
Unlike the relative entropy, the measured relative entropy is not additive under tensor products. The following proposition states how the measured relative entropy is related to the relative entropy and the fidelity.
\begin{proposition} \label{prop_measRelFid}
Let $\rho \in \St(A)$ and $\sigma \in \Pos(A)$. The measured relative entropy defined in~\eqref{eq_MeasRelEnt} satisfies
\begin{enumerate}
\item $\MD(\rho\| \sigma) \leq D(\rho \| \sigma)$ with equality if and only if $[\rho,\sigma]=0$. \vspace{1mm}
\item $\MD(\rho\| \sigma) \geq -\log F(\rho,\sigma)$. \vspace{1mm}
\item $\lim_{n\to \infty} \frac{1}{n} \MD(\rho^{\otimes n} \| \sigma^{\otimes n}) = D(\rho \|Ê\sigma)$.
\end{enumerate}
\end{proposition}
\begin{proof}
The first property of the proposition follows directly from the Golden-Thompson inequality (see Theorem~\ref{thm_GT}) together with the variational formulas for the relative and measured relative entropy (see Lemma~\ref{lem_varFormulaRelEnt} and Lemma~\ref{lem_varFormulaMeasRelEnt}, respectively).
To prove the second property, we recall that by Alberti's theorem (see Theorem~\ref{thm_Alberti}) there exists $\omega \in \Poss(A)$ such that
\begin{align}
-\log F(\rho,\sigma) 
&= -\log \tr\, \rho \omega - \log \tr\, \sigma \omega^{-1}\\
&\leq - \log \tr \, \ee^{\log \rho + \log \omega} - \log \tr\, \sigma \omega^{-1}\\
&\leq \tr\, \rho \log \omega^{-1} - \log \tr\, \sigma \omega^{-1}\\
&\leq \MD(\rho \| \sigma) \, ,
\end{align}
where the first inequality follows from the Golden-Thompson inequality.
The second inequality uses the Peierls-Bogoliubov inequality (see Corollary~\ref{cor_Peierls} applied for $H_1=\log \rho$ and $H_2= \log \omega$). The final step uses the variational formula for the measured relative entropy (see Lemma~\ref{lem_varFormulaMeasRelEnt}).
The third statement of the proposition is proven in~\cite[Section~4.3.3]{marco_book}. \qed
\end{proof}
We have seen in Proposition~\ref{prop_MeasRelEnt} that the measured relative entropy is jointly convex in its arguments. The following lemma shows that the measured relative entropy also satisfies a weak form of a concavity property in its second argument~\cite[Lemma~3.11]{STH15}.
\begin{lemma} \label{lem_singleLetterMeasRelEnt} 
Let $X$ be a compact space. For any probability measure $\mu$ on $X$, any sequence $(\sigma_x)_{x \in X}$ such that $\sigma_x \in \Pos(A)$ for all $x \in X$, any $\rho \in \St(A)$ and any $n \in \N$, we have
\begin{align} \label{eq_convMeasRelEnt}
\frac{1}{n}\MD\Bigl( \rho^{\otimes n} \Big\| \int_X \mu(\di x) \sigma_x^{\otimes n}  \Bigr) \geq \min_{\sigma \in \conv\{\sigma_x\, : \, x \in X \} } \MD(\rho \| \sigma) \ .
\end{align}
\end{lemma}
\begin{proof}
The variational characterization for the measured relative entropy given by Lemma~\ref{lem_varFormulaMeasRelEnt} implies
\begin{align}
\MD\Bigl( \rho^{\otimes n} \Big\| \int_X \mu(\di x) \sigma_x^{\otimes n}  \Bigr)  
&\geq \sup_{\omega \in \Poss(A)} \Big \{ \tr \bigl( \rho^{\otimes n} \log \omega^{\otimes n} \bigr) - \log \tr \Bigl( \int_X \mu(\di x) \sigma_x^{\otimes n} \omega^{\otimes n} \Bigr)  \Big \} \\
&\geq \sup_{\omega \in \Poss(A)} \min_{x \in X} \big \{ n \tr(\rho \log \omega) - n \log\tr(\sigma_x \omega) \big \} \ .
\end{align}
For $x \in \R_+$, clearly $\log x \leq x -1$ and thus $- \log \tr(\sigma \omega) \geq 1 - \tr(\sigma \omega)$ for all $\omega\in \Poss(A)$. This implies that
\begin{align}
\MD\Bigl( \rho^{\otimes n} \Big\| \int_X \mu(\di x) \sigma_x^{\otimes n}  \Bigr) 
&\geq n \sup_{\omega \in \Poss(A)} \min_{x \in X} \big \{ \tr(\rho \log \omega) + 1 - \tr(\sigma_x \omega)  \big \} \\
&\geq n \sup_{\omega \in \Poss(A)} \min_{\sigma \in \conv\{\sigma_x \, : \, x \in X \}} \big \{ \tr(\rho \log \omega) + 1 - \tr(\sigma \omega) \big \} \ .
\end{align}

The function $\omega \mapsto \tr(\rho \log \omega) + 1 - \tr(\sigma \omega) $ is concave and the function $\sigma \mapsto  \tr(\rho \log \omega) + 1 - \tr(\sigma \omega) $ is linear. The set $\conv\{\sigma_x : x \in X \}$ is compact and convex and the set of strictly positive operators is convex.
 As a result we can apply Sion's minimax theorem~\cite{Sion58} which gives
\begin{align}
\frac{1}{n}\MD\Bigl( \rho^{\otimes n} \Big\| \int_X \mu(\di x) \sigma_x^{\otimes n}  \Bigr) 
&\geq \min_{\sigma \in \conv\{\sigma_x  \,: \,  x \in X \}} \sup_{\omega \in \Poss(A)}   \big \{ \tr(\rho \log \omega) + 1 - \tr(\sigma \omega) \big \}  \\
& = \min_{\sigma \in \conv\{\sigma_x  \, : \, x \in X \} } \MD(\rho \| \sigma)  \ ,
\end{align}
where the final step follows by the variational characterization of the measured relative entropy given in Lemma~\ref{lem_varFormulaMeasRelEnt}. \qed
\end{proof}

\begin{remark}
We note that Lemma~\ref{lem_singleLetterMeasRelEnt} is no longer valid if the measured relative entropy terms in~\eqref{eq_convMeasRelEnt} are replaced with relative entropy terms. This can be seen by contradiction. Suppose~\eqref{eq_convMeasRelEnt} is valid for relative entropies. Theorem 12 from~\cite{hirche17} implies that for any $\rho_{ABC} \in \St(AÊ\otimes B \otimes C)$ we have\footnote{This is explained in more detail in Remark~\ref{rmk_multiLetter}.}
\begin{align}
I(A:C|B)_{\rho} 
&\geq \lim \sup_{n \to \infty} \frac{1}{n} D\Big(\rho_{ABC}^{\otimes n} \| \int_{-\infty}^{\infty} \di t  \beta_0(t) \cT^{[t]}_{B \to BC}(\rho_{AB})^{\otimes n} \Big)\\
&\geq D\big(\rho_{ABC} \| \cR_{B \to BC}(\rho_{AB}) \big) \quad \text{\lightning} \label{eq_contrad}
\end{align}
where $\beta_0$ is a probability density defined in~\eqref{eq_beta_0}, $\cT^{[t]}_{B \to BC}$ is a recovery map defined in~\eqref{eq_PetzRecMap2} for all $t \in \R$, and a recovery map $\cR_{B \to BC} \in \TPCP(B, B \otimes C)$. Inequality~\eqref{eq_contrad} however is in contradiction with~\cite[Section~5]{fawzi17} (see Remark~\ref{rmk_FR_optimal} for further details) which shows that~\eqref{eq_convMeasRelEnt} is not valid for relative entropies. 
\end{remark}

\subsection{R\'enyi relative entropy} \label{sec_oneShot} \index{R\'enyi relative entropy}
There exist different families of relative entropies that are useful in quantum information theory. Among the most prominent examples are the so-called \emph{R\'enyi relative entropies} that are carefully discussed in several textbooks such as, e.g.,~\cite{marco_book}.
In this section, we review a specific member of this family called the \emph{minimal R\'enyi relative entropy} that has been introduced in~\cite{MLDSFT13,wilde_strong_2014}.
\begin{svgraybox}
\vspace{-4mm}
\begin{definition} \index{minimal R\'enyi relative entropy} \label{def_sandwiched}
For $\alpha \in (0,1) \cup (1,\infty)$, $\rho \in \St(A)$ and $\sigma \in \Pos(A)$ the \emph{minimal R\'enyi relative entropy} between $\rho$ and $\sigma$ is defined as
\begin{align} \label{eq_a_relEnt}
D_{\alpha}(\rho \| \sigma) := \left \lbrace \begin{array}{ll}
\frac{\alpha}{\alpha -1} \log \norm{\sigma^{\frac{1-\alpha}{2 \alpha}} \rho \sigma^{\frac{1-\alpha}{2 \alpha}}}_{\alpha} &  \quad \text{if } \rho \ll \sigma \,  \text{ or } \, \alpha <1 \\
+ \infty & \quad \text{otherwise} \, .
\end{array} \right.
\end{align}
\end{definition}
\vspace{-4mm}
\end{svgraybox}
The minimal R\'enyi relative entropy is also known as \emph{sandwiched R\'enyi relative entropy}.\index{sandwiched R\'enyi relative entropy} It satisfies many desirable properties. We will only discuss those that are relevant for this book. The interested reader can find a more detailed treatment about this entropy measure in~\cite{marco_book}.

The family of minimal R\'enyi relative entropies comprises three particularly well-known one-shot relative entropies, i.e., the \emph{min-relative entropy}~\cite{renner_phd} \index{min-relative entropy}
\begin{align}
D_{\min}(\rho \| \sigma):=-\log \norm{\sqrt{\rho} \sqrt{\sigma}}^2_1  = - \log F(\rho,\sigma)=D_{\frac{1}{2}}(\rho \| \sigma) \, ,
\end{align}
the relative entropy
\begin{align}
D(\rho \| \sigma)= \lim_{\alpha \to 1} D_{\alpha}(\rho \| \sigma)   \, ,
\end{align}
and the \emph{max-relative entropy}~\cite{datta09,renner_phd}\index{max-relative entropy}
\begin{align}
D_{\max}(\rho \|Ê\sigma):=  \inf \{Ê\lambda \in \R : \rho \leq 2^{\lambda} \sigma \} = \log \norm{\sigma^{-\frac{1}{2}} \rho\sigma^{-\frac{1}{2}} }_{\infty}=\lim_{\alpha \to \infty} D_{\alpha}(\rho \| \sigma)  \, .
\end{align}
As the names suggest, the min-relative entropy cannot be larger than the max-relative entropy, or more precisely we have
\begin{align} \label{eq_DminMax}
D_{\min}(\rho\|\sigma) \leq D(\rho \| \sigma) \leq D_{\max}(\rho\| \sigma) \, ,
\end{align}
with strict inequalities in the generic case. 
The max-relative entropy turns out to be the largest relative entropy measure that satisfies the data-processing inequality and is additive under tensor products~\cite[Section~4.2.4]{marco_book}.
It is known that the minimal $\alpha$-R\'enyi relative entropy is monotonically increasing in $\alpha$~\cite{MLDSFT13}.
\begin{lemma} \label{lem_monoMinimal}
Let $\rho \in \St(A)$, $\sigma \in \Pos(A)$, $\alpha,\alpha' \in (0,\infty)$ such that $ \alpha \leq \alpha' $. Then
\begin{align} \label{eq_mono}
D_{\alpha}(\rho \| \sigma) \leq D_{\alpha'}(\rho \| \sigma) \, .
\end{align}
\end{lemma}
The minimal R\'enyi divergence vanishes if and only if its two arguments coincide, i.e.,
\begin{align} \label{eq_DaZero}
D_\alpha(\rho \| \sigma) = 0 \quad \text{for} \quad \alpha \in (\tfrac{1}{2},1) \cup (1,\infty) \qquad \iff \qquad \rho = \sigma  \, .
\end{align}
To see this we note that Lemma~\ref{lem_monoMinimal} guarantees that $D_\alpha(\rho \| \sigma) = 0$ implies $D_{\frac{1}{2}}(\rho \| \sigma) = 0$ and hence by Proposition~\ref{prop_propFid} we have $\rho=\sigma$. The other direction follows by definition of the minimal R\'enyi divergence.

It is well-known that the relative entropy does not satisfy the triangle inequality. For the three (classical) qubit states $\rho=\frac{1}{2} \proj{0} + \frac{1}{4} \id_2$, $\sigma=\frac{1}{2} \proj{1} + \frac{1}{4} \id_2$, and $\omega = \frac{ 1}{2} \id_2$ we have $D(\rho \| \sigma) > D(\rho \| \omega) + D(\omega \| \sigma)$. The following lemma proves a triangle-like inequality for the minimal quantum R\'enyi relative entropy~\cite{Christandl2017,sutter17}.
\begin{lemma} \label{lem_triangleD}
Let $\rho \in \St(A)$, $\sigma,\omega \in \Pos(A)$ and let $\alpha \in [\frac{1}{2},\infty)$. Then
\begin{align} \label{eq_triangleD}
D_{\alpha}(\rho \| \sigma) \leq D_{\alpha}(\rho \| \omega) + D_{\max}(\omega \| \sigma) \, .
\end{align}
\end{lemma}
\begin{proof}
For $\alpha \in [\frac{1}{2},1)$, the function $t \mapsto t^{\frac{1-\alpha}{\alpha}}$ is operator monotone on $[0,\infty)$~(see Table~\ref{table_opConvex}). Furthermore, according to Proposition~\ref{prop_traceFunctions}, the function $\Pos(A) \ni X \mapsto \tr\, X^{\alpha}$ is monotone. By definition of the max-relative entropy we find
\begin{align}
D_{\alpha}(\rho \| \sigma) 
= \frac{1}{\alpha-1} \log \tr \Big( \rho^{\frac{1}{2}} \sigma^{\frac{1-\alpha}{\alpha}} \rho^{\frac{1}{2}} \Big)^{\alpha} 
\leq D_{\alpha}(\rho \|Ê\omega) + D_{\max}(\omega \| \sigma) \, . 
\end{align}
For $\alpha \in (1,\infty)$ the argument is exactly the same, where we note that $t \mapsto t^{\frac{1-\alpha}{\alpha}}$ is operator anti-monotone~(see Table~\ref{table_opConvex}). The case $\alpha=1$ then follows by continuity. \qed
\end{proof}

\section{Background and further reading}
We refer to Bhatia's book~\cite[Chapter~IV]{bhatia_book} for a comprehensive introduction to matrix norms.
Functions on Hermitian operators are carefully treated in Carlen's book~\cite{carlen_book}, Bhatia's book about matrix analysis~\cite{bhatia_book} (see also~\cite{bhatia_psd_book} for an emphasis on positive definite operators), Hiai and Petz' book~\cite{hiai2014introduction}, Simon's book~\cite{simon_book79}, Ohya and Petz' book~\cite{petz_statbook}, and Zhang's book~\cite{zhangbook}. An important result for operator monotone and operator convex function is the L\"owner-Heinz theorem~\cite{lowner34} (see also~\cite{Donoghue1974} for a more general version) which is summarized in Table~\ref{table_opConvex}. An alternative proof for the Peierls-Bogoliubov theorem can be found in~\cite[Theorem~2.12]{carlen_book}.
Lieb's theorem was proven in the remarkable paper~\cite{Lieb73}. Tropp showed how Lieb's theorem can be derived from the joint convexity of the relative entropy~\cite{tropp12}.

Entropy measures are carefully discussed in various books, such as the one by Ohya and Petz~\cite{petz_statbook}, Nielsen and Chuang~\cite{nielsenChuang_book}, Wilde~\cite{wilde_book}, Hayashi~\cite{hayashi_book,hayashi_book2}, Tomamichel~\cite{marco_book}, and Holevo~\cite{holevo_book}. The fidelity was introduced by Uhlmann~\cite{Uhl76} and later popularized in quantum information theory by Josza~\cite{jozsa94}.
The fidelity features another characterization that is not discussed here. It can be expressed as a semidefinite program~\cite{Watrous09}. Appendix B of~\cite{FR14} discussed further interesting properties of the fidelity.
 The relative entropy was introduced by Umegaki~\cite{umegaki62} and then used in mathematical physics by Lindblad~\cite{lindblad75}. Recently it was shown~\cite{hermes15} that the DPI for the relative entropy is valid even for trace-preserving positive maps.
The measured relative entropy was first studied by Donald~\cite{Donald1986} as well as Hiai and Petz~\cite{Hiai1991}. 
More information about quantum channels can be found in Wolf's lecture notes~\cite{Wolf_skript} and Holevo's book~\cite{holevo_book}.


\chapter{Tools for non-commuting operators}
\label{chapter_nonCommute} 

\abstract{ Complementarity is one of the central mysteries of quantum mechanics. In the mathematical formalism this is represented by the fact that different operators do not necessarily commute. In this chapter we discuss two different mathematical tools to deal with non-commuting operators.}\index{complementarity}
\vspace{8mm}
\noindent One eminent difference between classical physics and quantum mechanics is the principle of complementarity. This phenomenon arises from the fact that quantum mechanical operators (unlike classical ones) do not commute in general. Complementarity summarizes different purely quantum mechanical features such as \emph{uncertainty relations}~\cite{Heisenberg1927,coles17} or the \emph{wave-particle duality}~\cite{feynman1964}.     

On a more technical level, the complementarity aspect of quantum mechanics displays a major hurdle in the rigorous understanding of the behavior of quantum mechanical systems. To name one example, consider the conditional mutual information. Let $P_{XYZ}$ denote a classical tripartite distribution. It is straightforward to verify that the conditional mutual information defined in~\eqref{eq_classicalCMI} is nonnegative, i.e., $I(X:Z|Y)_P\geq 0$.\footnote{This follows for example immediately from the variational formula for the (classical) conditional mutual information given in~\eqref{eq_ClassicalCMIRelEnt}.} For quantum mechanical systems this gets more complicated. The celebrated \emph{strong subadditivity of quantum entropy} (SSA)~\cite{LieRus73_1,LieRus73} \index{strong subadditivity} ensures that for any tripartite density operator $\rho_{ABC}$ we have 
\begin{align} \label{eq_SSAsecPinching}
I(A:C|B)_{\rho}:=H(AB)_{\rho}+H(BC)_{\rho}-H(ABC)_{\rho}-H(B)_{\rho} \geq 0\, .
\end{align}
Unlike the classical case, this result is far from being trivial which is mainly due to the fact that density operators and their marginals do not commute. We will discuss the proof of SSA in Section~\ref{sec_AQMC}.

To understand the properties of quantum mechanical systems, we need tools to deal with non-commuting operators. In this chapter, we will discuss two techniques that can be useful for this purpose --- the \emph{method of pinching} and \emph{complex interpolation theory}. Another tool that is helpful are trace inequalities which are discussed in Chapter~\ref{chapter_traceIneq}.
\section{Pinching} \label{sec_pinching} \index{pinching!spectral}
Any Hermitian operator $H \in \Her(A)$ has a spectral decomposition, i.e., it can be written as
\begin{align} \label{eq_spectral_dec}
H = \sum_{\lambda \in \spec(H)} \lambda \Pi_{\lambda} \, ,
\end{align}
where $\lambda \in \spec(H) \subseteq \R$ are unique eigenvalues and $\Pi_{\lambda}$ are mutually orthogonal projectors. For $\kappa >0$, let us define the following family of probability densities on $\R$
\begin{align} \label{eq_kappaDist}
\mu_{\kappa}( t) = \frac{12}{\pi \kappa^3 t^4} \left(3 + \cos(\kappa t) - 4 \cos\Big(\frac{\kappa t}{2} \Big) \right) .
\end{align}
Its Fourier transform $\hat \mu_{\kappa}$ turns out to be a convolution of two centered triangular functions of width $\kappa$, i.e.,
\begin{align}
\hat \mu_{\kappa}(\omega) = \frac{3}{\kappa} (\mathrm{tri}_\kappa \star \mathrm{tri}_\kappa)(\omega) \, ,
\end{align}
where
\begin{align} \index{triangular function}
\mathrm{tri}_\kappa(x) := \left \lbrace \begin{array}{l l}
1- \frac{2 |x|}{\kappa} & |x| \leq \kappa \\
0 & \text{otherwise}\, .
\end{array} \right .
\end{align}
We immediately see that $\hat \mu_{\kappa}$ satisfies the following properties:
\begin{enumerate}
\item $\hat \mu_{\kappa}(0) =1$. \label{prop_i}  \vspace{1mm}
\item $\hat \mu_{\kappa}(\omega)= 0$ if and only if $|\omega|\geq \kappa$. \label{prop_ii}  \vspace{1mm}
\item $\omega \mapsto \hat \mu_{\kappa}(\omega)$ is a real valued even function. \label{prop_iii}  \vspace{1mm}
\item $\omega \mapsto \hat \mu_{\kappa}(\omega)$ is monotonically decreasing for $\omega \in \R_+$. \label{prop_iv}  \vspace{1mm}
\item $\hat \mu_{\kappa}(\omega) \in [0,1]$. \label{prop_v}
\end{enumerate}
\begin{exercise} \label{ex_pinchingFT}
Verify that $\mu_{\kappa}$ is a probability distribution on $\R$ for all $\kappa >0$ and its Fourier transform $\hat \mu_\kappa$ satisfies the properties given above.
\end{exercise}

\subsection{Spectral pinching}\index{pinching!spectral} \label{subsec_pinching}
The motivation for studying the spectral pinching method arises from the following (vague) question: Given two Hermitian operators $H_1$ and $H_2$ that do not commute. Does there exist a method to modify one of the two operators such that they commute without completely destroying the structure of the original operator? The spectral pinching method achieves this task. Before explaining this method in detail we have to introduce the pinching map.
\begin{svgraybox}
\vspace{-4mm}
\begin{definition} \label{def_pinching} \index{pinching!regular pinching}
Let $H \in \Her(A)$ with a spectral decomposition given in~\eqref{eq_spectral_dec}. The \emph{pinching map} with respect to $H$ is defined as 
\begin{align} \label{eq_pinchingDef}
\cP_{H}\, : \,  \Her(A) \ni  X \mapsto  \sum_{\lambda \in \spec(H)} \Pi_{\lambda} X \Pi_{\lambda} \, .
\end{align}
\end{definition}
\vspace{-4mm}
\end{svgraybox}
Pinching maps have several nice properties. They are trace-preserving, completely positive, unital, self-adjoint, and can be viewed as dephasing operations that remove off-diagonal blocks of an operator.\footnote{Hence the name \emph{pinching} map, as it pinches the off-diagonal blocks.} As a result, if we pinch a Hermitian operator $H_1$ with respect to another Hermitian operator $H_2$, the resulting operator $\cP_{H_2}(H_1)$ commutes with $H_2$. This will be explained more carefully in Lemma~\ref{lem_propertiesPinching}.
\begin{exercise}\label{ex_pinchingMapTPCP}
Verify that the pinching map is trace-preserving, completely positive and unital.
\end{exercise}

The pinching map features an alternative representation. It can be written as an average over commuting unitaries. The \emph{spectral gap} of a Hermitian operator $H$ with eigenvalues $(\lambda_k)_{k}$ is defined as the smallest distance of two distinct eigenvalues, i.e., $\Delta_H:=\min \{ |\lambda_k - \lambda_j| \, : \, \lambda_k \ne \lambda_j \}$.\index{spectral gap}
\begin{lemma}[Integral representation of pinching map] \label{lem_IntegralPinching} \index{pinching!integral representation}
Let $H \in \Her(A)$ and $\mu_{\kappa}$ as defined in~\eqref{eq_kappaDist}. Then
\begin{align}
\cP_{H}(X) = \int_{-\infty}^{\infty} \di t \mu_{\Delta_H}(t) \, \ee^{\ci t H} X \ee^{-\ci t H} \quad \text{for all} \quad X\in \Her(A)  \, .
\end{align}
\end{lemma}
\begin{proof}
We start by recalling the spectral decomposition of $H$, i.e.,
\begin{align}
H = \sum_{\lambda \in \spec(H)} \lambda \Pi_{\lambda} \, ,
\end{align}
and the fact that eigenvectors corresponding to distinct eigenvalues of Hermitian operators are orthogonal. We thus have for any $t \in \R$
\begin{align}
\ee^{\ci t H} = \sum_{\lambda \in \spec(H)} \ee^{\ci t \lambda} \Pi_{\lambda}
\end{align}
and
\begin{align}
\ee^{\ci t H} X \ee^{-\ci t H} = \sum_{\lambda,\lambda' \in \spec(H)} \ee^{-\ci t (\lambda' - \lambda)} \Pi_{\lambda} X \Pi_{\lambda'} \, .
\end{align}
With this we obtain
\begin{align}
\int_{-\infty}^{\infty} \di t \mu_{\Delta_H}(t) \, \ee^{\ci t H} X \ee^{- \ci t H}
&= \int_{-\infty}^{\infty} \di t \mu_{\Delta_H}( t) \, \sum_{\lambda,\lambda' \in \spec(H)} \ee^{-\ci t (\lambda' - \lambda)} \Pi_{\lambda} X \Pi_{\lambda'} \\
&= \sum_{\lambda,\lambda' \in \spec(H)} \Pi_{\lambda} X \Pi_{\lambda'} \,  \hat \mu_{\Delta_H}(\lambda' - \lambda) \, ,
\end{align}
where in the final step we used the linearity of the integral to interchange the integral and the summation. Employing Property~\ref{prop_i} and Property~\ref{prop_ii} of $\hat \mu_{\Delta_H}$ and the definition of the spectral gap $\Delta_H$ we obtain
\begin{align}
\int_{-\infty}^{\infty} \di t \mu_{\Delta_H}( t) \, \ee^{\ci t H} X \ee^{- \ci t H} 
= \sum_{\lambda \in \spec(H)} \Pi_{\lambda} X \Pi_{\lambda} = \cP_{H}(X) \, ,
\end{align}
which proves the assertion.\footnote{We note that every probability measure whose Fourier transform satisfies Property~\ref{prop_i} and Property~\ref{prop_ii} would work for Lemma~\ref{lem_IntegralPinching}. } \qed
\end{proof}

As mentioned at the beginning of this chapter, the pinching map can be used to modify one Hermitian operator such that it commutes with another Hermitian operator. 
Pinching maps are user-friendly since they fulfill several nice properties. The following lemma summarizes the most important ones.
In Section~\ref{sec_specPinchGT}, we demonstrate how pinching maps can be used to prove the Golden-Thompson inequality (see Theorem~\ref{thm_GT}) in an intuitive and transparent way. 
\index{pinching inequality} \index{properties pinching map} \index{pinching!properties pinching map}
\begin{svgraybox}
\vspace{-4mm}
\begin{lemma}[Properties of pinching map] \label{lem_propertiesPinching}
Let $H \in \Her(A)$. Then
\begin{enumerate}
\item $[\cP_{H}(X),H]=0$ \quad for all $X \in \Her(A)$. \label{prop_lem_pinch} \vspace{1.5mm}
\item $\cP_{H}(X) \geq \frac{1}{|\spec(H)|} X $ \quad for all $X \in \Pos(A)$. \hspace{30mm} (Pinching inequality) \vspace{1.5mm}
\item $\tr\, \cP_{H}(X)H = \tr \, X H $ \quad for all $X \in \Her(A)$.  \vspace{1.5mm}
\item $f(\cP_{H}(X)) \leq \cP_{H}(f(X))$ \quad for all $X \in \Her(A)$ and $f(\cdot)$ operator convex. \vspace{1.5mm}
\item $\normU{\cP_{H}(X)} \leq \normU{X}$ \quad  for all $X \in \Her(A)$ and any unitarily invariant norm $\normU{\cdot}$.
\end{enumerate}
\end{lemma}
\vspace{-4mm}
\end{svgraybox}
\begin{proof}
Since eigenvectors corresponding to distinct eigenvalues of Hermitian operators are orthogonal we find
\begin{align}
\cP_{H}(X)H 
&= \sum_{\lambda,\lambda' \in \spec(H)} \Pi_{\lambda} X \Pi_{\lambda} \lambda' \Pi_{\lambda'}
= \sum_{\lambda \in \spec(H)} \lambda \Pi_{\lambda} X \Pi_{\lambda}\\
&= \sum_{\lambda,\lambda' \in \spec(H)} \lambda' \Pi_{\lambda'} \Pi_{\lambda} X \Pi_{\lambda} 
= H \cP_{H}(X) \, ,
\end{align}
which proves the first statement of the lemma. 

The pinching inequality follows since
\begin{align} \label{eq_hayashiPinchingIneqIntro}
\cP_{H}(X) = \sum_{\lambda \in \spec(H)} \Pi_{\lambda}X \Pi_{\lambda}
=\frac{1}{|\spec(H)|} \sum_{y=1}^{|\spec(H)|} U_y X U_y^{\dagger}
 \geq \frac{1}{|\spec(H)|} X   \, ,
\end{align}
for all $X \in \Pos(A) $, where $\spec(H):=\{\lambda_1,\dots,\lambda_{|\spec(H)|}\}$ and
\begin{align}
U_y:=\sum_{z=1}^{|\spec(H)|} \exp\left(\frac{\ci 2\pi y z}{|\spec(H)|}\right) \Pi_{\lambda_z}
\end{align}
are unitaries and we used the fact that
\begin{align}
\sum_{y=1}^{|\spec(H)|} \exp \left( \frac{\ci 2 \pi y (z-z')}{|\spec(H)|} \right) = |\spec(H)| \indic\{ z = z' \} \, .
\end{align}
The inequality step in~\eqref{eq_hayashiPinchingIneqIntro} follows form the facts that $U_y X U^\dagger_y \geq 0$ and $U_{|\spec(H)|} = \id_A$.

The third property of the lemma follows from the cyclic property of the trace and the fact that $\ee^{\ci t H}$ commutes with $H$ for all $t \in \R$. Lemma~\ref{lem_IntegralPinching} shows that
\begin{align}
\tr\, \cP_{H}(X) H 
= \int_{-\infty}^{\infty} \di t \mu_{\Delta_H}( t) \tr\, \ee^{\ci t H} X \ee^{-\ci t H} H 
= \int_{-\infty}^{\infty} \di t \mu_{\Delta_H}( t) \tr\, X H
= \tr\, X H \, .
\end{align}

The fourth property of the lemma follows form Jensen's operator inequality (see Theorem~\ref{thm_Jensen}) which shows that in case $f$ is operator convex we have
\begin{align}
f\big(\cP_{H}(X)\big)
= f \Big( \sum_{\lambda \in \spec(H)} \Pi_{\lambda} X \Pi_{\lambda} \Big)
\leq \sum_{\lambda \in \spec(H)} \Pi_{\lambda} f(X) \Pi_{\lambda}
= \cP_{H}\big(f(x) \big) \, .
\end{align}

Finally it remains to prove the fifth property of the lemma. Lemma~\ref{lem_IntegralPinching} shows that
\begin{align}
\normU{\int_{-\infty}^{\infty} \di t \mu_{\Delta_H}(t) \ee^{\ci t H} X \ee^{-\ci t H}}
&\leq \int_{-\infty}^{\infty} \di t \mu_{\Delta_H}(t)  \normU{\ee^{\ci t H} X \ee^{-\ci t H}}\\
&= \int_{-\infty}^{\infty} \di t \mu_{\Delta_H}(t)  \normU{X}\\
&= \normU{X} \, ,
\end{align}
where the penultimate step uses that $\ee^{\ci t H}$ is unitary for all $t \in \R$. \qed
\end{proof}

\subsection{Smooth spectral pinching}\index{pinching!smooth pinching}
The pinching map can change an operator considerably. More precisely, there exist Hermitian operators $H_1,H_2 \in \Her(A)$ such that $\cP_{H_2}(H_1)$ is far from $H_1$. To see this let $\delta \in (0,1)$ and consider the following two-dimensional operators $H_1=\proj{0}$ and $H_2=(1-\delta) \frac{\id_2}{2} + \delta \proj{+}$, where $\ket{+}:=\frac{1}{\sqrt{2}}(\ket{0}+\ket{1})$. A simple calculation reveals that $\cP_{H_2}(H_1)=\frac{\id_2}{2}$ and hence $\norm{H_1 - \cP_{H_2}(H_1)}_{\infty} = \frac{1}{2}$ for any $\delta \in (0,1)$. 

We next discuss a smooth version of the pinching method which guarantees that the pinching does not change the operator too much at the cost that Property~\ref{prop_lem_pinch} of Lemma~\ref{lem_propertiesPinching} no longer holds. 
\begin{svgraybox}
\vspace{-4mm}
\begin{definition}
Let $H \in \Her(A)$ with a spectral decomposition given in~\eqref{eq_spectral_dec} and $\kappa >0$. The \emph{$\kappa$-smooth pinching map} with respect to $H$ is defined as 
\begin{align}
\cP^{\kappa}_{H}\, : \,  \Her(A) \ni X \mapsto  \int_{-\infty}^{\infty} \di t \mu_{\kappa}(t) \ee^{\ci t H} X \ee^{-\ci t H}  \, ,
\end{align}
with probability density $\mu_{\kappa}$ defined in~\eqref{eq_spectral_dec}.
\end{definition}
\vspace{-4mm}
\end{svgraybox}
For any $\kappa\leq \Delta_H$ the $\kappa$-smooth pinching map coincides with the regular pinching map given in Definition~\ref{def_pinching}. This can be easily seen from the proof of Lemma~\ref{lem_IntegralPinching}.
As a result, whenever $\kappa \leq \Delta_H$, we write $\cP_H$ instead of $\cP^\kappa_H$. The $\kappa$-smooth pinching map fulfills several nice properties that are summarized in the following lemma.
\begin{svgraybox}
\vspace{-4mm}
\begin{lemma}[Properties of smooth pinching map] \label{lem_hastings}
Let $\kappa>0$, $H,X \in \Her(A)$, and $\normU{\cdot}$ a unitarily invariant norm. Then
\begin{enumerate}
\item $\normU{[H,\cP_{H}^{\kappa}(X)]} \leq \normU{[H,X]}$. \label{prop:1} \vspace{1.5mm}
\item $\normU{[H,\cP_{H}^{\kappa}(X)]}  \leq \kappa \normU{X} \indic\{ \kappa > \Delta_H \} $. \label{prop:11}\vspace{1.5mm}
\item Let $\ket{h}$, $\ket{h'}$ be eigenvectors of $H$ with corresponding eigenvalues $h$, $h'$ such that $|h - h'| \geq \kappa$. Then, $\bra{h} \cP_{H}^{\kappa}(X) \ket{h'} =0$. \label{prop:2}\vspace{1.5mm}
\item $\norm{X- \cP_{H}^{\kappa}(X)}_{\infty} \leq  \norm{[H,X]}_{\infty} \frac{12 \log 2}{ \pi \kappa}$. \label{prop:3}\vspace{1.5mm}
\item $\normU{\cP^{\kappa}_{H}(X)} \leq \normU{X}$. \label{prop:4}
\end{enumerate}
\end{lemma}
\vspace{-4mm}
\end{svgraybox}
Properties~\ref{prop:11} and~\ref{prop:3} suggest that there is a tradeoff between reducing the commutator to zero (by choosing $\kappa \leq \Delta_H$) and increasing the distance between $X$ and $\cP_{H}^{\kappa}(X)$.
Before proving the lemma we state a technical result that is used in the proof, and which shows that the complex matrix exponential is operator Lipschitz continuous.
\begin{lemma} \label{lem_lipschitz}
Let $L \in \LL(A)$, $H \in \Her(A)$ and $t \in \R$. Then
\begin{align}
\norm{[L,\ee^{\ci t H}]}_{\infty} \leq |t| \norm{[L,H]}_{\infty} \, .
\end{align}
\end{lemma}
\begin{proof}
Since $H$ is Hermitian it can be decomposed into $H=U\Lambda U^\dagger$, where $\Lambda$ is a diagonal matrices containing the eigenvalues of $H$ and $U$ is a unitary matrix whose rows consist of the eigenvectors of $H$. Since the operator norm is unitarily invariant we obtain
\begin{align}
\norm{[L,\ee^{\ci t H}]}_{\infty} 
= \norm{[U^\dagger L U, \ee^{\ci t \Lambda}]}_{\infty}
\leq |t| \norm{[U^\dagger L U,\Lambda]}_{\infty}
= |t| \norm{[ L , U \Lambda U^\dagger]}_{\infty}
= |t| \norm{[ L , H]}_{\infty} \, ,
\end{align}
where the inequality step uses the fact that the function $f: x \mapsto \ee^{\ci t x}$ is Lipschitz continuous with constant $|t|$ and the fact that $\Lambda$ is diagonal. As a result $\Lambda \mapsto \ee^{\ci t \Lambda}$ is operator Lipschitz continuous on the set of diagonal matrices with constant $|t|$. Theorem~3.1~in~\cite{aleksandrov12} then implies the assertion. \qed
\end{proof}

\begin{proof}[Lemma~\ref{lem_hastings}]
Since $H$ and $X$ are Hermitian and $\mu_{\kappa}$ is an even function it follows that $\cP_{H}^{\kappa}(X)$ is Hermitian.
By using the triangle inequality and the fact that $\ee^{\ci t H }$ commutes with $H$, we find
\begin{align}
\normU{[H, \cP_{H}^{\kappa}(X)]} 
&= \normU{[H,\int_{-\infty}^\infty \di t \mu_{\kappa}(t) \ee^{\ci t H} X \ee^{-\ci t H }]} \\
&\leq \int_{-\infty}^\infty \di t \mu_{\kappa}( t) \normU{[H,\ee^{\ci t H} X \ee^{-\ci tH }]} \\
&= \int_{-\infty}^\infty \di t \mu_{\kappa}(t) \normU{[H,X]} \\
& = \normU{[H,X]} \, ,
\end{align}
which proves Property~\ref{prop:1} of the lemma.

We next prove Property~\ref{prop:11} of the lemma. Note that in case $\kappa \leq \Delta_H$ we have a perfect pinching and hence $[H,\cP_{H}^{\kappa}(X)]=0$. For $\kappa > \Delta_H$ we find
\begin{align}
\normU{[H, \cP_{H}^{\kappa}(X)]} 
&= \normU{[H,\int_{-\infty}^\infty \di t \mu_{\kappa}(t) \ee^{\ci t H} X \ee^{-\ci t H }]} \\
&= \normU{\sum_{\ell,n} \ket{\ell} \bra{n} (\lambda_{\ell} - \lambda_n) \bra{\ell}X\ket{n} \hat \mu_{\kappa}(\lambda_{\ell} - \lambda_n)}
\end{align}
where we expressed the term inside the norm in the eigenbasis of $H$. Properties~\ref{prop_ii} and~\ref{prop_v} of $\hat \mu_{\kappa}$ now imply that
\begin{align}
\normU{[H, \cP_{H}^{\kappa}(X)]} 
\leq \kappa \normU{\sum_{\ell,n} \ket{\ell} \bra{n}  \bra{\ell}X\ket{n} } = \kappa \normU{X}\, .
\end{align}

We next prove Property~\ref{prop:2} of the lemma. Let $\ket{h}$ and $\ket{h'}$ be two eigenvectors of $H$ such that the corresponding eigenvalues $h$ and $h'$ satisfy $|h- h'|\geq \kappa$. By definition of the Fourier transform together with Property~\ref{prop_ii} of $\hat \mu_{\kappa}$, mentioned at the beginning of this chapter, we find
\begin{align} \label{eq_midStepDa}
0=\hat \mu_{\kappa}(h'-h) = \int_{-\infty}^{\infty} \di t \mu_{\kappa}( t) \ee^{\ci t(h-h')} \, .
\end{align}
This can be used to show that Property~\ref{prop:2} of the lemma indeed holds. By definition of the $\kappa$-smooth pinching map, we have
\begin{align}
\bra{h} \cP_{H}^{\kappa}(X) \ket{h'}
 = \int_{-\infty}^{\infty} \di t \mu_{\kappa}(t) \bra{h} \ee^{\ci t H} X \ee^{-\ci t H} \ket{h'} 
= \bra{h} X \ket{h'} \int_{-\infty}^{\infty} \di t \mu_{\kappa}(t)   \ee^{\ci t(h - h')} 
= 0 \, ,
\end{align}
where the final step follows from~\eqref{eq_midStepDa}.

We next prove Property~\ref{prop:3} of the lemma. The triangle inequality together with the fact that the operator norm is unitarily invariant give
\begin{align}
\norm{X-\cP_{H}^{\kappa}(X)}_{\infty} 
&\leq \int_{-\infty}^{\infty} \di t \mu_{\kappa}(t) \norm{X - \ee^{\ci t H} X \ee^{-\ci t H}}_{\infty} 
 = \int_{-\infty}^{\infty} \di t \mu_{\kappa}(t) \norm{[X,\ee^{\ci t H}]}_{\infty} \, .
 \end{align}
 Lemma~\ref{lem_lipschitz} then implies that
\begin{align}
 \int_{-\infty}^{\infty} \di t \mu_{\kappa}(t) \norm{[X,\ee^{\ci t H}]}_{\infty} 
 \leq \int_{-\infty}^{\infty} \di t  \mu_{\kappa}(t) |t| \norm{[X,H]}_{\infty}  
=\norm{[X,H]}_{\infty} \frac{12 \log 2}{ \pi \kappa} \, .
\end{align}

It thus remains to prove Property~\ref{prop:4} of the lemma. By the triangle inequality we have
\begin{align}
\normU{\cP^{\kappa}_H(X)} \leq \int_{-\infty}^{\infty} \di t \mu_{\kappa}(t) \normU{\ee^{\ci t H} X \ee^{-\ci t H}} = \normU{X} \, ,
\end{align}
which thus completes the proof. \qed
\end{proof}  

\subsection{Asymptotic spectral pinching}\index{pinching!asymptotic pinching} \label{sec_asymptoticPinching}
The spectral pinching method explained in Section~\ref{subsec_pinching} is particularly powerful if we apply it in an asymptotic setting. To understand what we mean by that let us first recall two basic statements (given by Remark~\ref{remark_types} and Exercise~\ref{exercise_tensorProduct}).

\begin{remark} \label{remark_types}
Let $B \in \Pos(A)$. The number of distinct eigenvalues of $B^{\otimes m}$, i.e., $|\spec(B^{\otimes m})|$ grows polynomially in $m$. This is due to the fact that the number of distinct eigenvalues of $B^{\otimes m}$ is bounded by the number of different types of sequences of $\dim(A)$ symbols of length $m$, a concept widely used in information theory~\cite{cover}. More precisely~\cite[Lemma~II.1]{csiszar98} gives
\begin{align} \label{eq_types}
|\spec(B^{\otimes m})| \leq 
\left(\begin{matrix} m + \dim(A) - 1 \\ \dim(A) - 1 \end{matrix}
\right) 
\leq \frac{(m+\dim(A)-1)^{\dim(A)-1}}{(\dim(A)-1)!}
 &\leq (m+1)^{\dim(A)-1}\\
&= O\bigl(\poly(m)\bigr) \ ,
\end{align} 
where $\poly(m)$ denotes a polynomial in $m$.
\end{remark}

\begin{exercise} \label{exercise_tensorProduct}
Let $L_1 \in \LL(A)$, $L_2 \in \LL(B)$ and $C_1 \in \Pos(A), C_2 \in \Pos(B)$. Verify the following identities for the tensor product:
\begin{enumerate}
\item $\tr \, L_1 \otimes L_2 = (\tr\, L_1 ) (\tr\, L_2)$.\vspace{1mm}
\item $\log C_1 \otimes C_2 = (\log C_1) \otimes \id_B + \id_A \otimes (\log C_2)$.\vspace{1mm}
\item $\exp(L_1) \otimes \exp(L_2) = \exp(L_1 \otimes \id_B + \id_A \otimes L_2)$.
\end{enumerate}
\end{exercise}

With this preliminary knowledge in mind let us explain what we mean by the asymptotic spectral pinching method. We apply this technique to prove a famous trace inequality --- the so-called \emph{Golden-Thompson (GT) inequality} which states that any two Hermitian operators $H_1, H_2 \in \Her(A)$ satisfy
\begin{align} \label{eq_GT_pin}
\tr\, \ee^{H_1 + H_2} \leq \tr\, \ee^{H_1} \ee^{H_2} \, .
\end{align}
We refer to Theorem~\ref{thm_GT} and the subsequent paragraph for more details about this inequality. We next present a proof of the GT inequality based on the asymptotic spectral pinching method.

\subsubsection{An intuitive proof of the Golden-Thompson inequality} \label{sec_specPinchGT}
Let $B_1,B_2 \in \Pos(A)$ be such that $B_1=\exp(H_1)$ and $B_2 = \exp(H_2)$. The identities for the tensor product of the exponential, logarithm and trace function given in Exercise~\ref{exercise_tensorProduct} show that
\begin{align}
\log \tr \exp(\log B_1 + \log B_2)
&= \frac{1}{m} \log \tr \exp\bigl( \log B_1^{\otimes m} + \log B_2^{\otimes m} \bigr) \label{eq_toypinching0} \\
&\leq \frac{1}{m} \log \tr \exp\left( \log \cP_{B_2^{\otimes m}}(B_1^{\otimes m}) + \log B_2^{\otimes m} \right)  + \frac{\log \poly(m)}{m} \label{eq_toypinching1} \\
&=\frac{1}{m}  \log \tr \, \cP_{B_2^{\otimes m}}(B_1^{\otimes m})B_2^{\otimes m}  + \frac{\log \poly(m)}{m} \label{eq_toypinching2}  \\
&= \log \tr \, B_1 B_2 + \frac{\log \poly(m)}{m} \label{eq_toypinching3}
\ ,
\end{align}
where~\eqref{eq_toypinching1} follows by the pinching inequality (see Lemma~\ref{lem_propertiesPinching}), together with the fact that the logarithm is operator monotone (see Table~\ref{table_opConvex}) and $H \mapsto \tr \exp H$ is monotone (see Proposition~\ref{prop_traceFunctions}). Furthermore we use the observation presented in Remark~\ref{remark_types}, i.e., that the number of distinct eigenvalues of $B_2^{\otimes m}$ grows polynomially in $m$.
Equality~\eqref{eq_toypinching2} uses Lemma~\ref{lem_propertiesPinching} which ensures that $\cP_{B_2^{\otimes m}}(B_1^{\otimes m})$ commutes with $B_2^{\otimes m}$ and hence $\log \cP_{B_2^{\otimes m}}(B_1^{\otimes m}) + \log B_2^{\otimes m} = \log \cP_{B_2^{\otimes m}}(B_1^{\otimes m})  B_2^{\otimes m}$. Equality~\eqref{eq_toypinching3} uses again Lemma~\ref{lem_propertiesPinching} and the properties of the exponential, logarithm and trace function under the tensor product given by Exercise~\ref{exercise_tensorProduct}.
Considering the limit $m\to \infty$ finally implies the GT inequality~\eqref{eq_GT_pin}. \qed

We believe that the proof of the GT inequality presented above is intuitive and transparent. The high-level intuition may be summarized as follows: We know that the GT inequality is trivial if the operators commute. The spectral pinching method forces our operators to commute. At the same time the pinching should hopefully not destroy the operator which it acts on too much. This is indeed the case (guaranteed by the pinching inequality) if we lift our problem to high dimensions, i.e., if we consider an $m$-fold tensor product of our operators and the limit $m \to \infty$.\footnote{This phenomenon is known as the \emph{tensor power trick} and is described, e.g., in~\cite{tao_blog}.}

\section{Complex interpolation theory} \label{sec_interpolation} 
Consider a sufficiently well-behaved holomorphic function defined on the strip $S:=\{ z \in \C : 0 \leq \Real\, z \leq 1\}$.
Complex interpolation theory allows us to control the behavior of the function at $(0,1)$ by its value on the boundary, i.e., at $\Real \, z = 0$ and $\Real \, z = 1$. 
Complex interpolation theory is an established technique that is vast and extensive. In this section we review a specific interpolation theorem for Schatten norms, commonly attributed to Stein~\cite{stein56}, and based on Hirschman's improvement of the Hadamard three-lines theorem~\cite{H52}. In Chapter~\ref{chapter_traceIneq} we will use this interpolation result to prove multivariate extensions of known trace inequalities. 

Before stating the main result let us define a family of probability densities on $\R$
\begin{align} 
\beta_{\theta} (t) := \frac{\sin(\pi\theta)}{2\theta\bigl( \cosh(\pi t)+\cos(\pi\theta)\bigr)} \quad \text{for} \quad \theta \in (0,1) \, . \label{eq_densi}
\end{align}
These densities are depicted in Figure~\ref{fig_betaT}. Furthermore, the following limits hold:
\begin{align}\label{eq_beta_0}
\beta_{0}(t):=\lim_{\theta \searrow 0}\beta_{\theta}(\di t)=\frac{\pi}{2}\bigl(  \cosh(\pi
t)+1\bigr)^{-1} 
\end{align}
and
\begin{align}
\beta_{1}(t):=\lim_{\theta\nearrow 1}\beta_{\theta}( t)=\delta(t) \,.
\end{align}
Here $\beta_{0}$ is another probability density on $\R$ and $\delta$ denotes the Dirac $\delta$-distribution. 
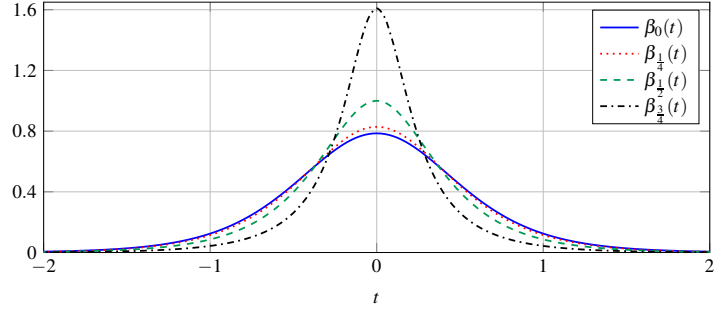
\begin{figure}[!tb]
\centering
\scalebox{0.85}{  \begin{tikzpicture}
	\begin{axis}[
		height=5.5cm,
		width=12cm,
		grid=major,
		xlabel=$t$,
		xmin=-2,
		xmax=2,
		ymax=1.65,
		ymin=0,
	     xtick={-2,-1,0,1,2},
          ytick={1.6,1.2,0.8,0.4,0},
		legend style={at={(0.904,0.97)},anchor=north,legend cell align=left,font=\footnotesize} 
	]



	\addplot[blue,thick,smooth] coordinates {
(-3.,0.000253484) (-2.95,0.000296591) (-2.9,0.000347026) (-2.85,0.000406036) (-2.8,0.000475078) (-2.75,0.000555855) (-2.7,0.00065036) (-2.65,0.000760925) (-2.6,0.000890277) (-2.55,0.0010416) (-2.5,0.00121863) (-2.45,0.00142572) (-2.4,0.00166796) (-2.35,0.00195131) (-2.3,0.00228272) (-2.25,0.00267032) (-2.2,0.00312361) (-2.15,0.00365367) (-2.1,0.00427342) (-2.05,0.00499796) (-2.,0.00584489) (-1.95,0.00683471) (-1.9,0.00799129) (-1.85,0.00934241) (-1.8,0.0109204) (-1.75,0.0127627) (-1.7,0.0149127) (-1.65,0.0174209) (-1.6,0.0203454) (-1.55,0.0237533) (-1.5,0.0277215) (-1.45,0.0323386) (-1.4,0.0377056) (-1.35,0.0439371) (-1.3,0.051163) (-1.25,0.0595295) (-1.2,0.0691992) (-1.15,0.0803521) (-1.1,0.0931845) (-1.05,0.107908) (-1.,0.124746) (-0.95,0.143929) (-0.9,0.165685) (-0.85,0.190235) (-0.8,0.217769) (-0.75,0.248436) (-0.7,0.282316) (-0.65,0.319396) (-0.6,0.359536) (-0.55,0.402441) (-0.5,0.447625) (-0.45,0.494393) (-0.4,0.541823) (-0.35,0.58877) (-0.3,0.633898) (-0.25,0.675725) (-0.2,0.712712) (-0.15,0.74336) (-0.1,0.766334) (-0.05,0.780573) (0.,0.785398) (0.05,0.780573) (0.1,0.766334) (0.15,0.74336) (0.2,0.712712) (0.25,0.675725) (0.3,0.633898) (0.35,0.58877) (0.4,0.541823) (0.45,0.494393) (0.5,0.447625) (0.55,0.402441) (0.6,0.359536) (0.65,0.319396) (0.7,0.282316) (0.75,0.248436) (0.8,0.217769) (0.85,0.190235) (0.9,0.165685) (0.95,0.143929) (1.,0.124746) (1.05,0.107908) (1.1,0.0931845) (1.15,0.0803521) (1.2,0.0691992) (1.25,0.0595295) (1.3,0.051163) (1.35,0.0439371) (1.4,0.0377056) (1.45,0.0323386) (1.5,0.0277215) (1.55,0.0237533) (1.6,0.0203454) (1.65,0.0174209) (1.7,0.0149127) (1.75,0.0127627) (1.8,0.0109204) (1.85,0.00934241) (1.9,0.00799129) (1.95,0.00683471) (2.,0.00584489) (2.05,0.00499796) (2.1,0.00427342) (2.15,0.00365367) (2.2,0.00312361) (2.25,0.00267032) (2.3,0.00228272) (2.35,0.00195131) (2.4,0.00166796) (2.45,0.00142572) (2.5,0.00121863) (2.55,0.0010416) (2.6,0.000890277) (2.65,0.000760925) (2.7,0.00065036) (2.75,0.000555855) (2.8,0.000475078) (2.85,0.000406036) (2.9,0.000347026) (2.95,0.000296591) (3.,0.000253484)
	};
	\addlegendentry{$\beta_0(t)$}
	
	\addplot[red,thick,smooth,dotted] coordinates {
(-3.,0.000228227) (-2.95,0.00026704) (-2.9,0.000312454) (-2.85,0.000365589) (-2.8,0.000427758) (-2.75,0.000500497) (-2.7,0.000585601) (-2.65,0.000685171) (-2.6,0.000801664) (-2.55,0.000937954) (-2.5,0.0010974) (-2.45,0.00128394) (-2.4,0.00150216) (-2.35,0.00175743) (-2.3,0.00205604) (-2.25,0.00240533) (-2.2,0.00281388) (-2.15,0.0032917) (-2.1,0.0038505) (-2.05,0.00450395) (-2.,0.00526799) (-1.95,0.00616125) (-1.9,0.00720543) (-1.85,0.0084258) (-1.8,0.00985185) (-1.75,0.0115178) (-1.7,0.0134636) (-1.65,0.0157355) (-1.6,0.0183871) (-1.55,0.0214806) (-1.5,0.0250878) (-1.45,0.0292916) (-1.4,0.0341873) (-1.35,0.039884) (-1.3,0.0465066) (-1.25,0.054197) (-1.2,0.0631155) (-1.15,0.0734426) (-1.1,0.0853791) (-1.05,0.0991463) (-1.,0.114985) (-0.95,0.133155) (-0.9,0.153925) (-0.85,0.17757) (-0.8,0.204359) (-0.75,0.234535) (-0.7,0.268297) (-0.65,0.305767) (-0.6,0.346956) (-0.55,0.391718) (-0.5,0.439704) (-0.45,0.490309) (-0.4,0.542633) (-0.35,0.59545) (-0.3,0.647207) (-0.25,0.696069) (-0.2,0.740008) (-0.15,0.776951) (-0.1,0.804966) (-0.05,0.822471) (0.,0.828427) (0.05,0.822471) (0.1,0.804966) (0.15,0.776951) (0.2,0.740008) (0.25,0.696069) (0.3,0.647207) (0.35,0.59545) (0.4,0.542633) (0.45,0.490309) (0.5,0.439704) (0.55,0.391718) (0.6,0.346956) (0.65,0.305767) (0.7,0.268297) (0.75,0.234535) (0.8,0.204359) (0.85,0.17757) (0.9,0.153925) (0.95,0.133155) (1.,0.114985) (1.05,0.0991463) (1.1,0.0853791) (1.15,0.0734426) (1.2,0.0631155) (1.25,0.054197) (1.3,0.0465066) (1.35,0.039884) (1.4,0.0341873) (1.45,0.0292916) (1.5,0.0250878) (1.55,0.0214806) (1.6,0.0183871) (1.65,0.0157355) (1.7,0.0134636) (1.75,0.0115178) (1.8,0.00985185) (1.85,0.0084258) (1.9,0.00720543) (1.95,0.00616125) (2.,0.00526799) (2.05,0.00450395) (2.1,0.0038505) (2.15,0.0032917) (2.2,0.00281388) (2.25,0.00240533) (2.3,0.00205604) (2.35,0.00175743) (2.4,0.00150216) (2.45,0.00128394) (2.5,0.0010974) (2.55,0.000937954) (2.6,0.000801664) (2.65,0.000685171) (2.7,0.000585601) (2.75,0.000500497) (2.8,0.000427758) (2.85,0.000365589) (2.9,0.000312454) (2.95,0.00026704) (3.,0.000228227)
	};
	\addlegendentry{$\beta_{\frac{1}{4}}(t)$}

	\addplot[ForestGreen,thick,smooth,dashed] coordinates {
(-3.,0.000161399) (-2.95,0.000188851) (-2.9,0.000220973) (-2.85,0.000258558) (-2.8,0.000302535) (-2.75,0.000353993) (-2.7,0.000414204) (-2.65,0.000484655) (-2.6,0.000567089) (-2.55,0.000663545) (-2.5,0.000776406) (-2.45,0.000908464) (-2.4,0.00106298) (-2.35,0.00124379) (-2.3,0.00145534) (-2.25,0.00170288) (-2.2,0.00199251) (-2.15,0.00233142) (-2.1,0.00272797) (-2.05,0.00319196) (-2.,0.00373487) (-1.95,0.00437013) (-1.9,0.00511343) (-1.85,0.00598315) (-1.8,0.00700079) (-1.75,0.00819151) (-1.7,0.00958474) (-1.65,0.0112149) (-1.6,0.0131223) (-1.55,0.015354) (-1.5,0.0179651) (-1.45,0.0210202) (-1.4,0.0245945) (-1.35,0.0287761) (-1.3,0.033668) (-1.25,0.0393905) (-1.2,0.0460837) (-1.15,0.0539115) (-1.1,0.0630643) (-1.05,0.0737637) (-1.,0.0862667) (-0.95,0.10087) (-0.9,0.117916) (-0.85,0.137795) (-0.8,0.160949) (-0.75,0.187873) (-0.7,0.219108) (-0.65,0.255231) (-0.6,0.296828) (-0.55,0.344451) (-0.5,0.398537) (-0.45,0.459301) (-0.4,0.526566) (-0.35,0.599546) (-0.3,0.676591) (-0.25,0.75494) (-0.2,0.830584) (-0.15,0.898389) (-0.1,0.952603) (-0.05,0.987789) (0.,1.) (0.05,0.987789) (0.1,0.952603) (0.15,0.898389) (0.2,0.830584) (0.25,0.75494) (0.3,0.676591) (0.35,0.599546) (0.4,0.526566) (0.45,0.459301) (0.5,0.398537) (0.55,0.344451) (0.6,0.296828) (0.65,0.255231) (0.7,0.219108) (0.75,0.187873) (0.8,0.160949) (0.85,0.137795) (0.9,0.117916) (0.95,0.10087) (1.,0.0862667) (1.05,0.0737637) (1.1,0.0630643) (1.15,0.0539115) (1.2,0.0460837) (1.25,0.0393905) (1.3,0.033668) (1.35,0.0287761) (1.4,0.0245945) (1.45,0.0210202) (1.5,0.0179651) (1.55,0.015354) (1.6,0.0131223) (1.65,0.0112149) (1.7,0.00958474) (1.75,0.00819151) (1.8,0.00700079) (1.85,0.00598315) (1.9,0.00511343) (1.95,0.00437013) (2.,0.00373487) (2.05,0.00319196) (2.1,0.00272797) (2.15,0.00233142) (2.2,0.00199251) (2.25,0.00170288) (2.3,0.00145534) (2.35,0.00124379) (2.4,0.00106298) (2.45,0.000908464) (2.5,0.000776406) (2.55,0.000663545) (2.6,0.000567089) (2.65,0.000484655) (2.7,0.000414204) (2.75,0.000353993) (2.8,0.000302535) (2.85,0.000258558) (2.9,0.000220973) (2.95,0.000188851) (3.,0.000161399)
	};
	\addlegendentry{$\beta_{\frac{1}{2}}(t)$}

	\addplot[black,thick,smooth,dashdotted] coordinates {
(-3.,0.0000760929) (-2.95,0.0000890372) (-2.9,0.000104184) (-2.85,0.000121908) (-2.8,0.000142647) (-2.75,0.000166916) (-2.7,0.000195315) (-2.65,0.000228547) (-2.6,0.000267436) (-2.55,0.000312945) (-2.5,0.000366202) (-2.45,0.000428529) (-2.4,0.000501472) (-2.35,0.000586842) (-2.3,0.00068676) (-2.25,0.000803711) (-2.2,0.000940606) (-2.15,0.00110086) (-2.1,0.00128846) (-2.05,0.00150811) (-2.,0.0017653) (-1.95,0.00206648) (-1.9,0.00241924) (-1.85,0.00283247) (-1.8,0.00331662) (-1.75,0.00388401) (-1.7,0.00454912) (-1.65,0.00532901) (-1.6,0.00624384) (-1.55,0.00731738) (-1.5,0.00857781) (-1.45,0.0100585) (-1.4,0.0117991) (-1.35,0.0138469) (-1.3,0.0162583) (-1.25,0.0191009) (-1.2,0.0224558) (-1.15,0.0264213) (-1.1,0.0311164) (-1.05,0.0366861) (-1.,0.0433083) (-0.95,0.0512028) (-0.9,0.0606426) (-0.85,0.0719696) (-0.8,0.0856159) (-0.75,0.102132) (-0.7,0.122225) (-0.65,0.146813) (-0.6,0.177097) (-0.55,0.214659) (-0.5,0.26159) (-0.45,0.320658) (-0.4,0.395477) (-0.35,0.490627) (-0.3,0.611507) (-0.25,0.763405) (-0.2,0.948757) (-0.15,1.16111) (-0.1,1.37577) (-0.05,1.54429) (0.,1.60948) (0.05,1.54429) (0.1,1.37577) (0.15,1.16111) (0.2,0.948757) (0.25,0.763405) (0.3,0.611507) (0.35,0.490627) (0.4,0.395477) (0.45,0.320658) (0.5,0.26159) (0.55,0.214659) (0.6,0.177097) (0.65,0.146813) (0.7,0.122225) (0.75,0.102132) (0.8,0.0856159) (0.85,0.0719696) (0.9,0.0606426) (0.95,0.0512028) (1.,0.0433083) (1.05,0.0366861) (1.1,0.0311164) (1.15,0.0264213) (1.2,0.0224558) (1.25,0.0191009) (1.3,0.0162583) (1.35,0.0138469) (1.4,0.0117991) (1.45,0.0100585) (1.5,0.00857781) (1.55,0.00731738) (1.6,0.00624384) (1.65,0.00532901) (1.7,0.00454912) (1.75,0.00388401) (1.8,0.00331662) (1.85,0.00283247) (1.9,0.00241924) (1.95,0.00206648) (2.,0.0017653) (2.05,0.00150811) (2.1,0.00128846) (2.15,0.00110086) (2.2,0.000940606) (2.25,0.000803711) (2.3,0.00068676) (2.35,0.000586842) (2.4,0.000501472) (2.45,0.000428529) (2.5,0.000366202) (2.55,0.000312945) (2.6,0.000267436) (2.65,0.000228547) (2.7,0.000195315) (2.75,0.000166916) (2.8,0.000142647) (2.85,0.000121908) (2.9,0.000104184) (2.95,0.0000890372) (3.,0.0000760929)
	};
	\addlegendentry{$\beta_{\frac{3}{4}}(t)$}


	\end{axis}  

\end{tikzpicture}}
\caption{This plot depicts the probability density $\beta_\theta$ defined in~\eqref{eq_densi} for $\theta \in \{0,\frac{1}{4},\frac{1}{2},\frac{3}{4}\}$.}
\label{fig_betaT}
\end{figure}

\begin{svgraybox}
\vspace{-4mm}
\begin{theorem}[Stein-Hirschman]\label{thm_hirschman} \index{complex interpolation theory} \index{complex interpolation theory!Stein-Hirschman}
Let $p_{0},p_{1}\in [1,\infty]$, $\theta\in(0,1)$, $\beta_{\theta}$ given in~\eqref{eq_densi}, define $p_{\theta}$ by $\frac{1}{p_{\theta}}=\frac{1-\theta}{p_{0}}+\frac{\theta}{p_{1}}$, and $S:=\left\{  z\in\mathbb{C}:0\leq\Real\, z\leq1\right\} $. 
For any function $F: S \to \LL(A)$ that is holomorphic in $\interior(S)$, continuous on $\partial S$, and $z \mapsto \norm{F(z)}_{p_{\operatorname{Re}\,z}}$ is uniformly bounded on $S$ we have
\begin{align}
\log \left\Vert F(\theta)\right\Vert _{p_{\theta}}  \leq
\int_{-\infty}^{\infty} \di t \Bigl(  \beta_{1-\theta}( t)\log   \left\Vert
F(\ci t)\right\Vert _{p_{0}}^{1-\theta} +\beta_{\theta}( t)\log
\left\Vert F(1+\ci t)\right\Vert _{p_{1}}^{\theta} \Bigr) \, .\label{eq:oper-hirschman}
\end{align}
\end{theorem}
\vspace{-4mm}
\end{svgraybox}
We note that the assumption that $z \mapsto \norm{F(z)}_{p_{\operatorname{Re}\,z}}$ is uniformly bounded on $S$ can be relaxed to  
\begin{align}
\sup_{z \in S} \exp(-\alpha \operatorname{Im}\, z) \log \norm{F(z)}_{p_{\operatorname{Re}\,z}} \leq \gamma \quad  \text{for some constants} \quad  \alpha<\pi \quad \text{and} \quad \gamma<\infty \, .
\end{align}
In order to prove Theorem~\ref{thm_hirschman} we first recall Hirschman's strengthening~\cite{H52} (see also~\cite[Lemma~1.3.8]{G08}) of Hadamard's three line theorem. \index{complex interpolation theory!Hadamard's three line theorem}
\begin{lemma}[Hirschman]\index{complex interpolation theory!Hirschman}
\label{lem_hirschman} Let $S:=\left\{  z\in\mathbb{C}:0\leq\Real\,z \leq1\right\}$ and let $f(z)$ be holomorphic on $\interior(S)$, continuous on $\partial S$ and uniformly bounded on $S$.
Then for $\theta\in(0,1)$ and $\beta_{\theta}$ given in~\eqref{eq_densi}, we have
\begin{equation}
\log\left\vert f(\theta)\right\vert \leq\int_{-\infty}^{\infty} \di t \left(
\beta_{1-\theta}(t)\log  \left\vert f(\ci t)\right\vert ^{1-\theta}
+\beta_{\theta}(t)\log  \left\vert f(1+\ci t)\right\vert ^{\theta} \right)
 \,.
\end{equation}
\end{lemma}
We note that the assumption that the function is uniformly bounded in the lemma just above can be relaxed to
\begin{align} \label{eq_HirschweakAss}
\sup_{z \in S} \exp\big(-\alpha | \Immag\,z| \big) \log |f(z)| \leq \gamma  \quad \text{for some constants} \quad  \alpha<\pi \quad \text{and} \quad \gamma<\infty \, .
\end{align}
\begin{proof}
We start by recalling \emph{Poisson's integral formula}~\cite[p.~258]{rudin1987} which ensures that any harmonic function\footnote{A function $f:X \to \R$ where $X$ is an open subset of $\R^n$ is called \emph{harmonic} if it is twice continuously differentiable and satisfies the Laplace equation everywhere on $X$, i.e., $\Delta f =0$.} $u$ defined on the unit disk $D=\{ z \in \C : |z| <1 \}$ can written as
\begin{align} \label{eq_hirsch1}
u(z) = \frac{1}{2 \pi} \int_{-\pi}^{\pi} \di \varphi \, u(q \ee^{\ci \varphi}) \frac{q^2 - r^2 }{|q \ee^{\ci \varphi} - r \ee^{\ci \phi}|} \quad \text{where} \quad z=r \ee^{\ci \phi}, \, \, r<q<1 \, .
\end{align}
Consider a subharmonic function\footnote{A function $f:X \to \R$ where $X$ is an open subset of $\R^n$ is called \emph{subharmonic} if it is twice continuously differentiable and satisfies $\Delta f \geq 0$.} $v$ on $D$ that is continuous on the circle $|\xi|=q <1$ and coincides with $u$ on the circle. In case $u=v$ on the circle $|\xi|$, the right-hand side of~\eqref{eq_hirsch1} defines a harmonic function on $\{z \in \C : |z|<q \}$ that coincides with $v$ on the circle $|\xi|=q$. Since subharmonic functions obey the maximum priciple~\cite[p.~362]{rudin1987} we find for $|z|<q<1$ 
\begin{align} \label{eq_hirschmax}
v(z) \leq \frac{1}{2 \pi} \int_{-\pi}^{\pi} \di \varphi \, u(q \ee^{\ci \varphi}) \frac{q^2-r^2}{|q \ee^{\ci \varphi} - r \ee^{\ci \phi}|} \quad \text{where} \quad z=r \ee^{\ci \phi} \, .
\end{align}
This is valid for all subharmonic functions on $D$ that are continuous on the circle $|\xi|=q$ for $r <q<1$.

We note that
\begin{align}
D \ni \xi \mapsto g(\xi):= \frac{1}{\pi \ci } \log \left(\ci \frac{1 + \xi}{1-\xi} \right) \in (0,1) \times \ci \R 
\end{align}
is a conformal map. Since $f\circ g$ is a holomorphic function on $D$ we know that $\log |f \circ g|$ is a subharmonic function on $D$. Applying the maximum principle (see~\eqref{eq_hirschmax}) yields for $|z|=r <q$ 
\begin{align} \label{eq_hirschm2}
\log |(f\circ g)(z)| \leq \frac{1}{2 \pi} \int_{-\pi}^{\pi} \di \varphi \, \log |(f\circ g)(q \ee^{\ci \varphi})| \frac{q^2 -r^2}{q^2 - 2 r q \cos(\phi - \varphi) + r^2} \, .
\end{align}
where $z=r \ee^{\ci \varphi}$. In case $|\xi|=1$ and $\xi \ne \pm 1$ we have $\Real\,g(\xi) \in \{0,1Ê\}$. By assumption of the lemma (see~\eqref{eq_HirschweakAss}) we have
\begin{align}
\log | (f \circ g)(\xi)| \leq \gamma\,  \ee^{\alpha |\Immag\, h(\xi)|} 
= \gamma\, \ee^{\alpha | \Immag \frac{1}{\pi \ci} \log\left(\frac{1+\xi}{1-\xi} \right)|}
\leq \gamma\, \ee^{\frac{\alpha}{\pi} |\log |\frac{1+\xi}{1-\xi}||} \, .
\end{align}
This shows that $\log | (f \circ g)(\xi)|$ is bounded by a multiple of $|1+\xi|^{-\frac{\alpha}{\pi}}+|1-\xi|^{-\frac{\alpha}{\pi}}$ which is integrable of the set $|\xi|=1$ as $\alpha <\pi$. Let $z=r \ee^{\ci \phi}$ with $r<q$ and consider $q\to 1$ in~\eqref{eq_hirschm2}. By the dominated convergence theorem we find
\begin{align} \label{eq_hirschDom}
\log |(f\circ g)(r \ee^{\ci \phi})| \leq \frac{1}{2 \pi} \int_{-\pi}^{\pi} \di \varphi \log |(f\circ g)( \ee^{\ci \varphi})|\frac{1 -r^2}{1 - 2 r \cos(\phi - \varphi) + r^2} \, .
\end{align}
For $x:=g(r \ee^{\ci \varphi})$ we obtain
\begin{align}
r \ee^{\ci \varphi} 
= g^{-1}(x)
= \frac{\ee^{\ci \pi x}-\ci}{\ee^{\ci \pi x}+\ci} 
= -\ci \frac{\cos(\pi x)}{1+\sin(\pi x)} = \left(\frac{\cos(\pi x)}{1+\sin(\pi x)} \right) \ee^{-\ci \frac{\pi}{2}} \, ,
\end{align}
from which we see that in case $x \in (0,\frac{1}{2}]$ we have
 $r = \frac{\cos(\pi x)}{1+\sin(\pi x)}$ and $\theta = - \frac{\pi}{2}$ and in case $x \in (\frac{1}{2},1)$ we have $r = - \frac{\cos(\pi x)}{1+\sin(\pi x)}$ and $\theta=\frac{\pi}{2}$.
 In both cases we find
 \begin{align}
 \frac{1-r^2}{1-2r \cos(\phi - \varphi) + r^2} = \frac{\sin(\pi x)}{1+\cos(\pi x) \sin(\varphi)} \, .
 \end{align}
Plugging this into~\eqref{eq_hirschDom} shows that
  \begin{align} \label{eq_hirschdone0}
  \log |f(x)| \leq \frac{1}{2\pi} \int_{-\pi}^{\pi} \di \varphi \, \frac{\sin(\pi x)}{1+\cos(\pi x) \sin(\varphi)}  \log |(f\circ g)( \ee^{\ci \varphi})| \, .
  \end{align}

To conclude we change variables. In case $\varphi \in [-\pi,0]$ we introduce $y$ such that $\ci y = h(\ee^{\ci \varphi})$ or equivalently $\ee^{\ci \varphi}=-\tanh(\pi y) -\frac{\ci}{\cosh(\pi y)}$. Since $\varphi \in  [-\pi,0]$ we obtain $y \in (-\infty,\infty)$ and $\di \varphi = -\frac{\pi}{\cosh(\pi y)} \di y$. As a result we find
\begin{align}
\frac{1}{2 \pi} \int_{-\pi}^{0} \di \varphi \, \frac{\sin(\pi x)}{1+\cos(\pi x) \sin(\varphi)}  \log |(f\circ g)( \ee^{\ci \varphi})|
= \frac{1}{2} \int_{-\infty}^{\infty} \di y \, \frac{\sin(\pi x)}{\cosh(\pi y) - \cos(\pi x)} \log |f(\ci y)| \, . \label{eq_hirschdone1}
\end{align}
In case $\varphi \in [0, \pi]$ we define $y$ such that $1+\ci y=h(\ee^{\ci \varphi})$ or equivalently $\ee^{\ci \varphi}= - \tanh(\pi y) + \frac{\ci}{\cosh(\pi y)}$. Since $\varphi \in [0, \pi]$ we obtain $y \in (-\infty,\infty)$ and $\di \varphi = \frac{\pi}{\cosh(\pi y)}\di y$. As a result we find
\begin{align}
\frac{1}{2 \pi} \int_{0}^{\pi} \di \varphi \, \frac{\sin(\pi x)}{1+\cos(\pi x) \sin(\varphi)}  \log |(f\circ g)( \ee^{\ci \varphi})|
= \frac{1}{2} \int_{-\infty}^{\infty} \di y \, \frac{\sin(\pi x)}{\cosh(\pi y) - \cos(\pi x)} \log |f(1+\ci y)| \, . \label{eq_hirschdone2}
\end{align}
Combining~\eqref{eq_hirschdone0}~\eqref{eq_hirschdone1} and~\eqref{eq_hirschdone2} proves the assertion. \qed
\end{proof}

\begin{proof}[Theorem~\ref{thm_hirschman}]
By assumption, the operator $F(\theta)$ is bounded for any fixed $\theta \in (0,1)$. Consequently, $F(\theta)$ has a polar decomposition~\cite[Theorem~VI.10]{simon_book}, i.e., $F(\theta) = V B$, where $B$ is positive semi-definite and $V$ is a partial isometry satisfying $B V^{\dag}V = V^{\dag}V B = B$. 
 Let $x \in [0,1]$ and define $q_{x}$ as the H\"older conjugate of $p_{x}$ such that $p_{x}^{-1} + q_{x}^{-1} = 1$. By definition of $p_x$ (see Theorem~\ref{thm_hirschman}), we have
\begin{align}
  \frac{1}{q_x} = \frac{1-x}{q_0} + \frac{x}{q_1} \,.
\end{align}
We next define $X(z)$ by
\begin{align}
  X(z)^{\dag} =  \kappa^{-p_{\theta} \left( \frac{1-{z}}{q_0} + \frac{{z}}{q_1} \right)} B^{ p_{\theta} \left( \frac{1-{z}}{q_0} + \frac{{z}}{q_1} \right)  } V^{\dag}
  \qquad \textrm{with} \qquad \kappa := \norm{ B }_{p_{\theta}} = \norm{ F(\theta) }_{p_{\theta}} < \infty \,.
\end{align}
It is easy to see that $z \mapsto X(z)$ is anti-holomorphic on $S$ and
\begin{align}
  \norm{ X(x + \ci y) }_{q_x}^{q_x} = \tr\, \left( \kappa^{-1} B \right)^{p_{\theta}q_{x}\left( \frac{1-x}{q_0} + \frac{x}{q_1} \right)} = \tr\, \left( \kappa^{-1} B \right)^{p_{\theta}} = 1 \,.
\end{align}
As a result $f(z) := \tr\, X(z)^{\dag}F(z)$ is holomorphic and bounded on $S$ since by H\"older's inequality (see, e.g.,~\cite[Theorem~7.8]{weidmann_book}) we have
\begin{align}
   |f(x + \ci y)| \leq \norm{ X(x + \ci y) }_{q_x}  \norm{ F(x + \ci y) }_{p_x} \leq \norm{ F(x + \ci y) }_{p_x}  
   \label{eq:hoelder-used} \,.
\end{align}
Consequently, our assumptions on $F(z)$ imply that $f(z)$ satisfies the assumptions of Lemma~\ref{lem_hirschman}.

By definition of $X(z)$ we find
\begin{align}
  f(\theta) &= \tr\, X(\theta)^\dagger F(\theta) = \kappa^{-p_{\theta}\frac{1}{q_{\theta}}}\, \tr\, B^{p_{\theta}-1} V^{\dag} V B = \kappa^{1-p_{\theta}}\, \tr\, B^{p_{\theta}} = \norm{F(\theta)}_{p_{\theta}}  \, .
\end{align}  
Furthermore, according to~\eqref{eq:hoelder-used} we have
 \begin{align} 
  |f(\ci t)| &\leq \norm {F(\ci t)}_{p_0}  \quad \textrm{and} \quad |f(1+\ci t)| \leq \norm {F(1+\ci t)}_{p_1}
  \,.
\end{align}
Plugging this into Lemma~\ref{lem_hirschman} yields the desired result. \qed
\end{proof}

\section{Background and further reading}
A question that is related to the topics discussed in this chapter is whether Hermitian operators that almost commute are close to Hermitian operators that commute (with respect to the operator norm). This question has a long history that dates back to the 1950s or earlier (see, e.g.,~\cite{rosenthal69,halmos76}). It has been finally solved in~\cite{lin97} (see also~\cite{friis96} for a simplified proof). Recent progress has been obtained in~\cite{Hastings2009,Kach16}, where~\cite{Hastings2009} uses the concept of smooth pinching. Lemma~\ref{lem_hastings} is similar to Lemma~1~in~\cite{Hastings2009}. The pinching inequality (given in Lemma~\ref{lem_propertiesPinching}) was proven in~\cite{hayashi02}. More information about the spectral pinching method can be found in~\cite{marco_book,carlen_book}.

Complex interpolation theory is an established technique that is frequently used by mathematical physicists. Epstein~\cite{epstein73} showed how interpolation theory can be utilized in matrix analysis. 
Recently, the technique attracted attention in quantum information theory for proving entropy inequalities. 
Beigi~\cite{beigi13} and Dupuis~\cite{dupuis15} used variations of the Riesz-Thorin theorem based on Hadamard's three line theorem to show properties of the minimal R\'enyi relative entropy and conditional R\'enyi entropy, respectively. Wilde~\cite{wilde15} first used complex interpolation theory to prove remainder terms for the monotonicity of quantum relative entropy. Extensions and further applications of this approach are discussed by Dupuis and Wilde~\cite{DW15}. Hirschmann's refinement was first studied in this context by Junge \emph{et al.}~\cite{JRSWW15}.


\chapter{Multivariate trace inequalities}
\label{chapter_traceIneq} 

\abstract{ In this chapter we discuss real valued inequalities of functionals of operators. We show how the techniques of pinching and complex interpolation theory discussed in the preceding chapter can be applied to prove such inequalities. These inequalities will be applied in the next chapter to understand the behavior of approximate quantum Markov chains.}
\vspace{8mm}
\noindent Trace inequalities are mathematical relations between different multivariate trace functionals. Oftentimes
these relations are straightforward equalities if the involved matrices commute --- and can be difficult
to prove for the non-commuting case.

\section{Motivation} \label{sec_mot_trace} \index{trace inequalities}
Arguably one of the most powerful trace inequalities is the celebrated \emph{Golden-Thompson (GT) inequality}~\cite{golden65,thompson65}. It relates the trace of the exponential of a sum of two matrices with the trace of the product of the individual exponentials.
\begin{svgraybox}
\vspace{-4mm}
\begin{theorem}[Golden-Thompson]\index{Golden-Thompson inequality!original version} \label{thm_GT}  \index{trace inequalities!Golden-Thompson}
Let $H_1,H_2 \in \Her(A)$. Then
\begin{align} \label{eq_2GT}
\tr\, \ee^{H_1 + H_2} \leq \tr\, \ee^{H_1} \ee^{H_2} \, ,
\end{align}
with equality if and only if $[H_1,H_2]=0$.
\end{theorem}
\vspace{-4mm}
\end{svgraybox}
We note that the GT inequality is relating two nonnegative real numbers. To see this, we note that the right-hand side can be rearranged as $\tr\, \exp(\frac{H_2}{2})\exp(H_1)\exp(\frac{H_2}{2})$, using the cyclic property of trace, which is always nonnegative since $\exp(\frac{H_2}{2})\exp(H_1)\exp(\frac{H_2}{2}) \in \Pos(A)$.

The GT inequality has found applications ranging from statistical physics~\cite{thompson65}, random matrix theory~\cite{AW02,Tropp11,tao2012topics}, and linear system theory~\cite{bernstein88} to quantum information theory~\cite{LieRus73_1,LieRus73}.

There exists a variety of different proofs for the GT inequality. In Section~\ref{sec_specPinchGT} we presented an intuitive proof that is based on the spectral pinching method discussed in Section~\ref{sec_pinching}. The motivation for the use of the pinching technique comes from the fact that~\eqref{eq_2GT} is trivial if $H_1$ and $H_2$ commute. 

\begin{proof}[Theorem~\ref{thm_GT}]
Inequality~\eqref{eq_2GT} has been proven in Section~\ref{sec_specPinchGT} based on the asymptotic spectral pinching method.
It is immediate to see that \eqref{eq_2GT} holds with equality in case $[H_1,H_2]=0$. The converse statement is proven in~\cite[Theorem~2.1]{hiai94}.
\qed
\end{proof}
As we will see later, the proof presented in Section~\ref{sec_specPinchGT} already suggests an extension of the GT inequality to $n$ matrices by iterative pinching.
\begin{exercise}\label{ex_3GT}
Apply the asymptotic spectral pinching method (as shown in the proof given in Section~\ref{sec_specPinchGT}) to prove the following extension of the GT inequality to three matrices
\begin{align} \label{eq_2GT_pinch}
\tr\, \ee^{H_1 + H_2+H_3} \leq \sup_{t \in \R} \, \tr\, \ee^{H_1} \ee^{\tfrac{1+\ci t}{2}H_2} \ee^{H_3} \ee^{\tfrac{1-\ci t}{2}H_2} 
\end{align}
and compare it to~\eqref{eq_ourGT3} that we will prove later.
[Hint: use the integral representation of the pinching map given by Lemma~\ref{lem_IntegralPinching}]
\end{exercise}

The GT inequality can be derived from the more general \emph{Araki-Lieb-Thirring (ALT) inequality}~\cite{Lieb76,araki90}, which relates the trace of a product of two positive operators with a global and a local power.
\begin{svgraybox}
\vspace{-4mm}
\begin{theorem}[Araki-Lieb-Thirring]\index{Araki-Lieb-Thirring inequality!original version} \label{thm_ALT} \index{trace inequalities!Araki-Lieb-Thirring}
Let $B_1,B_2 \in \Pos(A)$ and $q>0$. Then
\begin{align} \label{eq_2ALT}
\tr\, \big(B_1^{\frac{r}{2}} B_2^r B_1^{\frac{r}{2}} \big)^{\frac{q}{r}} \leq \tr\, \big(B_1^{\frac{1}{2}} B_2 B_1^{\frac{1}{2}}\big)^q \quad \text{if} \quad r \in (0,1] \, ,
\end{align}
with equality if and only if $[B_1,B_2]=0$. The inequality holds in the opposite direction for $r \geq 1$.
\end{theorem}
\vspace{-4mm}
\end{svgraybox}
\begin{proof}
We present a proof based on the asymptotic spectral pinching method that is similar as the proof for the GT inequality explained in Section~\ref{sec_specPinchGT}. Using basic properties of the tensor product that are stated in Exercise~\ref{exercise_tensorProduct} we find for $r \in (0,1]$ and $m\in \N$
\begin{align}
\log \tr\, \big(B_1^{\frac{r}{2}} B_2^r B_1^{\frac{r}{2}} \big)^{\frac{q}{r}}
&= \frac{1}{m} \log \tr\, \Big( (B_1^{\frac{r}{2}})^{\otimes m} (B_2^r)^{\otimes m} (B_1^{\frac{r}{2}})^{\otimes m} \Big)^{\frac{q}{r}} \\
&\leq \frac{1}{m} \log \tr\, \Big( (B_1^{\frac{r}{2}})^{\otimes m} \cP_{B_1^{\otimes m}} \big( (B_2^r)^{\otimes m} \big) (B_1^{\frac{r}{2}})^{\otimes m} \Big)^{\frac{q}{r}} + \frac{\log \poly(m)}{m} \, ,
\end{align}
where the final step uses the pinching inequality (see Lemma~\ref{lem_propertiesPinching}), the monotonicity of the function $X \mapsto \tr \, X^{\alpha}$ for $\alpha \geq 0$ (see Proposition~\ref{prop_traceFunctions}) and the fact that the number of distinct eigenvalues of $B_1^{\otimes m}$ grows polynomially on $m$ (see Remark~\ref{remark_types}). Since $t \mapsto t^\alpha$ for $\alpha \in (0,1]$ is operator concave Lemma~\ref{lem_propertiesPinching} shows that
\begin{align}
\log \tr\, \big(B_1^{\frac{r}{2}} B_2^r B_1^{\frac{r}{2}} \big)^{\frac{q}{r}}
&\leq  \frac{1}{m} \log \tr\, \Big( (B_1^{\frac{r}{2}})^{\otimes m} \cP_{B_1^{\otimes m}}(B_2^{\otimes m})^r  (B_1^{\frac{r}{2}})^{\otimes m} \Big)^{\frac{q}{r}} + \frac{\log \poly(m)}{m}\\
&= \frac{1}{m} \log \tr\, \Big( \cP_{B_1^{\otimes m}} \big( (B_1^{\frac{1}{2}})^{\otimes m} B_2^{\otimes m}  (B_1^{\frac{1}{2}})^{\otimes m} \big) \Big)^q + \frac{\log \poly(m)}{m} \label{eq_mids} \\
& \leq  \frac{1}{m} \log \tr\, \Big( (B_1^{\frac{1}{2}})^{\otimes m} B_2^{\otimes m}  (B_1^{\frac{1}{2}})^{\otimes m}\Big)^q + \frac{\log \poly(m)}{m} \\
& =   \log \tr\, \Big( B_1^{\frac{1}{2}} B_2 B_1^{\frac{1}{2}} \Big)^q + \frac{\log \poly(m)}{m} \,  ,
\end{align}
where~\eqref{eq_mids} uses that $\cP_{B_1}(B_2)$ commutes with $B_1$. The penultimate step uses Lemma~\ref{lem_normConvexNEW} (see also~\eqref{eq_TCpropnew}) together with the integral representation of the pinching map (Lemma~\ref{lem_IntegralPinching}) and the fact that $p$-norms are unitarily invariant for all $p\geq 0$. The final step uses basic properties of the tensor product described in Exercise~\ref{exercise_tensorProduct}. Considering the limit $m \to \infty$ then proves~\eqref{eq_2ALT}.
The fact that~\eqref{eq_2ALT} holds in the opposite direction in case $r \geq 1$ follows from the substitution $B_k^r \leftarrow B_k$ for $k \in \{1,2\}$, $\frac{q}{r}\leftarrow q$, and $\frac{1}{r}\leftarrow r$. That~\eqref{eq_2ALT} is an equality if and only if the two matrices commute is proven in~\cite[Theorem~2.1]{hiai94}. \qed
\end{proof}

The GT inequality is implied by the ALT inequality. To see this we recall the \emph{Lie product formula} for operators (see, e.g.,~\cite[Problem~IX.8.5]{bhatia_book}).\index{Lie product formula}
\begin{lemma}[Lie product formula] \label{lem_LieProduct}
Let $n\in \N$ and $(L_k)_{k=1}^n$ be a finite sequence of linear operators on $A$. Then
\begin{align}
\lim_{m \to \infty} \left(  \prod_{k=1}^n \ee^{\frac{L_k}{m}} \right)^m = \exp \left( \sum_{k=1}^n L_k \right)\, .
\end{align}
\end{lemma}
We note that for $r \to 0$ the Lie product formula shows that the ALT inequality~\eqref{eq_2ALT} simplifies to
\begin{align} \label{eq_2GTnorm}
\tr \, (\ee^{\log B_1 + \log B_2})^q \leq \tr \, (B_1^{\frac{1}{2}} B_2 B_1^{\frac{1}{2}})^q \, ,
\end{align}
which for $q=1$ gives the GT inequality~\eqref{eq_2GT}

The straightforward logarithmic analog of the GT inequality is a relation between $\tr \log B_1 B_2 $ and $\tr \log B_1 + \tr \log B_2$ for $B_1,B_2 \in \Pos(A)$. As the determinant is multiplicative and since $\tr \log B_1 = \log \det B_1$ we find that
\begin{align}
 \tr \log B_1 + \tr \log B_2 = \tr \log B_2^{\frac{1}{2}} B_1 B_2^{\frac{1}{2}}\,.
\end{align}
This trivially extends to $n$ matrices. 
\begin{exercise}\label{ex_logdet}
Show that $\tr\, \log B = \log \det B$ for all $B \in \Pos(A)$.
\end{exercise}

The following theorem states a more interesting logarithmic trace inequality~\cite{HP93,ando94}. In particular it provides an upper and lower bound for the relative entropy defined in Definition~\ref{def_relEnt}.
\begin{svgraybox}
\vspace{-4mm}
\begin{theorem}[Logarithmic trace inequality]\index{logarithmic trace inequality} \label{thm_LogTrace} \index{trace inequalities!logarithmic}
Let $B_1,B_2 \in \Pos(A)$ and $p>0$. Then
\begin{align}\label{eq_logTrace}
 \frac{1}{p}\tr \, B_1 \log B^{\frac{p}{2}}_2 B_1^p B^{\frac{p}{2}}_2 
 \leq  \tr \, B_1 ( \log B_1 + \log B_2)
\leq \frac{1}{p}\tr \, B_1 \log B_1^{\frac{p}{2}} B_2^p B_1^{\frac{p}{2}} \, ,
\end{align}
with equalities in the limit $p\to0$.
\end{theorem}
\vspace{-4mm}
\end{svgraybox}
\begin{proof}
First, note that both inequalities are invariant under multiplication of the operators $B_1$, $B_2$ with positive scalars $b_1,b_2>0$ and hence additional constraints on the norms of the matrices can be introduced without loss of generality. We thus assume without loss of generality that $\tr\, B_1 =1$.

We start by proving the first inequality. Using the variational formula for the relative entropy given by Lemma~\ref{lem_varFormulaRelEnt} we find for any $p>0$
\begin{align}
\tr \, B_1 ( \log B_1 + \log B_2) 
&= D(B_1 \| B_2^{-1}) \\
&=  \sup_{\omega > 0} \left \lbrace \tr\, B_1 \log \omega + 1 - \tr \, \ee^{\log \omega - \log B_2} \right \rbrace \\
&\geq \sup_{\omega > 0} \left \lbrace \tr\, B_1 \log \omega + 1 - \tr \left(B_2^{-\frac{p}{2}} \omega^p B_2^{-\frac{p}{2}} \right)^p \right \rbrace \\
&\geq \frac{1}{p} \tr\, B_1 \log B_2^{\frac{p}{2}} B_1^p B_2^{\frac{p}{2}} \, ,
\end{align}
where the first inequality follows form the GT inequality given in~\eqref{eq_2GTnorm}. The final step uses that $\omega = (B_2^{\frac{p}{2}} B_1^p B_2^{\frac{p}{2}})^{\frac{1}{p}}>0$.

The second inequality is proven in~\cite{HP93}. A simplified argument for the case $p=1$ can be found in~\cite[Section~3.5.1]{raban_MA}.  \qed
\end{proof}

All the trace inequalities presented in this section involve two operators. It is a natural question if they feature extensions to arbitrarily many operators --- so-called \emph{multivariate trace inequalities}. \index{multivariate trace inequalities}
The remaining part of this chapter deals with this question.

\section{Multivariate Araki-Lieb-Thirring inequality} \label{sec_ALT}
The ALT inequality presented in Theorem~\ref{thm_ALT} has been extended in various directions (see, e.g.,~\cite{kosaki92,ando94,wang95,audenaert08}). Recently, an extension of the ALT inequality to arbitrarily many operators has been proven~\cite{SBT16} which was further generalized in~\cite{Hiai16}.
\begin{svgraybox}
\vspace{-4mm}
\begin{theorem} [$n$-matrix extension of ALT] \label{thm_ALT_Hirschman}\index{Araki-Lieb-Thirring inequality!multivariate version} \index{trace inequalities!Araki-Lieb-Thirring}
Let $p>0$, $r \in (0, 1]$, $\beta_r$ as defined in~\eqref{eq_densi}, $n\in \N$, and consider a finite sequence $(B_k)_{k=1}^n$ of nonnegative operators. Then
\begin{align}\label{eq_ALT_new1}
&\log \norm{ \left| \prod_{k=1}^n B_k^r \right|^{\frac{1}{r}} }_p \leq \int_{-\infty}^{\infty} \di t  \beta_{r}(t)\, \log \norm{ \prod_{k=1}^n B_k^{1 + \ci  t} }_p \,.
\end{align}
\end{theorem}
\vspace{-4mm}
\end{svgraybox}
\begin{proof}
The case $r = 1$ holds trivially with equality, so suppose $r\in (0,1)$. We prove the result for strictly positive operators and note that the generalization to nonnegative operators follows by continuity. Furthermore, we assume in a first step that $p\geq 1$. The idea is to prove the assertion by using complex interpolation theory. To do so, we define the function 
 \begin{align}
 F(z):=\prod_{k=1}^n B_k^z = \prod_{k=1}^n \exp(z \log B_k) \, ,
 \end{align}
which satisfies the regularity assumptions of the Stein-Hirschman theorem (see Theorem~\ref{thm_hirschman}). Furthermore we pick $\theta = r$, $p_0 = \infty$ and $p_1 = p$ such that $p_{\theta} = \frac{p}{r}$. A simple calculation reveals that
\begin{align}
\log \norm{F(1+\ci t)}_{p_1}^\theta = r \log \norm{\prod_{k=1}^n B_k^{1+\ci t}}_{p}
\end{align}
and
\begin{align}
\log \norm{F(\ci t)}_{p_0}^{1-\theta} = (1- r) \log \norm{\prod_{k=1}^n B_k^{\ci t} }_{\infty} = 0 \ ,
\end{align}
since the operators $B_k^{\ci t}$ are unitary.
Moreover, we have
\begin{align}
\log \norm{ F(\theta) }_{p_{\theta}} 
= \log \norm{ \prod_{k=1}^n B_k^{r} }_{\frac{p}{r}} 
= r \log \norm{ \left| \prod_{k=1}^n B_k^{r} \right|^{\frac{1}{r}} }_{p}  \ .
\end{align}
Substituting this into Theorem~\ref{thm_hirschman} yields the desired inequality for $p\geq 1$. The case $0<p\leq 1$ follows from a standard technique called \emph{antisymmetric tensor power calculus}. This is explained in detail in~\cite{Hiai16}. \qed
\end{proof}

\begin{remark}
Using antisymmetric tensor power calculus it can be shown that~\eqref{eq_ALT_new1} is true for any unitarily invariant norm (see~\cite{Hiai16} for more information).
\end{remark}

Let us now comment on various aspects of~\eqref{eq_ALT_new1}. For $q \in \R_+$, $r \in (0,1]$, and the substitution $p \leftarrow 2q$ and $B_k \leftarrow \sqrt{B_k}$ we can rewrite~\eqref{eq_ALT_new1} as
\begin{multline} \label{eq_suggestiveForm}
\log \tr \left( B_1^{\frac{r}{2}} B_2^{\frac{r}{2}} \cdots B_{n-1}^{\frac{r}{2}} B_n^{r} B_{n-1}^{\frac{r}{2}} \cdots B_2^{\frac{r}{2}} B_1^{\frac{r}{2}} \right)^{\frac{q}{r}}  \\
 \leq \int_{-\infty}^{\infty} \di t  \beta_{r}(t)\, \log \tr \left( B_1^{\frac{1}{2}} B_2^{\frac{1+\ci t}{2}} \cdots B_{n-1}^{\frac{1+\ci t}{2}} B_n B_{n-1}^{\frac{1-\ci t}{2}} \cdots B_2^{\frac{1-\ci t}{2}} B_1^{\frac{1}{2}} \right)^q .
\end{multline}
For $n = 2$ this simplifies to the original ALT inequality given by Theorem~\ref{thm_ALT}.  
By Jensen's inequality we can remove the logarithm in~\eqref{eq_ALT_new1}. Furthermore, for $q \in [0,1]$ we may shift the integral inside the quasi-norm using the fact that $X \mapsto \log \norm{X}_p$ is concave for $p \in [0,1]$\footnote{This follows from Proposition~\ref{prop_traceFunctions}.}, which yields 
\begin{multline}
 \norm{ \left( B_1^{\frac{r}{2}} B_2^{\frac{r}{2}} \cdots B_{n-1}^{\frac{r}{2}} B_n^{r} B_{n-1}^{\frac{r}{2}} \cdots B_2^{\frac{r}{2}} B_1^{\frac{r}{2}} \right)^{\frac{1}{r}} 
}_q \\
 \leq \norm{ \int_{-\infty}^{\infty} \di t  \beta_{r}(t)\, B_1^{\frac{1}{2}} B_2^{\frac{1+\ci t}{2}} \cdots B_{n-1}^{\frac{1+\ci t}{2}} B_n B_{n-1}^{\frac{1-\ci t}{2}} \cdots B_2^{\frac{1-\ci t}{2}} B_1^{\frac{1}{2}} }_q .
\end{multline}

\section{Multivariate Golden-Thompson inequality} \label{sec_GT}
Given the usefulness of the GT inequality presented in Theorem~\ref{thm_GT}, it is natural to ask if the GT inequality can be extended to more than two operators. In 1973, Lieb proved a three operator extension of the GT inequality~\cite{Lieb73} that attracted a lot of interest and raised the question if the GT inequality can be extended to more than three matrices. This has been an open question until recently (see Theorem~\ref{thm_GT_steinHirschman}). 
\begin{theorem}[Lieb's triple operator inequality] \index{Lieb's triple operator inequality} \label{thm_Lieb3} \index{trace inequalities!Lieb's triple operator inequality}
Let $H_1, H_2, H_3 \in \Her(A)$. Then
\begin{align} \label{eq_lieb3}
\tr\, \ee^{H_1 + H_2 + H_3}  \leq \int_{0}^{\infty} \!\di s \, \tr  \, \ee^{H_1} \big(\ee^{-H_2}+s\,\id_A \big)^{-1} \ee^{H_3} \big(\ee^{-H_2}+s\,\id_A\big)^{-1} \ .
\end{align}
\end{theorem}
Lieb's triple operator inequality has been shown to be equivalent to many other interesting statements such as Lieb's concavity theorem (see Theorem~\ref{thm_liebConc}) or strong subadditivity of quantum entropy~\cite{LieRus73,LieRus73_1} (see~\eqref{eq_SSA}).\footnote{The reason why all these statements are equivalent is explained in~\cite{Lieb73} (see also~\cite{ruskai05}).} \index{strong subadditivity}
We postpone the proof of Theorem~\ref{thm_Lieb3} to the end of this section. It can be verified easily that in case $H_2=0$~\eqref{eq_lieb3} simplifies to the original GT inequality~\eqref{eq_2GT}.

The $n$-operator extension of the ALT inequality presented in Theorem~\ref{thm_ALT_Hirschman} implies (via the Lie product formula given by Lemma~\ref{lem_LieProduct}) an extension of the GT inequality to arbitrarily many operators.
\begin{svgraybox}
\vspace{-4mm}
\begin{theorem}[$n$-matrix extension of GT] \label{thm_GT_steinHirschman}  \index{Golden-Thompson inequality!multivariate version} \index{trace inequalities!Golden-Thompson}
Let $p>0$, $\beta_0$ as defined in~\eqref{eq_beta_0}, $n \in \N$ and consider a finite sequence $( H_k)_{k=1}^n$ of Hermitian operators. Then
\begin{align} \label{eq_mainResGT}
  \log \norm{ \exp \left( \sum_{k=1}^n H_k \right) }_p \leq 
  \int_{-\infty}^{\infty} \di t  \beta_0(t)\, \log \norm{ \prod_{k=1}^n  \exp\bigl( (1+\ci t) H_k \bigr) }_p \ \,.
\end{align}
\end{theorem}
\vspace{-4mm}
\end{svgraybox}
\begin{proof}
Follows from Theorem~\ref{thm_ALT_Hirschman} together with the Lie product formula (see Lemma~\ref{lem_LieProduct}) when considering the limit $r \to 0$. \qed
\end{proof}
\begin{remark}
Using antisymmetric tensor power calculus it can be shown that~\eqref{eq_mainResGT} is true for any unitarily invariant norm (see~\cite{Hiai16} for more details).
\end{remark}

If we evaluate~\eqref{eq_mainResGT} for $n=3$ and $p=2$ using the substitution $H_k \leftarrow \frac12 H_k$ we obtain
\begin{align}
\log \tr \, \ee^{H_1 + H_2 + H_3} \leq \int_{-\infty}^{\infty} \di t \beta_0(t) \log \tr  \, \ee^{H_1} \ee^{\tfrac{1+\ci t}{2}H_2} \ee^{H_3} \ee^{\tfrac{1-\ci t}{2}H_2} \, .
\end{align}
By the concavity of the logarithm we can further simplify this inequality to
\begin{align} \label{eq_ourGT3}
  \tr \, \ee^{H_1 + H_2 + H_3} 
  \leq \int_{-\infty}^{\infty} \di t  \beta_0(t)\, \tr \, \ee^{H_1} \ee^{\tfrac{1+\ci t}{2}H_2} \ee^{H_3} \ee^{\tfrac{1-\ci t}{2}H_2}\, .
\end{align} 
As it happens this inequality coincides with Lieb's triple operator inequality~\eqref{eq_lieb3}. To see this we consider the following lemma.
\begin{lemma}\label{lem_LiebRep} \index{Fr\'echet derivative}
  Let $B \in \Pos(A)$ and $H \inÊ\Her(A)$. Then, the following two expressions for the Fr\'echet derivative of the logarithm are equivalent:
  \begin{align} 
    \frac{\di }{\di r} \Big |_{r=0} \log(B + r H) &= \int_{0}^{\infty} \!\di s \,(B+s\,\id_A)^{-1} H (B+s\,\id_A)^{-1} \\
    &= \int_{-\infty}^{\infty} \di t \beta_0(t)\, B^{-\frac{1+\ci t}{2} } H B^{-\frac{1-\ci t}{2}} \,. \label{eq_derivative}
  \end{align}
\end{lemma}
\begin{proof}
The first equality in the lemma is well-known and can be derived using integral representations of the operator logarithm (see, e.g., \cite{carlen_book}). To see why the second equality step is true we expand both terms in the eigenbasis of $B$. More precisely, for  $B=\sum_k\lambda_k\ket{k}\!\bra{k}$ we find
\begin{align}
\int_{0}^{\infty} \di s \,(B+s\,\id_A)^{-1} H (B+s\,\id_A)^{-1}
&=\sum_{k,\ell} \bra{k}H\ket{\ell} \ket{k} \bra{\ell} \int_0^{\infty} \di s (\lambda_k +s)^{-1}(\lambda_\ell +s)^{-1} \, .
\end{align}
A simple calculation shows that
\begin{align}
\int_0^{\infty} \di s (\lambda_k +s)^{-1}(\lambda_\ell +s)^{-1}
= \frac{1}{\lambda_\ell - \lambda_k} \log \frac{\lambda_\ell}{\lambda_k}
=\frac{1}{\sqrt{\lambda_k \lambda_\ell}} \int_{-\infty}^{\infty} \beta_0(\di t) \left( \frac{\lambda_\ell}{\lambda_k} \right)^{\frac{\ci t}{2}} \, .
\end{align}
As a result we have
\begin{align}
\int_{0}^{\infty} \di s \,(B+s\,\id_A)^{-1} H (B+s\,\id_A)^{-1}
&= \int_{-\infty}^{\infty} \di t \beta_0(t) \sum_{k,\ell} \lambda_k^{-\frac{1}{2}-\frac{\ci t}{2}} \lambda_\ell^{-\frac{1}{2}+\frac{\ci t}{2}} \bra{k}H\ket{\ell} \ket{k} \bra{\ell} \\
&= \int_{-\infty}^{\infty} \di t  \beta_0(t)\, B^{-\frac{1}{2} - \frac{\ci t}{2}} H B^{-\frac{1}{2}+\frac{\ci t}{2}} \, ,
\end{align}
which proves the second equality of the lemma. \qed
\end{proof}
Lemma~\ref{lem_LiebRep} presents two alternative expressions for the Fr\'echet derivative of the operator logarithm, one in terms of \emph{resolvents} and the other one in terms of an average over unitaries. The lemma also provides further insight in the probability density $\beta_0$ which we obtain from Hirschman's interpolation theorem.
Lieb's triple operator inequality (see Theorem~\ref{thm_Lieb3}) thus follows directly by combining~\eqref{eq_ourGT3} with Lemma~\ref{lem_LiebRep}. 
\begin{remark}
Recently it was shown that the right-hand side of~\eqref{eq_mainResGT} features an alternative representation without any unitaries, however in terms of resolvents~\cite{lemm17_2} as in Theorem~\ref{thm_Lieb3} for the special case of three matrices. 
\end{remark}

The multivariate GT inequality presented by Theorem~\ref{thm_GT_steinHirschman} is valid for Hermitian operators.
The following theorem proves an $n$-operator extension of the GT inequality for general linear operators.
\begin{theorem} \label{thm_GT_general} \index{Golden-Thompson inequality!linear operators}
Let $p>0$, $\beta_0$ as defined in~\eqref{eq_beta_0}, $n \in \N$ and consider a finite sequence $( L_k)_{k=1}^n$ of linear operators. Define the real part of $L_k$ by $\Real(L_k) := \frac12 (L_k + L_k^\dag)$. Then
\begin{align}
  \log \norm{ \exp \left( \sum_{k=1}^n L_k \right) }_p \leq 
  \int_{-\infty}^{\infty}  \di t \beta_0(t)\, \log \norm{ \prod_{k=1}^n  \exp\bigl( (1+\ci t) \Real(L_k) \bigr) }_p .
\end{align}
\end{theorem}
\begin{proof}
We define the imaginary part of $L_k$ by $\Immag(L_k) := \frac1{2\ci}(L_k - L_k^{\dag})$. Note that $L_k = \Real(L_k) + \ci\, \Immag(L_k)$ and that both $\Real(L_k)$ and $\Immag(L_k)$ are Hermitian. 
The idea is to prove the assertion of the Theorem via complex interpolation theory. Therefore we consider the function
\begin{align}
  F(z) := \prod_{k=1}^n \exp\big(z \, \Real(L_k) + \ci \theta \, \Immag(L_k)\big), 
\end{align}
which satisfies the regularity assumption of Theorem~\ref{thm_hirschman}. 
We first suppose that $p\geq 1$ and pick $\theta = r \in (0,1)$, $p_0 = \infty$ and $p_1 = p$ such that $p_{\theta} = \frac{p}{r}$. 
Theorem~\ref{thm_hirschman} thus gives
\begin{align}
  r \log \norm{ \left| \exp \left( r \sum_{k=1}^n L_k \right) \right|^{\frac{1}{r}} }_p &= \log \| F(\theta) \|_{p_{\theta}} \\
  &\leq \int_{-\infty}^{\infty} \di t \, \beta_r(t)\, \log \norm{ F(1 + \ci t) }_p^r \\
  &= r \int_{-\infty}^{\infty} \di t \, \beta_r(t)\, \log \norm{ \prod_{k=1}^n \exp\big((1+\ci t)\Real(L_k) + r \, \Immag(L_k) \big) }_p \, , \label{eq_ALTgeneral}
\end{align}
where in the inequality step we used that $\log \norm{F(\ci t)}_{\infty}=0$ as $F(\ci t)$ is unitary.  
Dividing by $r$ and taking the limit $r \to 0$ then yields the desired result via the Lie product formula (see Lemma~\ref{lem_LieProduct}). As before, the case $0<p\leq 1$ follows from antisymmetric tensor power calculus which is described in detail in~\cite{Hiai16}. \qed
\end{proof}
We note that~\eqref{eq_ALTgeneral} can be viewed as an ALT inequality for linear operators.
For $n=1$ and $p=2$, Theorem~\ref{thm_GT_general} simplifies to
\begin{align}
\tr \, \ee^{L} \ee^{L^\dagger} \leq  \tr \, \ee^{L+L^\dagger} \, ,
\end{align}
which was derived in~\cite{bernstein88}.
We further note that for the case of normal operators $N$, the matrices $\Real(N)$ and $\Immag(N)$ commute, which allows us to slightly simplify the above formula by employing the fact that $\exp(\Real(N)) = \big| \exp(N) \big|$. For two normal operators the result then reads
\begin{align}
  \norm{ \exp \left( N_1 + N_2 \right) }_p \leq \norm{ \big| \exp(N_1) \big| \big| \exp(N_2) \big| }_p ,
\end{align}
generalizing an inequality derived in~\cite{Li14}.

\section{Multivariate logarithmic trace inequality} \label{sec_logTrace}
The extension of the GT inequality presented in Theorem~\ref{thm_GT_steinHirschman} can be used to derive an extension of the logarithmic trace inequality given in Theorem~\ref{thm_LogTrace} to arbitrarily many operators~\cite{davidLogTrace}.
\begin{svgraybox}
\vspace{-4mm}
\begin{theorem} \label{thm_logTraceMulti} \index{trace inequalities!logarithmic}
Let $q>0$, $\beta_0$ as defined in~\eqref{eq_beta_0}, $n \in \N$, and consider a finite sequence $(B_k)_{k=1}^n$ of nonnegative operators. Then, we have
\begin{align} \label{eq_multi_logTr}
&\sum_{k=1}^n \tr\, B_1 \log B_k 
\geq \int_{-\infty}^{\infty} \di t \beta_0( t) \frac{1}{q}\tr \, B_1 \log  B_{n}^{\frac{q(1+\ci t)}{2}} \cdots B_{3}^{\frac{q(1+\ci t)}{2}} B_{2}^{\frac{q}{2}} B_1^q B_{2}^{\frac{q}{2}} B_{3}^{\frac{q(1-\ci t)}{2}} \cdots B_{n}^{\frac{q(1-it)}{2}}, 
\end{align}
with equality in the limit $q\to0$.
\end{theorem}
\vspace{-4mm}
\end{svgraybox}
For two matrices (i.e.,~$n=2$)~\eqref{eq_multi_logTr} simplifies to the first inequality of~\eqref{eq_logTrace}.
\begin{proof} 
First, note that the statement that we aim to show is invariant under multiplication of the operators $B_1, B_2, \ldots, B_n$ with positive scalars $b_1, b_2, \ldots, b_n > 0$, and hence additional constraints on the norms of the matrices can be introduced without loss of generality.

Let us first show the inequality for $q > 0$, where we suppose that $\tr \, B_1 =1$. By definition of the relative entropy we have
\begin{align}
\sum_{k=1}^n \tr \, B_1 \log B_k 
&= D\Big(B_1 \Big\| \exp\Big(\sum_{k=2}^n \log B^{-1}_k\Big) \Big) \\
&= \sup_{\omega>0} \left \lbrace \tr \, B_1 \log \omega + 1 - \tr \exp \Big(\log \omega - \sum_{k=2}^n \log B_k \Big) \right \rbrace , \label{eq:log_variational}
\end{align}
where we used the variational formula for the relative entropy given in Lemma~\ref{lem_varFormulaRelEnt}. Now note that the $n$-operator extension of the GT inequality (Theorem~\ref{thm_GT_steinHirschman}) can for $pH_k=\log B_k$ and $p=\frac{1}{q}$ be relaxed to
\begin{align*}
\tr\exp\left( \sum_{k=1}^n \log B_k \right)\leq \int_{-\infty}^{\infty} \di t \beta_0( t) \tr\,\left(B_{n}^{\frac{q}{2}} \cdots B_{3}^{\frac{q(1+\ci t)}{2}} B_{2}^{\frac{q(1+\ci t)}{2}} B_1^q B_{2}^{\frac{q(1-\ci t)}{2}} B_{3}^{\frac{q(1-\ci t)}{2}} \cdots B_{n}^{\frac{q}{2}}\right)^{\frac{1}{q}}
\end{align*}
using the concavity of the logarithm and Jensen's inequality. Applying this to~\eqref{eq:log_variational} we find
\begin{align}
&\sum_{k=1}^n \tr B_1 \log B_k 
\geq\sup_{\omega>0} \Bigg\{\int_{-\infty}^{\infty} \di t \beta_0(t) \tr B_1 \log \omega +1  \nonumber \\
 &\hspace{35mm}- \tr \Big(B_2^{-\frac{q}{2}} B_3^{-\frac{q(1+\ci t)}{2}}\! \cdots B_n^{-\frac{q(1+\ci t)}{2}} \omega^q B_n^{-\frac{q(1-\ci t)}{2}}  \cdots  B_3^{-\frac{q(1-\ci t)}{2}} B_2^{-\frac{q}{2}}\Big)^{\frac{1}{q}}\Bigg\}\, .\label{eq_mid}
\end{align}
Now since 
\begin{align}
\omega:=\left(B_{n}^{\frac{q(1+\ci t)}{2}} \cdots B_{3}^{\frac{q(1+\ci t)}{2}} B_{2}^{\frac{q}{2}} B_1^q B_{2}^{\frac{q}{2}} B_{3}^{\frac{q(1-\ci t)}{2}} \cdots B_{n}^{\frac{q(1-\ci t)}{2}}\right)^{\frac{1}{q}}
\end{align}
is a nonnegative operator we can insert this into~\eqref{eq_mid}, which then proves the assertion for $q>0$.

Next, we show that in the limit $q\to0$ the inequality in Theorem~\ref{thm_logTraceMulti} also holds in the opposite direction. For the following we suppose that $A_k \geq 1$ for all $k \in \{1, 2, \ldots, n\}$. We use that $\log X\geq1-X^{-1}$ for $X>0$ and hence
\begin{align}
\tr \, B_1 \log  B_{n}^{\frac{q(1+\ci t)}{2}} \cdots B_{2}^{\frac{q}{2}} B_1^q B_{2}^{\frac{q}{2}} \cdots B_{n}^{\frac{q(1-\ci t)}{2}}
&\geq \tr \, B_1 \Big(1-B_{n}^{\frac{-q(1-\ci t)}{2}} \cdots B_{2}^{-\frac{q}{2}} B_1^{-q}B_{2}^{-\frac{q}{2}} \cdots B_{n}^{-\frac{q(1+\ci t)}{2}}\Big)\\
&=: Z_q(t) \, .
\end{align}
By assumption on our operators we have that $B_i^{-1} \leq 1$ for all $i \in \{1, 2, \ldots, n\}$ and thus $Z_q(t) \geq 0$ for all $t \in \mathbb{R}$.
By Fatou's lemma (see, e.g.,~\cite{rudin1987}), we further find 
\begin{align*}
\liminf_{q\to0}\int_{-\infty}^{\infty}  \beta_0(\di t)\, \frac{Z_q(t)}{q}\geq\int_{-\infty}^{\infty}  \beta_0(\di t)\liminf_{q\to0}\frac{Z_q(t)}{q} \, .
\end{align*}
Moreover, since $Z_0(t) = 0$ and
\begin{align*}
\frac{\mathrm{d}}{\di q}Z_q(t)\bigg|_{q=0}= \sum_{k=1}^n \tr\, B_1 \log B_k \quad \text{for all} \quad t \in \mathbb{R},
\end{align*}
an application of l'Hopital's rule yields
\begin{align*}
&\liminf_{q\to0}\frac{Z_q(t)}{q}=\sum_{k=1}^n \tr\, B_1 \log B_k\, .
\end{align*}
Since $\beta_0$ is normalized this proves the assertion. \qed
\end{proof}

\section{Background and further reading}

The GT inequality was proven independently by Golden~\cite{golden65} and Thompson~\cite{thompson65} for an application in statistical physics. It has been generalized in various directions (see, e.g.,~\cite{breitenecker72,ruskai72,araki73,simon_book79,kilmek91,kosaki92,HP93,Li14}).
For example, it has been shown that it remains valid by replacing the trace with any unitarily invariant norm~\cite{segal69,lenard71,thompson71} and an extension to three non-commuting matrices was suggested in~\cite{Lieb73}. An interesting topic that is not covered here is the question for reverse GT inequalities~\cite{HP93,hiai09,Hiai16_2} in terms of matrix means~\cite{bhatia_psd_book}.

The ALT inequality was first proven by Lieb and Thirring~\cite{Lieb76} and then generalized by Araki~\cite{araki90}. It has also been extended in various directions (see, e.g.,~\cite{kosaki92,ando94,wang95,audenaert08}). Similarly as with the GT inequality it is interesting to study reverse ALT inequalities~\cite{Ando94_2,iten17}.

Lieb's triple operator inequality (Theorem~\ref{thm_Lieb3}) is important as it can be used to prove many interesting statements such as strong subadditivity of quantum entropy, the monotonicity of the relative entropy, the joint convexity of the relative entropy, or Lieb's concavity theorem~\cite{Lieb73} (see also~\cite{rus02,tropp12}). 
Lieb's concavity theorem is particularly useful to derive tail bounds for sums of independent random matrices~\cite{Tropp11,Tropp_book} that can be better than if you derive them via the original GT inequality, as done in~\cite{AW02}.
The multivariate GT inequality (Theorem~\ref{thm_GT_steinHirschman}) has been used to derive concentration bounds for expander walks~\cite{garg17}.
Recently, Lemma~\ref{lem_LiebRep} was a key ingredient to prove remainder terms for the superadditivity of the relative entropy~\cite{angela17}.


\chapter{Approximate quantum Markov chains}
\label{chapter_recoverability} 

\abstract{ In this chapter we discuss the concept of quantum Markov chains with a particular focus on the robustness of their properties. This results in a new class of states called approximate quantum Markov chains. To understand the properties of these states the mathematical tools introduced in the preceding chapters will be helpful. As it happens, the key result that justifies the definition of  approximate quantum Markov chains (see Theorem~\ref{thm_FR}) is closely related to various celebrated entropy inequalities. We explain this connection and show how the new insights about approximate quantum Markov chains can be used to prove strengthened versions of different famous entropy inequalities.}
\vspace{8mm}

\noindent  In Chapter~\ref{chapter_intro} we informally discussed the concept of a Markov chain and the differences between the classical and quantum case. Here we formally introduce quantum Markov chains and discuss their properties before explaining which properties remain valid in the approximate case.
\section{Quantum Markov chains} \label{sec_QMC}
\index{Markov chain!quantum}
We start with the formal definition of a quantum Markov chain.
\begin{svgraybox}
\vspace{-4mm}
\begin{definition} \label{def_Markov2} \index{Markov chain!quantum}
A tripartite state $\rho_{ABC} \in \St(A\otimes B \otimes C)$ is called a \emph{quantum Markov chain} in order $A\leftrightarrow B\leftrightarrow C$ if there exists a recovery map $\cR_{B \to BC}\in \TPCP(B,B\otimes C)$ such that 
\begin{align} \label{eq_Markov2}
\rho_{ABC} =  \cR_{B \to BC}(\rho_{AB}) \, .
\end{align}
\end{definition}
\vspace{-4mm}
\end{svgraybox}
Informally the definition above states that the $C$-part can be reconstructed by only acting on the $B$-part. It is interesting to further study the structure of Markov chains --- in particular, if there exists an entropic and an algebraic characterization. 
The following theorem presents an entropic characterization of quantum Markov chains~\cite{Pet86,Pet03}.
\begin{svgraybox}
\vspace{-4mm}
\begin{theorem} \label{thm_PetzCMI} \index{rotated Petz recovery map} \index{Markov chain!entropic structure }
A tripartite state $\rho_{ABC} \in \St(A\otimes B \otimes C)$ is a quantum Markov chain in order $A\leftrightarrow B\leftrightarrow C$  if and only if $I(A:C|B)_{\rho}=0$. Furthermore, in case $I(A:C|B)_{\rho}=0$ the rotated Petz recovery map
\begin{align} \label{eq_PetzRecMap2}
\cT_{B \to BC}^{[t]}\,: \, X_B \mapsto \rho_{BC}^{\frac{1+ \ci t}{2}}\left( \rho_B^{-\frac{1+\ci t}{2}} X_B \, \rho_B^{-\frac{1-\ci t}{2}} \otimes \id_C \right) \rho_{BC}^{\frac{1-\ci t}{2}} \quad \text{for} \quad t \in \R
\end{align}
satisfies~\eqref{eq_Markov2}, i.e., $\cT^{[t]}_{B\to BC}(\rho_{AB})=\rho_{ABC}$ for all $t \in \R$.
\end{theorem}
\vspace{-4mm}
\end{svgraybox}
\begin{proof}
One direction of the theorem is almost trivial. Suppose $\rho_{ABC}$ is a Markov chain. The data-processing inequality then shows that
\begin{align} \label{eq_oneWay}
I(A:C|B)_{\rho} 
=H(A|B)_{\rho} - H(A|BC)_{\rho} 
\leq  H(A|BC)_{\cT^{[t]}_{B\to BC}(\rho_{AB})} - H(A|BC)_{\rho} 
=0 \, .
\end{align}
The inequality step is justified by
\begin{align}
- H(A|BC)_{\rho} 
&= D(\rho_{ABC} \| \id_A \otimes \rho_{BC}) \\
&\geq D(\rho_{AB} || \id_A \otimes \rho_B) \\
&\geq D\big( \cT_{B \to BC}^{[t]}(\rho_{AB}) \| \id_A \otimes  \cT_{B \to BC}^{[t]}(\rho_{B})  \big) \\
&= - H(A|BC)_{\cT^{[t]}_{B\to BC}(\rho_{AB})} \, ,
\end{align}
where we used that $\tr_A \cT^{[t]}_{B\to BC}(\rho_{AB}) = \cT^{[t]}_{B\to BC}(\rho_{B})$.  
The final step in~\eqref{eq_oneWay} uses that $\rho_{ABC}$ is a Markov chain and hence $\rho_{ABC}= \cT^{[t]}_{B\to BC}(\rho_{AB})$. Together with the strong subadditivity of quantum entropy (see~\eqref{eq_SSAsecPinching}) this implies that $I(A:C|B)_{\rho}=0$.

The other direction, i.e., that $I(A:C|B)_{\rho}=0$ implies that $\rho_{ABC}$ is a Markov chain and that in such a case every rotated Petz recovery maps satisfies~\eqref{eq_Markov2} is more complicated to show. We postpone this proof to Section~\ref{sec_dpi} (see Remark~\ref{rmk_perfectRec}). \qed
\end{proof}

\begin{remark} \label{rmk_PetzTPCP}
The rotated Petz recovery map $\cT_{B \to BC}^{[t]}$ defined in~\eqref{eq_PetzRecMap2} is trace-preserving and completely positive for all $t \in \R$. That the map is completely positive is immediate. It is also trace preserving as
\begin{align}
\tr \, \cT_{B \to BC}^{[t]}(X_B) 
= \tr \, \rho_{BC} \big( \rho_B^{-\frac{1+\ci t}{2}} X_B \rho_B^{-\frac{1-\ci t}{2}} \otimes \id_C \big) 
= \tr \, X_B \, ,
\end{align}
where the first step uses the cyclic invariance of the trace and the final step uses two basic properties of the partial trace, i.e., for $X_{AB} \in \LL(A\otimes B)$ and $Y_A \in \LL(A)$ we have $\tr\, X_{AB} = \tr_A \, \tr_B\, X_{AB}$ and $\tr_B \, X_{AB} (Y_A \otimes \id_B) = \tr_B(X_{AB}) Y_A$.
\end{remark}

Theorem~\ref{thm_PetzCMI} is interesting as it links quantum Markov chains that are defined in an operational way (i.e., that parts of a composite system can be recovered by only acting on other parts) with an entropic quantity, the conditional mutual information. Entropy measures are well studied and obey many nice properties (as discussed in Section~\ref{sec_entropyMeas}). More concretely, Theorem~\ref{thm_PetzCMI} can be helpful in practice: Suppose you are given a tripartite state $\rho_{ABC}$ and want to determine if it is a quantum Markov chain or not. Theorem~\ref{thm_PetzCMI} tells us that all we need to do is to calculate the conditional mutual information $I(A:C|B)_{\rho}$. 

Theorem~\ref{thm_PetzCMI} links Markov chains and the conditional mutual information. The following result further deepens our understanding of Markov chains. It presents an algebraic characterization of quantum Markov chains~\cite{HJPW04}. 
\begin{svgraybox}
\vspace{-4mm}
\begin{theorem} \label{thm_algebraicQMC} \index{Markov chain!algebraic structure}
A state $\rho_{ABC} \in \St(A\otimes B \otimes C)$ is a Markov chain in order $A\leftrightarrow B\leftrightarrow C$  if and only if there exists a decomposition of the $B$ system as 
\begin{align}
B = \bigoplus_{j} b_j^L \otimes b_j^R
\end{align}
such that 
\begin{align} \label{eq_MarkovDec}
\rho_{ABC} = \bigoplus_j P(j) \rho_{A b_j^L} \otimes \rho_{b_j^R C} \, ,
\end{align}
with $\rho_{A b_j^L} \in \St(A \otimes b_j^L)$, $\rho_{b_j^R C} \in \St(b_j^R \otimes C)$ and a probability distribution $P$.
\end{theorem}
\vspace{-4mm}
\end{svgraybox}
\begin{proof}
One direction is trivial. If $\rho_{ABC}$ has the form~\eqref{eq_MarkovDec} we have $I(A:C|B)_{\rho}=0$. Theorem~\ref{thm_PetzCMI} then shows that $\rho_{ABC}$ is a Markov chain. It thus remains to show that any Markov chain can be written as~\eqref{eq_MarkovDec}.
For the channel $\cR_{B \to B} = \tr_C \circ \cR_{B \to BC}$ the Markov condition~\eqref{eq_Markov2} implies
\begin{align} \label{eq_hay1}
\cR_{B\to B}(\rho_{AB}) = \rho_{AB} \, .
\end{align}
 Let $M_A \in \Pos(A)$ such that $M_A \leq \id_A$ and define a state $\sigma_B \in \St(B)$ by
 \begin{align}
 p \sigma_B = \tr_A \, \rho_{AB} (M_A \otimes \id_B) \quad \text{with} \quad p = \tr\, \rho_{AB} (M_A \otimes \id_B) \, .
 \end{align}
 In case $p\ne 0$,~\eqref{eq_hay1} implies that $\cR_{B \to B}(\sigma_B) = \sigma_B$. Varying $M_A$ gives a family $\mathbb{M}(B)$ of states on $B$ that are invariant under $\cR_{B \to B}$.
 
Apply Theorem~9 from~\cite{HJPW04} (see also~\cite{koashi02}) gives a decomposition
\begin{align}
B = \bigoplus_{j} b_j^L \otimes b_j^R \, ,
\end{align}
such that every $\sigma_B \in \mathbb{M}(B)$ can be written
\begin{align}
\sigma_B = \bigoplus_j P(j,\sigma) \rho_j(\sigma) \otimes \omega_j \, ,
\end{align}
with $\rho_j(\sigma) \in \St(b_j^L)$, $\omega_j \in \St(b_j^R)$ and a probability distribution $P$. By definition of $\sigma_B$ this now implies
\begin{align} \label{eq_hayStep1}
\rho_{AB} = \bigoplus_j P(j) \rho_{Ab_j^L} \otimes \omega_{b_j^R} \, .
\end{align}
To see this, we define the map
\begin{align}
\cF_{B\to B}:  X_B \mapsto \bigoplus_j \tr_{b_j^R}(\Pi_j X_B \Pi_j) \otimes \omega_j \, ,
\end{align}
where $\Pi_j$ is the orthogonal projector onto the subspace $b_j^L\otimes b_j^R$. We then find for $M_A \in \Pos(A)$ such that $M_A \leq \id_A$ and $N_B \in \Pos(B)$ such that $N_B \leq \id_B$
\begin{align}
\tr \, \rho_{AB} (M_A \otimes N_B) 
&= p \, \tr\, \sigma_B N_B
= p \, \tr\, \cF_{B\to B}(\sigma_B) N_B
= p \, \tr\, \sigma_B \cF^\dagger_{B\to B}(N_B) \\
&=\tr \, \rho_{AB} (M_A \otimes \cF^\dagger_{B\to B}(N_B))
=\tr\, \cF_{B \to B}(\rho_{AB}) (M_A \otimes N_B) \, .
\end{align}
By linearity this is valid for all operators $M_A \otimes N_B$ such that we obtain
\begin{align}
\rho_{AB} = \cF_{B \to B}(\rho_{AB}) \, .
\end{align}
This now implies~\eqref{eq_hayStep1} since
\begin{align}
\cF_{B \to B}(X_B) = \bigoplus_j \tr_{b_j^R}\big( \Pi_j X_B \Pi_j \big) \otimes \omega_j \, .
\end{align}

Let $E$ be a environment such that by the Stinespring dilation (see Proposition~\ref{prop_stinespring}) we can express the recovery map $\cR_{B \to BC}$ as
\begin{align}\label{eq_hayStep11}
\cR_{B \to BC} \, : \, X_B \mapsto \tr_E U_{BCE} ( X_B \otimes \proj{0}_C \otimes \tau_E) U_{BCE}^\dagger \, ,
\end{align}
for $U_{BCE} \in \U(B\otimes C \otimes E)$ and $\tau_E \in \St(E)$. Since $\cR_{B \to BC}(\rho_{AB})=\rho_{ABC}$ and~\eqref{eq_hay1} we see that the unitary $U_{BCE}$ must be of the form
\begin{align} \label{eq_hayStep2}
U_{BCE} = \bigoplus_j \id_{b_j^L} \otimes U_j \, ,
\end{align}
for $U_j \in \U(b_j^R \otimes C\otimes E)$. Combining~\eqref{eq_hayStep1}~\eqref{eq_hayStep11}, and~\eqref{eq_hayStep2} shows that
\begin{align}
\rho_{ABC}
&= \cR_{B \to BC}(\rho_{AB}) \\
&= \tr_E \, U_{BCE} (\rho_{AB} \otimes \proj{0}_C \otimes \tau_E) U_{BCE}^\dagger \\
&= \bigoplus_j P(j) \rho_{Ab_j^L} \otimes \tr_E \, U_{j} (\omega_{b_j^R} \otimes \proj{0}_C \otimes \tau_E) U_{j}^\dagger \\
&= \bigoplus_j P(j) \rho_{Ab_j^L} \otimes  \rho_{b_j^R C} \, ,
\end{align}
which proves the assertion. \qed
\end{proof}

\section{Sufficient criterion for approximate recoverability} \label{sec_AQMC} \index{Markov chain!approximate} \index{approximate recoverability!sufficient condition}
This section deals with the question whether the properties of quantum Markov chains discussed in the previous section are robust. In particular we are interested in the question if the entropic characterization of Markov chains given by Theorem~\ref{thm_PetzCMI} is robust. That is, we would like to understand the entropic structure of tripartite density operators that have a small conditional mutual information. In particular, if it is possible to relate the conditional mutual information with a measure of how well the $C$-system can be recovered by only acting on the $B$-system with a recovery map. 

The following theorem~\cite{FR14,BHOS14,SFR15,wilde15,STH15,JRSWW15,SBT16} shows that whenever the conditional mutual information $I(A:C|B)_{\rho}$ of a quantum state $\rho_{ABC}$ is small, then the Markov condition~\eqref{eq_Markov2} approximately holds, i.e, there exists a recovery map from $B$ to $B\otimes C$ that approximately reconstructs $\rho_{ABC}$ from $\rho_{AB}$. This therefore justifies the definition of \emph{approximate quantum Markov chains} as tripartite states $\rho_{ABC}$ such that the conditional mutual information $I(A:C|B)_{\rho}$ is small.
\begin{svgraybox}
\vspace{-4mm}
\begin{theorem} \label{thm_FR}\index{strong subadditivity}
Let $\rho_{ABC} \in \St(A\otimes B \otimes C)$. Then
\begin{align} \label{eq_FR}
I(A:C|B)_{\rho} \geq \MD\big(\rho_{ABC} \| \bar \cT_{B \to BC}(\rho_{AB}) \big) \, ,
\end{align}
with the rotated Petz recovery map\index{rotated Petz recovery map}
\begin{align} \label{eq_FRrecMap}
\bar \cT_{B\to BC}= \int_{-\infty}^{\infty} \di t \beta_0( t) \cT^{[t]}_{B\to BC} \, ,
\end{align}
where $\beta_0$ and $ \cT^{[t]}_{B\to BC}$ are defined in~\eqref{eq_beta_0} and~\eqref{eq_PetzRecMap2}, respectively.
\end{theorem}
\vspace{-4mm}
\end{svgraybox}
\begin{proof}
This theorem follows from Theorem~\ref{thm_strengthened_mono} by choosing $\rho=\rho_{ABC}$, $\sigma = \id_A \otimes \rho_{BC}$, and $\cE = \tr_C$.
 \qed
\end{proof}
The recovery map $\bar \cT_{B \to BC}$ defined in~\eqref{eq_FRrecMap} that satisfies~\eqref{eq_FR} fulfills several nice properties:\index{rotated Petz recovery map!properties}
\begin{enumerate}
\item It is trace-preserving and completely positive (see Remark~\ref{rmk_PetzTPCP}).\vspace{1mm}
\item It is \emph{explicit}.\vspace{1mm}
\item It is \emph{universal}, i.e., it depends on $\rho_{BC}$ only.\vspace{1mm}
\item It satisfies $\bar \cT_{B\to BC}(\rho_B) = \rho_{BC}$.
\end{enumerate}
Theorem~\ref{thm_FR} is of interest for various reasons. First and foremost, it shows that all tripartite density operators $\rho_{ABC}$ with a small conditional mutual information $I(A:C|B)_{\rho}$ are approximately recoverable in the sense that $\bar \cT_{B\to BC}(\rho_{ABC}) \approx \rho_{ABC}$ for the recovery map $\bar \cT_{B \to BC}$ defined in~\eqref{eq_FRrecMap}. This justifies the definition of approximate quantum Markov chains as state that have a small conditional mutual information. In Section~\ref{sec_smallCMInotMarkov} we will see that approximate Markov chains can be far from any Markov chain, with respect to the trace distance. 

Second, Theorem~\ref{thm_FR} immediately implies the celebrated \emph{strong subadditivity of quantum entropy}~\cite{LieRus73_1,LieRus73}, i.e.
\begin{align} \label{eq_SSA} \index{strong subadditivity}
I(A:C|B)_{\rho}\geq 0 \, ,
\end{align}
by recalling the nonnegativity of the measured relative entropy (see Proposition~\ref{prop_MeasRelEnt}). Theorem~\ref{thm_FR} thus is a strengthening of SSA.

\begin{remark} \label{rmk_FR_classical}
Inequality~\eqref{eq_FR} is tight in the classical case. To see this, we recall that according to~\eqref{eq_ClassicalCMIRelEnt} 
\begin{align} \label{eq_classicalPetz}
\rho_{ABC} \text{ is classical } \quad \implies \quad  I(A:C|B)_{\rho} = D(\rho_{ABC} \| \cT_{B \to BC}(\rho_{AB}))\, .
\end{align}
We recall that the state $\rho_{ABC}$ is classical if it can be written as 
\begin{align} \label{eq_classicalState} \index{classical state}
\rho_{ABC} = \sum_{a,b,c} P_{ABC}(a,b,c) \proj{a}_A \otimes \proj{b}_B \otimes \proj{c}_C \, ,
\end{align} 
for some probability distribution $P_{ABC}$. 
 Since for classical states the measured relative entropy coincides with the relative entropy and since the rotated Petz recovery map $\bar \cT_{B \to BC}$ defined in~\eqref{eq_FRrecMap} simplifies to the Petz recovery map $\cT_{B \to BC}$ defined in~\eqref{eq_PetzRecMap}, we see that~\eqref{eq_FR} holds with equality if $\rho_{ABC}$ is a classical state.  
\end{remark}

\begin{remark} \label{rmk_FR_optimal}
Theorem~\ref{thm_FR} is essentially optimal. It has been shown~\cite[Section~5]{fawzi17} that there exist tripartite density operators $\rho_{ABC} \in \St(A\otimes B \otimes C)$ such that
\begin{align}
I(A:C|B)_{\rho} < \min \limits_{\cR_{B \to BC} } \{ D\big(\rho_{ABC} \| \cR_{B \to BC}(\rho_{AB}) \big) \, : \, \cR_{B \to BC} \in \TPCP(B,B\otimes C) \} \, .
\end{align}
This shows that Theorem~\ref{thm_FR} is no longer valid when replacing the measured relative entropy in~\eqref{eq_FR} with a relative entropy --- even if we optimize over all possible recovery maps.
\end{remark}

\begin{remark} \label{rmk_multiLetter}
Remark~\ref{rmk_FR_optimal} just above shows that it is not possible to bound the conditional mutual information of a tripartite state $\rho_{ABC}$ from below by the relative entropy between $\rho_{ABC}$ a a recovered state $\cR_{B\to BC}(\rho_{AB})$. This, however, becomes possible if we consider a multi-letter formula. More precisely, it was shown~\cite[Theorem 12]{hirche17} (see also~\cite[Theorem~1]{BHOS14} and~\cite[Proposition~3.1]{STH15}) that 
\begin{align}
I(A:C|B)_{\rho} 
&\geq \lim \sup_{n \to \infty} \frac{1}{n} D\Big(\rho_{ABC}^{\otimes n} \| \int_{-\infty}^{\infty} \di t  \beta_0(t) \cT^{[t]}_{B \to BC}(\rho_{AB})^{\otimes n} \Big) \, ,
\end{align}
where $\beta_0$ and $ \cT^{[t]}_{B\to BC}$ are defined in~\eqref{eq_beta_0} and~\eqref{eq_PetzRecMap2}, respectively.

\end{remark}

\subsection{Approximate Markov chains are not necessarily close to Markov chains} \label{sec_smallCMInotMarkov}
Approximate Markov chains are tripartite states $\rho_{ABC}$ with a small conditional mutual information. Theorem~\ref{thm_FR} shows that such states are approximately recoverable in the sense that there exists a recovery map $\cR_{B\to BC}$ such that~\eqref{eq_Markov2} approximately holds. Surprisingly, approximate quantum Markov chains are, however, not necessarily close in trace distance to any Markov chain~\cite{CSW12,ILW08}. To see this, let 
\begin{align}
\Delta(\rho,\sigma):=\frac{1}{2} \norm{\rho-\sigma}_1
\end{align}
denote the \emph{trace distance} between $\rho$ and $\sigma$.~\index{trace distance} 

\begin{svgraybox}
\vspace{-4mm}
\begin{proposition} \label{prop_notCloseToMarkov}
For any $d>1$, there exist states $\rho_{ABC} \in \St(A\otimes B \otimes C)$ with $\dim(A)=\dim(C)=d$ such that
\begin{align}
I(A:C|B)_{\rho} \leq \frac{2}{d-1} \log d \qquad \text{and} \qquad \min_{\mu \in \QMC} \Delta(\rho_{ABC} , \mu_{ABC}) \geq \frac{1}{2} \, .
\end{align}
\end{proposition}
\vspace{-4mm}
\end{svgraybox}
Proposition~\ref{prop_notCloseToMarkov} shows that there exist tripartite density operators with an arbitrarily small conditional mutual information, whose distance to any Markov chain, however, is large. This shows that approximate quantum Markov chains are not close to Markov chains.

\begin{proof}
Let $\rho_{S_1,\ldots S_d} = \proj{\psi}_{S_1,\ldots S_d}$ on $S_1 \otimes \cdots \otimes S_d$ with $\dim S_k = d >1$ for all $k=1,\ldots,d$, where
\begin{align}
\ket{\psi}_{S_1,\ldots S_d} :=\sqrt{\frac{1}{d!}} \sum_{\pi \in \cS_d} \mathrm{sign}(\pi) \ket{\pi(1)} \otimes \ldots \otimes \ket{\pi(d)}
\end{align}
is the \emph{Slater determinant}\index{Slater determinant}, $\cS_d$ denotes the group of permutations of $d$ objects, and $\mathrm{sign}(\pi):=(-1)^L$, where $L$ is the number of transpositions in a decomposition of the permutation $\pi$.
The chain rule for the mutual information shows that 
\begin{align}
I(S_1:S_2 \ldots S_d)_{\rho} =  \sum_{k=2}^d I(S_1:S_k | S_2 \ldots S_{k-1})_{\rho}\leq 2 \log d \, ,
\end{align}
where the final step follows by the trivial upper bound for the conditional mutual information.
By the nonnegativity of the mutual information, there exists $k \in \{2,\ldots,d\}$ such that 
\begin{align} \label{eq_CMIsmall}
I(S_1:S_k| S_2 \ldots S_{k-1})_{\rho}\leq \frac{2}{d-1} \log d \, ,
\end{align}
which can be arbitrarily small as $d$ gets large. The density operator $\rho_{S_1,\ldots S_d}$ is chosen such that the reduced state $\rho_{S_1 S_k}$ is the \emph{antisymmetric state} on $S_1 \otimes S_k$ that is far from separable~\cite[p.~53]{brandao16}. More precisely, for any separable state $\sigma_{S_1 S_k}$ on $S_1 \otimes S_k$  we have $\Delta(\rho_{S_1 S_k},\sigma_{S_1 S_k}) \geq \frac{1}{2}$.

Theorem~\ref{thm_algebraicQMC} ensures that that for any state $\mu_{S_1 \ldots S_k}$ on $S_1 \otimes \cdots \otimes S_k$ that forms a Markov chain in order $S_1 \leftrightarrow S_2 \otimes \ldots \otimes S_{k-1} \leftrightarrow S_k$, it follows that its reduced state $\mu_{S_1 S_k}$ on $S_1 \otimes S_k$ is separable. The monotonicity of the trace distance under trace-preserving completely positive maps~\cite[Theorem~9.2]{nielsenChuang_book} then implies
\begin{align}
\Delta(\rho_{S_1 \cdots S_k},\mu_{S_1 \cdots S_k}) \geq \Delta(\rho_{S_1 S_k},\mu_{S_1 S_k}) \geq \frac{1}{2} \, .
\end{align}
This shows that the state $\rho_{S_1 \cdots S_k}$, despite having a conditional mutual information that is arbitrarily small (see~\eqref{eq_CMIsmall}), is far from any Markov chain. Relabeling $A=S_1$, $C=S_k$, and $B=S_2 \otimes \ldots \otimes S_{k-1}$ finally completes the proof. \qed
\end{proof}

\section{Necessary criterion for approximate recoverability} \label{sec_UB_CMI} \index{approximate recoverability!necessary condition}
Theorem~\ref{thm_FR} shows that a small conditional mutual information is a sufficient condition for a state to be approximately recoverable. In other words,~\eqref{eq_FR} gives an entropic characterization for the set of tripartite states that can be approximately recovered. 
In this section, we are interested in an opposite statement. This corresponds to an inequality that bounds the distance between $\rho_{ABC}$ and any reconstructed state $\cR_{B \to BC}(\rho_{AB})$ from below with an entropic functional of $\rho_{ABC}$ and the recovery map $\cR_{B \to BC}$ that involves the conditional mutual information. Such an inequality is the converse to~\eqref{eq_FR}, and gives a necessary condition for approximate recoverability. Furthermore it gives an entropic characterization for the set of tripartite states that cannot be approximately recovered~\cite{sutter17}.

For any $\cE \in \TPCP(A,A)$ we denote by $\mathrm{Inv}(\cE)$ the set of density operators $\tau \in \St(A)$ which are left invariant under the action of $\cE$, i.e., 
\begin{align} \label{eq_invSet}
  \mathrm{Inv}(\cE) := \{\tau \in \St(A): \, \cE(\tau) = \tau\} \ .
\end{align}
We may now quantify the deviation of any state $\rho \in \St(A)$ from the set $\mathrm{Inv}(\cE)$ by new entropic quantity.
\begin{svgraybox}
\vspace{-4mm}
\begin{definition}
For $\alpha \in (\frac{1}{2},1)\cup (1,\infty)$, $\rho \in \St(A)$, $\cE \in \TPCP(A,A)$ and $\mathrm{Inv}(\cE)$ given by~\eqref{eq_invSet}, we define
\begin{align} \label{eq_Lambda}
  \Lambda_{\alpha}(\rho \| \cE) := \inf_{\tau \in \mathrm{Inv}(\cE)}  D_{\alpha}(\rho \| \tau) \ .
\end{align}
\end{definition}
\vspace{-4mm}
\end{svgraybox}
We further denote the limit cases
\begin{align}
 \Lambda_{\max}(\rho \| \cE) 
 :=  \lim_{\alpha \to \infty} \inf_{\tau \in \mathrm{Inv}(\cE)}  D_{\alpha}(\rho \| \tau)
= \inf_{\tau \in \mathrm{Inv}(\cE)}  D_{\max}(\rho \| \tau) \, ,
\end{align}
where in the final step we are allowed to interchange the infimum and the limit as the sequence $\{D_{\alpha}(\rho\|\tau) \}_{\alpha}$ is monotonically increasing (due to Proposition~\ref{lem_monoMinimal}) and hence by Dini's theorem~\cite{babyRudin} it converges uniformly in $\tau$.
By the same arguments we also see that
\begin{align}
\Lambda(\rho \| \cE) 
:= \lim_{\alpha \to 1} \inf_{\tau \in \mathrm{Inv}(\cE)}  D_{\alpha}(\rho \| \tau)
= \inf_{\tau \in \mathrm{Inv}(\cE)}  D(\rho \| \tau)\, .
\end{align}

The $\Lambda_{\alpha}$-quantity has the property that it is zero if and only if $\cE$ leaves $\rho$ invariant (see~\eqref{eq_DaZero}), i.e., 
\begin{align} \label{eq_LambdaZERO}
  \Lambda_{\alpha}(\rho \| \cE)  = 0 \quad \iff \quad \cE(\rho) = \rho \ .
\end{align}
We can now state the main result of this section which gives a necessary criterion for approximate recoverability~\cite{sutter17}.
\begin{svgraybox}
\vspace{-4mm}
\begin{theorem} \label{thm_CMI_UB}
Let $\rho_{ABC}\in\St(A\otimes B \otimes C)$ and $\cR_{B \to BC} \in \TPCP(B,B\otimes C)$. Then
\begin{align} \label{eq_CMI_UB}
D\big(\rho_{ABC} \| \cR_{B\to BC}(\rho_{AB}) \big)  \geq I(A:C|B)_{\rho} -  \Lambda_{\max}(\rho_{AB}\| \cR_{B \to B}) \, ,
\end{align}
where $\cR_{B \to B} := \tr_C \circ \cR_{B \to B C}$ is the reduction of $\cR_{B \to B C}$ to the output space~$B$. 
\end{theorem}
\vspace{-4mm}
\end{svgraybox}

Before commenting on this result let us prove it. To do so we recall that the conditional mutual information of a tripartite density operator is bounded from above by the smallest relative entropy distance to Markov chains. More precisely, we have the following upper bound for the conditional mutual information~\cite[Theorem~4]{ILW08}.
\begin{lemma} \label{lem_winter}
Let $\rho_{ABC} \in \St(A\otimes B \otimes C)$. Then
\begin{align} \label{eq_winter}
I(A:C|B)_{\rho} \leq \min_{\mu  \in \mathrm{QMC}}D(\rho_{ABC} \| \mu_{ABC}) \, .
\end{align}
\end{lemma}
\begin{proof}[Lemma~\ref{lem_winter}]
By definition of the relative entropy and the conditional mutual information we find for all $\rho_{ABC} \in \St(A\otimes B \otimes C)$ and all $\mu_{ABC} \in \QMC(A\otimes B \otimes C)$
\begin{align}
D(\rho_{ABC} \| \mu_{ABC}) + D(\rho_{B} \| \mu_{B}) - D(\rho_{AB} \| \mu_{AB}) - D(\rho_{BC} \| \mu_{BC})
&= I(A:C|B)_{\rho} + \nu \, ,
\end{align}
where 
\begin{align}
\nu := \tr \, \rho_{ABC} \log \mu_{ABC} +\tr \, \rho_B \log \mu_B - \tr \, \rho_{AB} \log \mu_{AB} - \tr \, \rho_{BC} \log \mu_{BC} \, .
\end{align}
The algebraic structure of Markov chains predicted by Theorem~\ref{thm_algebraicQMC} shows that 
\begin{align}
\mu_{ABC} = \bigoplus_j P(j) \mu_{Ab_j^L} \otimes \mu_{b_j^R C} \quad \text{for} \quad B  = \bigoplus_{j} b_j^L \otimes b_j^R \, ,
\end{align}
with $\mu_{Ab_j^L}  \in \St(A \otimes b_j^L)$ and $\mu_{b_j^R C} \in \St(b_j^R \otimes C)$. A simple calculation then shows that $\nu=0$ and thus
\begin{align}
I(A:C|B)_{\rho}  = D(\rho_{ABC} \| \mu_{ABC}) + D(\rho_{B} \| \mu_{B}) - D(\rho_{AB} \| \mu_{AB}) - D(\rho_{BC} \| \mu_{BC}) \, .
\end{align}
The nonnegativity of the relative entropy (see Proposition~\ref{prop_relEntProperties}) guarantees that $D(\rho_{BC} \| \mu_{BC}) \geq 0$ and by the DPI (see Proposition~\ref{prop_relEntProperties}) we have $D(\rho_B \| \mu_B) \leq D(\rho_{AB} \| \mu_{AB})$. This then proves the assertion.  \qed
\end{proof}

For the proof of Theorem~\ref{thm_CMI_UB} we require one more lemma that relates the distance to Markov chains with the $\Lambda_{\max}$-quantity defined in~\eqref{eq_Lambda}.
 \begin{lemma} \label{lem_conn}
   Let $\rho_{AB} \inÊ\St(A\otimes B)$ and $\cR_{B \to B C} \in \TPCP(B,B\otimes C)$. Then
\begin{align}
     \inf_{\mu \in \mathrm{MC}} D_{\max}\bigl( \cR_{B \to B C}(\rho_{A B}) \| \mu_{A B C} \bigr) \leq \Lambda_{\max}( \rho_{A B} \|  \cR_{B \to B}) \, ,
\end{align}
   where $\cR_{B \to B} := \tr_C \circ \cR_{B \to B C}$ is the reduction of $\cR_{B \to B C}$ to the output space~$B$. 
 \end{lemma}
\begin{proof}[Proof of Lemma~\ref{lem_conn}]
The DPI for the max-relative entropy~\cite{datta09,marco_book} implies that
  \begin{align}
    &\inf_{\mu_{ABC}} \{ D_{\max}\bigl( \cR_{B \to BC}(\rho_{A B}) \| \mu_{A B C} \bigr) : \,  \mu_{ABC}\in \QMC \}  \nonumber \\
    &\hspace{8mm} \leq   \inf_{ \tau_{AB}}\{ D_{\max}\bigl( \cR_{B \to BC}(\rho_{A B}) \| \cR_{B \to B C}(\tau_{A B}) \bigr): \cR_{B\to BC}(\tau_{AB}) \in \QMC, \tau_{AB} \in \St(A\otimes B) \}  \\
    &\hspace{8mm} \leq \inf_{\tau_{AB}} \{ D_{\max}(\rho_{A B} \|\tau_{A B} ) : \cR_{B\to BC}(\tau_{AB}) \in \QMC, \tau_{AB} \in \St(A\otimes B) \}  \, . \label{eq_Dmaxrhomu}
  \end{align}
  The strong subadditivity of quantum entropy (see~\eqref{eq_SSA}) implies that $H(A|BC)_{\cR_{B \to BC}(\tau_{A B})} \geq H(A|B)_{\tau_{A B}}$ for any $\tau_{AB} \in \St(A\otimes B)$ and hence
  \begin{align}
    \tau_{A B} \in \mathrm{Inv}(\cR_{B \to B})
    & \quad \implies \quad
    H(A|BC)_{\mu} \geq H(A|B)_{\mu}  \quad \text{for} \quad \mu_{A B C} = \cR_{B \to B C}(\tau_{A B}) \ . \label{eq_implic1}
  \end{align}
  The strong subadditivity of quantum entropy together with the inequality on the right-hand side of~\eqref{eq_implic1} implies that  $I(A : C | B)_{\mu} = 0$ which means that $\mu \in \mathrm{MC}$ and hence
  \begin{align}
      \tau_{A B} \in \mathrm{Inv}(\cR_{B \to B})
      \quad \implies \quad
      \cR_{B \to B C}(\tau_{A B}) \in \QMC \, .
  \end{align}
This implication now shows that
  \begin{align}
    \Lambda_{\max}( \rho_{A B} \| \cR_{B \to B})
    & = 
     \inf_{\tau_{AB}} \{ D_{\max}(\rho_{A B} \|\tau_{A B} ) :\, \tau_{AB} \in \mathrm{Inv}(\cR_{B \to B}) \} \\
     & \geq 
     \inf_{\tau_{AB}} \{ D_{\max}(\rho_{A B} \| \tau_{A B})  : \, \cR_{B \to B C}(\tau_{AB}) \in \QMC, \tau_{AB} \in \St(A\otimes B) \} \, .
  \end{align}
  Combining this with~\eqref{eq_Dmaxrhomu} completes the proof.  \qed
\end{proof}

\begin{proof}[Theorem~\ref{thm_CMI_UB}]
Let $\mu_{ABC} \in \QMC(A\otimes B \otimes C)$. Combining Lemma~\ref{lem_winter} with Lemma~\ref{lem_triangleD} applied for $\alpha=1$, $\rho=\rho_{ABC}$, $\sigma =Ê\mu_{ABC}$ and $\omega = \cR_{B \to BC}(\rho_{AB})$ gives
\begin{align} \label{eq_midStep}
D\big(\rho_{ABC} \| \cR_{B\to BC}(\rho_{AB}) \big)  \geq I(A:C|B)_{\rho} -  \inf_{\mu \in \QMC}D_{\max}\big(\cR_{B\to BC}(\rho_{AB}) \| \mu_{ABC}  \big) \, .
\end{align}
Lemma~\ref{lem_conn} then proves the assertion of Theorem~\ref{thm_CMI_UB}. We note that~\eqref{eq_midStep} is stronger than~\eqref{eq_CMI_UB} and therefore may be of independent interest. \qed
\end{proof}

The remaining part of this section is dedicated to comments on Theorem~\ref{thm_CMI_UB}. In particular we will discuss the tightness of~\eqref{eq_CMI_UB} and the role of the $\Lambda_{\max}$-term. 

\begin{remark}
In this remark we discuss cases where the $\Lambda_{\max}$-term vanishes.
A recovery map $\cR_{B \to B C}$ generally not only reads the content of system $B$ in order to generate $C$, but  also disturbs it. $\Lambda_{\max}$ quantifies the amount of this disturbance of $B$, taking system $A$ as a reference. This is the operational significance of the $\Lambda_{\max}$-quantity. In particular,~\eqref{eq_LambdaZERO} directly implies that $\Lambda_{\max}(\rho_{AB}\| \cR_{B\to B})=0$ if $\cR_{B \to B C}$ is ``read only'' on $B$, i.e., if $\rho_{AB}= \cR_{B\to B}(\rho_{AB})$. Inequality~\eqref{eq_CMI_UB} then simplifies to\index{read only recovery map}
\begin{align}
D\big(\rho_{ABC} \| \cR_{B\to BC}(\rho_{AB}) \big) \geq I(A:C|B)_{\rho}   \, .
\end{align}
We further note that in case $\cR_{B\to BC}$ is a recovery map that is ``read only'' on $B$ its output state $\sigma_{ABC}:=\cR_{B\to BC}(\rho_{AB})$ is a Markov chain since
\begin{align} \label{eq_whyMarkov}
H(A|B)_{\rho} \leq H(A|BC)_{\sigma} \leq H(A|B)_{\sigma}=H(A|B)_{\rho} \, ,
\end{align}
where the two inequality steps follow from the DPI applied for $\cR_{B\to BC}$ and $\tr_C$, respectively and hence $I(A:C|B)_{\sigma}= H(A|B)_{\sigma} - H(A|BC)_{\sigma} = 0$.
\end{remark}

\subsection{Tightness of the necessary criterion}
It is legitimate to ask if Theorem~\ref{thm_CMI_UB} is tight. To answer this question we need to have a better understanding about the $\Lambda_{\max}$-term.
Combining~\eqref{eq_FR} with~\eqref{eq_CMI_UB} gives
\begin{align}
\MD\big(\rho_{ABC} \| \bar \cT_{B\to BC}(\rho_{AB}) \big) & \leq I(A:C|B)_{\rho} \label{eq_tight1} \\
&\leq \min_{\cR_{B\to BC}}\left \lbrace D\big(\rho_{ABC} \| \cR_{B\to BC}(\rho_{AB}) \big) + \Lambda_{\max}(\rho_{AB}\| \cR_{B\to B}) \right \rbrace \, , \label{eq_ourBound}
\end{align}
where the recovery map $\bar \cT_{B \to BC}$ on the left-hand side is given by~\eqref{eq_FRrecMap} and the minimum is over all recovery maps $\cR_{B\to BC} \in \TPCP(B,B\otimes C)$. The main difference between the lower and upper bound for the conditional mutual information given by~\eqref{eq_tight1} and~\eqref{eq_ourBound}, respectively, is the $\Lambda_{\max}$-term. In the following we will show that this term is necessary (i.e., we cannot drop it) as well as optimal (i.e., we cannot replace it by a similar term that is strictly smaller). 

\subsubsection{Classical case}
Inequalities~\eqref{eq_tight1} and~\eqref{eq_ourBound} hold with equality in case $\rho_{ABC}$ is a classical state, i.e., it can be written as in~\eqref{eq_classicalState}. To see this, we first note that if $\rho_{ABC}$ is classical (in which case $\rho_{ABC}$ and all its marginals commute pairwise) a simple calculation (see~\eqref{eq_classicalPetz}) gives
\begin{align} \label{eq_classical_Petz}
I(A:C|B)_{\rho}
= D\big(\rho_{ABC} \| \cT_{B\to BC}(\rho_{AB}) \big)   \, ,
\end{align}
for the Petz recovery map $\cT_{B\to BC}$ defined in~\eqref{eq_PetzRecMap}. Furthermore, if $\rho_{ABC}$ is classical $\cT_{B \to BC}(\rho_{AB}) = \rho_{BC} \rho_{B}^{-1}  \rho_{AB}$. We further see that $\tr_C \cT_{B \to BC}(\rho_{AB}) = \cT_{B \to B}(\rho_{AB}) =\rho_{AB}$ and hence
\begin{align}
\Lambda_{\max}(\rho_{AB}Ê\| \cT_{B\to B}) =0 \, .
\end{align}
 This shows that in the classical case~\eqref{eq_ourBound} is an equality and that the Petz recovery map $\cT_{B \to BC}$ minimizes the right-hand side of~\eqref{eq_ourBound}. 
 Remark~\ref{rmk_FR_classical} explains why~\eqref{eq_tight1} holds with equality in the classical case.  

\subsubsection{Necessity of the $\Lambda_{\max}$-term}
It is natural to ask if tripartite states with a large conditional mutual information cannot be recovered approximately. 
Alternatively this can be phrases as the question if Theorem~\ref{thm_CMI_UB} remains valid when removing the $\Lambda_{\max}$-term. Just above we saw that this is the case for classical states. We next show, however, that the $\Lambda_{\max}$-quantity is necessary in general, i.e.,~\eqref{eq_CMI_UB} is false when dropping the $\Lambda_{\max}$-term.  

More precisely, in Appendix~\ref{cha_constructionExample} we construct a generic example showing that for any constant $\kappa <\infty$ there exists a classical state $\rho_{ABC}$ (i.e., a state of the form~\eqref{eq_classicalState}) such that
\begin{align} \label{eq_counterExample1}
\kappa \, D_{\max}\big(\rho_{ABC} \| \cR_{B \to BC}(\rho_{AB})\big) < I(A:C|B)_{\rho}   \, ,
\end{align}
for some recovery map $\cR_{B\to BC} \in \TPCP(B,B\otimes C)$ that satisfies $\cR_{B \to BC}(\rho_{B})=\rho_{BC}$. A similar construction (also given in Appendix~\ref{cha_constructionExample}) shows that there exists another classical state $\rho_{ABC}$ such that 
\begin{align}  \label{eq_counterExample2}
\kappa \, D_{\max}\big(  \cR_{B \to BC}(\rho_{AB}) \|  \rho_{ABC} \big)  < I(A:C|B)_{\rho}  \, ,
\end{align}
for some recovery map $\cR_{B\to BC} \in \TPCP(B,B\otimes C)$ that satisfies $\cR_{B \to BC}(\rho_{B})=\rho_{BC}$. 

These constructions (which are explained in detail in Appendix~\ref{cha_constructionExample}) reveal the following interesting observations:
\begin{enumerate}
\item The term $\Lambda_{\max}(\rho_{AB}\| \cR_{B\to B})$, which measures the deviation from a ``read only'' map on $B$, is necessary to obtain a lower bound on the relative entropy between a state and its reconstruction version. 
The example has an even stronger implication. It shows that the $\Lambda_{\max}$-term is necessary even if one tries to bound the max-relative entropy between a state and its reconstruction version, i.e., $D_{\max}(\rho_{ABC}\| \cR_{B\to BC}(\rho_{AB}))$from below.\footnote{The max-relative entropy and its properties are discussed in more detail in Section~\ref{sec_oneShot}. It is the largest sensible relative entropy measure.}  The two strict inequalities~\eqref{eq_counterExample1} and~\eqref{eq_counterExample2} show that the $\Lambda_{\max}$-term is also necessary if one would allow for swapping the two arguments of the relative (or even max-relative) entropy. Furthermore, restricting the set of recovery maps such that they satisfy $\cR_{B\to BC}(\rho_{B})=\rho_{BC}$ still requires the $\Lambda_{\max}$-term. \vspace{1mm}
\item The Petz recovery map can be far from being optimal --- even in the classical case. To see this we recall that  for classical states~\eqref{eq_classical_Petz} holds. Inequality~\eqref{eq_counterExample1} shows that there exists a recovery map that recovers $\rho_{ABC}$ much better from $\rho_{AB}$ than the Petz recovery map.\vspace{1mm}
\item Considering recovery maps that leave the $B$ system invariant (i.e., they only ``read'' the $B$-part) is a considerable restriction.\footnote{Recall that for recovery maps that leave the $B$ system invariant the $\Lambda_{\max}$-term vanishes as explained above.}
\end{enumerate}
 We refer to Appendix~\ref{cha_constructionExample} for more information about these examples.

\subsubsection{Optimality of the $\Lambda_{\max}$-term} 
In the previous section we saw that the $\Lambda_{\max}$-term in~\eqref{eq_CMI_UB} cannot be dropped. This raises the question if it is possible to replace this term by a strictly smaller term that has similar properties. The purpose of this section is to present two arguments why this is not the case. As a result,~\eqref{eq_CMI_UB} is close to optimal.

First, we show that the $\Lambda_{\max}$-term cannot be replaced by a $\Lambda_{\alpha}$-term for any $\alpha <\infty$. More precisely, for any $\alpha <\infty$, we construct a tripartite density operator $\rho_{ABC} \in \St(A\otimes B \otimes C)$ and a recovery map $\cR_{B \to BC} \in \TPCP(B,B\otimes C)$ such that 
\begin{align} \label{eq_TSquantum_intro}
D\big(\rho_{ABC} \| \cR_{B\to BC}(\rho_{AB}) \big)  < I(A:C|B)_{\rho} -  \Lambda_{\alpha}(\rho_{AB}\| \cR_{B \to B}) \, .
\end{align}
The construction is explained in Appendix~\ref{app_Ex_tight}.

Second, we show that the $\Lambda_{\max}$-term in~\eqref{eq_CMI_UB} cannot be defined as a distance between $\rho_{A B}$ and $\cR_{B \to B}(\rho_{A B})$. Recall that $ \Lambda_{\max}(\rho_{AB}\| \cR_{B \to B})$ quantifies the (max-relative entropy) distance between $\rho_{AB}$ and its closest state that is invariant under $\cR_{B\to B}$. A natural question is if~\eqref{eq_CMI_UB} remains valid if the $\Lambda_{\max}$-term is replaced by the (max-relative entropy) distance between $\rho_{A B}$ and $\cR_{B \to B}(\rho_{A B})$, i.e., $  D_{\max}(\rho_{A B} \| \cR_{B \to B}(\rho_{A B}))$. This however is ruled out. To see this we recall that by the example presented above in~\eqref{eq_counterExample1} there exists a tripartite state $\rho_{ABC} \in \St(A\otimes B \otimes C)$ and a recovery map $\cR_{B \to BC} \in \TPCP(B,B\otimes C)$ such that
\begin{align}
2 D_{\max}\bigl(\rho_{A B C} \| \cR_{B \to B C}(\rho_{A B})\bigr) < I(A: C | B)_{\rho}  \, .
\end{align}
The data-processing inequality for the max-relative entropy~\cite{datta09,marco_book} and the fact that the max-relative entropy cannot be smaller than the relative entropy (see~\eqref{eq_DminMax}) then imply
\begin{align}
   D\bigl(\rho_{A B C} \| \cR_{B \to B C}(\rho_{A B})\bigr) < I(A: C | B)_{\rho} - D_{\max}\bigl(\rho_{A B} \| \cR_{B \to B}(\rho_{A B}) \bigr) \, ,
\end{align}
which shows that~\eqref{eq_CMI_UB} is no longer valid for the modified $\Lambda_{\max}$-term described above.  


\section{Strengthened entropy inequalities} \label{sec_EntropyIneq}
It is well-known that several fundamental entropy inequalities useful in quantum information theory are intrinsically related. For example, it was shown that the following statements 
\begin{enumerate}
\item strong subadditivity of quantum entropy (see~\eqref{eq_SSA})\index{strong subadditivity}\vspace{1mm}
\item data processing inequality (see Proposition~\ref{prop_relEntProperties})\index{data processing inequality}\vspace{1mm}
\item concavity of conditional entropy (i.e., $\rho_{AB} \mapsto H(A|B)_{\rho}$ is concave)\vspace{1mm}
\item joint convexity of relative entropy (i.e., $(\rho,\sigma) \mapsto D(\rho\| \sigma)$ is convex)\vspace{1mm}
\item Lieb's triple operator inequality (see Theorem~\ref{thm_Lieb3})\vspace{1mm}
\item Lieb's concavity theorem (see Theorem~\ref{thm_liebConc})
\end{enumerate}
are all equivalent~\cite{Lieb73,rus02,tropp12}.\footnote{Equivalent means that every statement can be derived from every other one by simple manipulations only.} The main result of this section, i.e., Theorem~\ref{thm_FR}, presents a strengthening of SSA in terms of recovery maps.
It is therefore natural to ask if the other equivalent statements listed above can also be improved. This is the purpose of this section.

\subsection{Data processing inequality} \label{sec_dpi} \index{data processing inequality}
The \emph{data processing inequality} (DPI), also known as \emph{monotonicity of the relative entropy}, is one of the very fundamental entropy inequalities. It states that the relative entropy between two density operators cannot increase by applying a quantum channel to both operators~\cite{lindblad75,uhlmann77}. More precisely, for any $\rho \in \St(A)$, $\sigma \in \Pos(A)$, and $\cE \in \TPCP(A,B)$ we have
\begin{align} \label{eq_DPI}
D(\rho \| \sigma) \geq D\big(\cE(\rho) \| \cE(\sigma) \big) \, .
\end{align}
\begin{remark}
For $\rho=\rho_{ABC}$, $\sigma=\id_A \otimes \rho_{BC}$ and $\cE=\tr_C$,~\eqref{eq_DPI} simplifies to
\begin{align}
I(A:C|B)_{\rho} \geq 0 \, ,
\end{align}
which is the celebrated SSA, presented in Section~\ref{sec_QMC}.
This substitution provides a useful link between Section~\ref{sec_QMC} and this section.
\end{remark}
With this in mind the careful reader will notice that some inequalities discussed next are generalized versions of inequalities from Section~\ref{sec_QMC}.

The DPI is well studied. The following proposition gives necessary and sufficient conditions under which~\eqref{eq_DPI} holds with equality. 
\begin{proposition} \label{prop_PetzDPI} \index{Petz recovery map} 
Let $\rho \in \St(A)$, $\sigma \in \Pos(A)$ and $\cE \in \TPCP(A,B)$. Then the following are equivalent
\begin{enumerate}
\item $D(\rho \| \sigma) = D(\cE(\rho) \| \cE(\sigma))$ \label{it_ffirst}\vspace{1.5mm}
\item $\exists \, \cR_{\sigma,\cE} \in \TPCP(B,A)$ such that $(\cR_{\sigma,\cE}\circ \cE)(\rho)=\rho$ and $(\cR_{\sigma,\cE}\circ \cE)(\sigma)=\sigma$. \label{it_secc}
\end{enumerate}
In particular, $\cR_{\sigma,\cE}$ can always be chosen to be the rotated Petz recovery map, i.e.,
\begin{align} \label{eq_rotPetzt}
\cT_{\sigma,\cE}^{[t]} \,: \, X_B  \mapsto \sigma^{\frac{1+\ci t}{2}}\cE^\dagger\left(\cE(\sigma)^{-\frac{1+\ci t}{2}}X_B \,\cE(\sigma)^{-\frac{1-\ci t}{2}}\right)\sigma^{\frac{1-\ci t}{2}} \, .
\end{align}
\end{proposition}
\begin{proof}
To see that~\ref{it_secc}$\implies$~\ref{it_ffirst} is simple. The DPI shows that $D(\rho \| \sigma) \geq D(\cE(\rho) \| \cE(\sigma))$. The other direction also follows from the DPI since
\begin{align}
D(\cE(\rho) \| \cE(\sigma)) \geq D((\cR_{\sigma,\cE}\circ \cE)(\rho) \| (\cR_{\sigma,\cE}\circ \cE)(\sigma))=D(\rho\| \sigma) \, ,
\end{align}
where the final step uses~\ref{it_secc}.

It thus remains to show that~\ref{it_ffirst}$\implies$~\ref{it_secc}. This is more complicated. Note that it is immediate to verify that $(\cT^{[t]}_{\sigma,\cE} \circ \cE)(\sigma)= \sigma$ hence the nontrivial part is to show that $(\cT^{[t]}_{\sigma,\cE} \circ \cE)(\rho)= \rho$ which is done in Remark~\ref{rmk_perfectRec}. \qed
\end{proof}
\begin{exercise} \label{ex_DPI_SSA}
Convince yourself that~Proposition~\ref{prop_PetzDPI} implies Theorem~\ref{thm_PetzCMI}.
\end{exercise}

The following theorem is the main result of this chapter. It is a strengthening of the data processing inequality and a robust version of Proposition~\ref{prop_PetzDPI}.
\begin{svgraybox}
\vspace{-4mm}
\begin{theorem}\label{thm_strengthened_mono}
Let $\rho \in \St(A)$, $\sigma \in \Pos(A)$ such that $\rho \ll \sigma$, and $\cE \in \TPCP(A,B)$. Then
\begin{align}\label{eq:strengthened_mono_d}
  D(\rho\|\sigma) - D\left(\cE(\rho) \middle\| \cE(\sigma) \right) \geq D_{\mathbb{M}}\left(\rho \middle\|\bar \cT_{\sigma,\cE}\circ \cE(\rho) \right) ,
\end{align}
with the rotated Petz recovery map \index{rotated Petz recovery map}
\begin{align}\label{eq:rotated_petz}
\bar \cT_{\sigma,\cE} :=\int_{-\infty}^{\infty} \di t \beta_0( t)\,\cT_{\sigma,\cE}^{\left[t\right]} \, ,
\end{align}
where $\beta_0$ and $\cT_{\sigma,\cE}^{[t]}$ are defined in~\eqref{eq_beta_0} and~\eqref{eq_rotPetzt}, respectively.
\end{theorem}
\vspace{-4mm}
\end{svgraybox}
\begin{proof}
We first prove a slightly restricted version of Theorem~\ref{thm_strengthened_mono} where we suppose that $\cE$ is a partial trace. In a second step we then show how this statement can be generalized (using the Stinespring dilation) to an arbitrary channel $\cE$.

Let $\rho_{AB} \in \St(A\otimes B)$ and $\sigma_{AB} \in \Pos(A\otimes B)$ be such that $\rho_{AB} \ll \sigma_{AB}$. Let us recall the multivariate GT inequality (see Theorem~\ref{thm_GT_steinHirschman}) applied for $n=4$ and $p=2$. Using the concavity of the logarithm  and Jensen's inequality, it yields 
\begin{align}
\tr \, \ee^{H_1+H_2+H_3+H_4} \leq \int_{- \infty}^\infty  \di t \beta_0( t)\, \tr \, \ee^{H_1} \ee^{\frac{1 + \ci t}2 H_2} \ee^{\frac{1 + \ci t}2 H_3 } \ee^{H_4}\ee^{\frac{1 - \ci t}2 H_3 } \ee^{\frac{1 - \ci t}2 H_2} \, , \label{eq:GT_4matrix}
\end{align}
for $H_k \in \Her(A\otimes B)$ and $k \in [4]$. Moreover, by definition of the relative entropy for positive definite operators $\rho_{AB}$ and $\sigma_{AB}$, we have
\begin{align}
 D(\rho_{AB}\|\sigma_{AB}) - D(\rho_A \| \sigma_A) = D\bigl( \rho_{AB} \| \exp(\log \sigma_{AB} + \log \rho_A \otimes \id_B - \log \sigma_A \otimes \id_B) \big) \,. \label{eq:contractlog}
\end{align}
For positive semi-definite operators $\rho_{AB}$ and $\sigma_{AB}$, the Hermitian operators $\log \sigma_{AB}$, $\log \rho_A$ and $\log \sigma_A$ are well-defined under the convention $\log 0=0$. Under this convention, the above equality~\eqref{eq:contractlog} also holds for positive semi-definite operators as long as $\rho_{AB} \ll \sigma_{AB}$, which is required by the theorem.  
By the variational formula for the relative entropy (see Lemma~\ref{lem_varFormulaRelEnt}) we thus find 
 \begin{align}
  &D(\rho_{AB}\|\sigma_{AB}) - D(\rho_A \| \sigma_A) \nonumber  \\
 &= \!\!\!\!\!\sup_{\omega_{AB} \in \Poss(A\otimes B)}\!\!\! \left \lbrace \tr \, \rho_{AB} \log \omega_{AB} \!+\! 1 \!-\! \tr \exp(\log \sigma_{AB} \!+\! \log \rho_A \otimes \id_B \!-\! \log \sigma_A \otimes \id_B \!+\! \log \omega_{AB}) \right \rbrace \\
 &\geq\!\!\!\!\!\!\sup_{\omega_{AB} \in \Poss(A\otimes B)}\!\!\! \left \lbrace \! \tr \, \rho_{AB} \log \omega_{AB} \!+\! 1 \!-\!  \int_{- \infty}^\infty \!\!\!\!  \di t\beta_0( t)\, \tr \, \sigma_{AB}^{\frac{1+\ci t}{2}} \! \left(\!  \sigma_A^{-\frac{1+\ci t}{2}} \rho_A \sigma_A^{-\frac{1-\ci t}{2}} \! \otimes \id_B \! \right)\!  \sigma_{AB}^{\frac{1-\ci t}{2}} \omega_{AB}\!  \right \rbrace \\
 &= \MD \left(\rho_{AB} \middle\| \int_{- \infty}^\infty \di t  \beta_0( t) \,  \sigma_{AB}^{\frac{1+\ci t}{2}}  \left( \sigma_A^{-\frac{1+\ci t}{2}} \rho_A \sigma_A^{-\frac{1-\ci t}{2}}  \otimes \id_B \right) \sigma_{AB}^{\frac{1-\ci t}{2}}  \right) \\
 &= \MD\big(\rho_{ABC} \| \bar \cT_{\sigma_{AB},\tr_B}(\rho_A) \big)  \ , \label{eq_partMain1}
\end{align}
where the single inequality step follows by the four matrix extension of the GT inequality in~\eqref{eq:GT_4matrix}.
The penultimate step uses the variational formula for the measured relative entropy given in Lemma~\ref{lem_varFormulaMeasRelEnt}.

 Let us introduce the Stinespring dilation of $\cE$, denoted $V$, and the states $\rho_{AB} = V \rho V^{\dag}$, $\quad \sigma_{AB} = V \sigma V^{\dag}$ such that $\cE(\rho) = \rho_A$ and $\cE(\sigma) = \sigma_A$. Then, using the fact that the relative entropy is invariant under isometries (see Proposition~\ref{prop_relEntProperties}), we have
  \begin{align}
  D(\rho\|\sigma) - D\big(\cE(\rho) \| \cE(\sigma) \big) &= D(\rho_{AB}\|\sigma_{AB}) - D(\rho_A\|\sigma_A) \\
  &\geq \MD\big(\rho_{AB} \|\bar \cT_{\sigma_{AB},\tr_B}(\rho_A) \big)  \\
  &= \MD\big(\rho \| (\bar \cT_{\sigma,\cE}\circ \cE)(\rho) \big) ,
  \end{align}
  where the inequality is due to~\eqref{eq_partMain1} and the last equality uses again invariance under isometries and the fact that for all $t \in \R$ and $X_A \in \Pos(A)$
  \begin{align}
    V^{\dag} \cT_{\sigma_{AB},\tr_B}^{[t]} (X_A) V 
    &= V^{\dag} V \sigma^{\frac{1+\ci t}{2}} V^{\dag} \left( \cE(\sigma)^{-\frac{1+\ci t}{2}} ( X_A ) \cE(\sigma)^{-\frac{1-\ci t}{2}} \otimes \id_B \right) V \sigma^{\frac{1-\ci t}{2}} V^{\dag} V \\
    &= \sigma^{\frac{1+\ci t}{2}} \cE^{\dag} \left( \cE(\sigma)^{-\frac{1+\ci t}{2}} ( X_A ) \cE(\sigma)^{-\frac{1-\ci t}{2}} \right) \sigma^{\frac{1-\ci t}{2}}\\
    &= \bar \cT_{\sigma,\cE}^{[t]}(X_A) \,. \label{eq_mapIsoInv}
  \end{align}
This therefore completes the proof. \qed
\end{proof}

\begin{exercise} \label{ex_redMainRes}
Convince yourself that Theorem~\ref{thm_FR} follows immediately from Theorem~\ref{thm_strengthened_mono} by choosing $\rho=\rho_{ABC}$, $\sigma = \id_A \otimes \rho_{BC}$, and $\cE = \tr_C$.
\end{exercise}

The recovery map $\bar \cT_{\sigma,\cE}$ from Theorem~\ref{thm_strengthened_mono} satisfies many desirable properties~\cite{petz_statbook,wilde15,JRSWW15}:\index{rotated Petz recovery map!properties}
\begin{enumerate}
\item It is \emph{trace-non-increasing} and \emph{completely positive}.\footnote{In case $\cE(\sigma) \in \Poss(B)$ the recovery map $\bar \cT_{\sigma,\cE}$ is trace-preserving.}\vspace{1mm}
\item It is \emph{explicit}. \vspace{1mm}
\item It is \emph{universal}, i.e., it depends on $\sigma$ and $\cE$ only. (It is independent of $\rho$.)\vspace{1mm}
\item It satisfies $(\bar \cT_{\sigma,\cE} \circ \cE)(\sigma)=\sigma$, i.e., it perfectly recovers $\sigma$ from $\cE(\sigma)$.\vspace{1mm}
\item It features a \emph{normalization} property. For $\cE=\cI$ we have $\bar \cT_{\sigma,\cI}(\cdot)=\Pi_{\sigma} (\cdot) \Pi_{\sigma}$, where $\Pi_{\sigma}$ denotes the projector onto the support of $\sigma$. Thus, in case $\sigma$ has full support $\bar \cT_{\sigma,\cI}$ is the identity map.\vspace{1mm}
\item It has a \emph{stabilization} property. For any $\omega \in \Poss(R)$, where $R$ denotes a reference system we have $ \bar \cT_{\sigma \otimes \omega,\cE \otimes \cI_R} = \bar \cT_{\sigma,\cE} \otimes \cI_R$.
\end{enumerate}
\begin{exercise} \label{ex_propRecMap}
Verify the six properties stated above.
\end{exercise}

Using similar techniques as in the proof of Theorem~\ref{thm_strengthened_mono}, we can derive another strengthening of the data processing inequality~\cite{JRSWW15}. 
\begin{proposition} \label{prop_strengthened_mono_integralOUT}
Let $\rho,\sigma \in \Pos(A)$ such that $\rho \ll \sigma$, $\tr\, \rho=1$, $\cE \in \TPCP(A,B)$, and $\beta_0$ defined in~\eqref{eq_beta_0}. Then
\begin{align}\label{eq:strengthened_mono_d_IntegralOUT}
  D(\rho\|\sigma) - D\big(\cE(\rho) \| \cE(\sigma) \big) \geq - \int_{-\infty}^{\infty} \di t \beta_0( t) \log F\big(\rho , (\cT^{[t]}_{\sigma,\cE}\circ \cE)(\rho) \big) \, ,
\end{align}
with the rotated Petz recovery map $\cT_{\sigma,\cE}^{[t]}$ given by~\eqref{eq_rotPetzt}.
\end{proposition}
We note that the main difference between this proposition and Theorem~\ref{thm_strengthened_mono} is that in~\eqref{eq:strengthened_mono_d_IntegralOUT} the integral is at the very outside, however we have a log-fidelity measure whereas in~\eqref{eq:strengthened_mono_d} we have a measured relative entropy with the integral inside (see Proposition~\ref{prop_measRelFid} for the relation between these two quantities).

\begin{proof}
We first show the assertion of the proposition for the case where $\cE$ is a partial trace and then explain how this result can be lifted to arbitrary quantum channels using the Stinespring dilation (see Proposition~\ref{prop_stinespring}).

Let $\rho_{AB}, \sigma_{AB} \in \Pos(A\otimes B)$ such that $\rho_{AB} \ll \sigma_{AB}$ and $\tr \, \rho_{AB}=1$. Let us recall the multivariate GT inequality given in Theorem~\ref{thm_GT_steinHirschman} for $n=4$ and $p=1$. By Jensen's inequality this reads as
\begin{align} \label{eq_4GTPeir}
\tr \, \ee^{H_1+H_2+H_3+H_4} \leq \int_{- \infty}^\infty \di t  \beta_0( t)\, \norm{\ee^{H_1} \ee^{(1+\ci t)H_2}  \ee^{(1+\ci t)H_3} \ee^{H_4} }_1 \, .
\end{align}
Furthermore the Peierls-Bogoliubov inequality (see Theorem~\ref{thm_Peierls}) ensures that
\begin{align}
\log \frac{\tr \,\ee^{H_5 + H_6}}{\tr\, \ee^{H_5}} \geq \frac{\tr\, H_6 \ee^{H_5}}{\tr\, \ee^{H_5}} \, .
\end{align}
For $H_5 = \log \rho_{AB}$ and $H_6=\frac{1}{2}(- \log \rho_{AB} + \log \sigma_{AB}-\log \sigma_A \otimes \id_B + \log \rho_A \otimes \id_B)$  this simplifies to
\begin{align}
&2\log \tr  \,\ee^{\frac{1}{2}(\log \rho_{AB} + \log \sigma_{AB}-\log \sigma_A \otimes \id_B  + \log \rho_A \otimes \id_B)} \nonumber \\
& \hspace{30mm}\geq  \tr\, \rho_{AB} \left(- \log \rho_{AB} + \log \sigma_{AB}-\log \sigma_A \otimes \id_B  + \log \rho_A \otimes \id_B \right) \, .
\end{align}
We thus find
\begin{align}
D(\rho_{AB} \| \sigma_{AB}) - D(\rho_A \| \sigma_A) 
&= \tr\, \rho_{AB} \left(   \log \rho_{AB} - \log \sigma_{AB} + \log \sigma_A \otimes \id_B  - \log \rho_A \otimes \id_B \right) \\
&\geq - 2\log \tr  \,\ee^{\frac{1}{2}(\log \rho_{AB} + \log \sigma_{AB}-\log \sigma_A \otimes \id_B  + \log \rho_A \otimes \id_B)} \, .
\end{align}
Applying the four operator extension of the GT inequality given in~\eqref{eq_4GTPeir} then gives
\begin{align}
&D(\rho_{AB} \| \sigma_{AB}) - D(\rho_A \| \sigma_A)  \nonumber \\
&\hspace{20mm}\geq -\int_{- \infty}^\infty \di t \beta_0( t)\, \norm{\rho_{AB}^\frac{1}{2} \sigma_{AB}^{\frac{1+\ci t}{2}} \left( \sigma_A^{-\frac{1+\ci t}{2}} \rho_A^{\frac{1}{2}}\otimes \id_B \right)}^2_1 \\
&\hspace{20mm}= -\int_{- \infty}^\infty \di t  \beta_0( t) \, \log F\left(\rho_{AB},  \sigma_{AB}^{\frac{1+\ci t}{2}}  \left( \sigma_A^{-\frac{1+\ci t}{2}} \rho_A \sigma_A^{-\frac{1-\ci t}{2}}  \otimes \id_B \right) \sigma_{AB}^{\frac{1-\ci t}{2}}\right)  \\
&\hspace{20mm}= -\int_{- \infty}^\infty \di t  \beta_0( t) \, \log F\big(\rho_{AB},     \cT^{[t]}_{\sigma_{AB},\tr_B}(\rho_A) \big) \, , \label{eq_stepLogFid1}
\end{align}
where the penultimate step follows by definition of the fidelity.

 Let $V$ be the Stinespring dilation of $\cE$ and let $\rho_{AB} = V \rho V^{\dag}$, $\sigma_{AB} = V \sigma V^{\dag}$ such that $\cE(\rho) = \rho_A$ and $\cE(\sigma) = \sigma_A$. Then, using the fact that the relative entropy is invariant under isometries (see Proposition~\ref{prop_relEntProperties}), we have
 \begin{align}
   D(\rho\|\sigma) - D\big(\cE(\rho) \| \cE(\sigma) \big) 
   &= D(\rho_{AB}\|\sigma_{AB}) - D(\rho_A\|\sigma_A) \\
   &\geq  -\int_{- \infty}^\infty \di t  \beta_0( t) \, \log F\big(\rho_{AB}, \cT^{[t]}_{\sigma_{AB},\tr_B}(\rho_A) \big) \\
   &=  -\int_{- \infty}^\infty \di t \beta_0(t)\log F\big(\rho , (\cT^{[t]}_{\sigma,\cE}\circ \cE)(\rho) \big) \, ,
 \end{align}
where the penultimate step uses~\eqref{eq_stepLogFid1} and the final step uses that the fidelity is invariant under isometries (see Proposition~\ref{prop_propFid}) together with~\eqref{eq_mapIsoInv}. This then completes the proof.
\qed
\end{proof}

\begin{remark} \label{rmk_perfectRec}
Since the mapping $\R \ni t \mapsto \cT_{\sigma,\cE}^{[t]}$ is continuous, Proposition~\ref{prop_strengthened_mono_integralOUT} shows that $D(\rho\|\sigma) = D(\cE(\rho) \| \cE(\sigma))$ implies that $(\cT_{\sigma,\cE}^{[t]} \circ \cE)(\rho)= \rho$ for all $t \in \R$, where we used the nonnegativity property of the fidelity discussed in Proposition~\ref{prop_propFid}.\footnote{Choosing $\rho=\rho_{ABC}$, $\sigma=\id_A \otimes \rho_{BC}$, and $\cE=\tr_C$ we obtain that $I(A:C|B)_{\rho}=0$ implies $\cT^{[t]}_{B\to BC}(\rho_{AB}) = \rho_{ABC}$ for $\cT^{[t]}_{B\to BC}$ defined in~\eqref{eq_PetzRecMap2}.}
\end{remark}

\subsection{Concavity of conditional entropy} \index{concavity conditional entropy}
It is well-known that the conditional entropy is concave, i.e., the function $\St(A\otimes B) \ni \rho \mapsto H(A|B)_{\rho}$ is concave. In the following we show that Theorem~\ref{thm_FR} implies a stronger version of this concavity result.
\begin{corollary} \label{cor_concCondEntropy}
Let $\mu$ be a probability measure on a measurable space $(X,\Sigma)$ and $(\rho_{AB,x})_{x \in X}$ be a sequence of density operators on $A\otimes B$. Then
\begin{align}
H(A|B)_{\bar \rho} - \int_X \mu(\di x) H(A|B)_{\rho_x} \geq \int
_X \mu(\di x) \, \MD\big(\rho_{AB,x} \| \bar \cT_{B \to AB}(\rho_{B,x}) \big) \geq 0 \, ,
\end{align}
where $\bar \rho_{AB}:=Ê\int_X \mu(\di x) \rho_{AB,x}$ and $\bar \cT_{B \to AB}(\cdot) := \bar \cT_{\rho_{AB},\tr_A}(\cdot)$ defined in~\eqref{eq:rotated_petz}. 
\end{corollary}
\begin{proof}
Consider the classical-quantum state
\begin{align}
\omega_{XAB}:=\int_{X} \mu(\di x) \proj{x}_X \otimes \rho_{AB,x} \, .
\end{align}
Theorem~\ref{thm_FR} implies that
\begin{align}
H(A|B)_{\bar \rho} - \int_X \mu(\di x) H(A|B)_{\rho_x} 
&= H(A|B)_{\omega} - H(A|BX)_{\omega} \\
&=I(X:A|B)_{\omega} \\
&\geq \MD\big( \omega_{XAB} \| \bar \cT_{B \to AB}(\omega_{XB}) \big) \\
&= \int_X \mu(\di x) \MD \big( \rho_{AB,x} \| \bar \cT_{B \to AB}(\rho_{B,x}) \big) \, ,
\end{align}
where the final step uses Proposition~\ref{prop_MeasRelEnt}.

Since $\bar \cT_{B \to AB}$ is trace-preserving and completely positive (as discussed in Section~\ref{sec_dpi}), Proposition~\ref{prop_MeasRelEnt} implies $\MD \big( \rho_{AB,x} \| \bar \cT_{B \to AB}(\rho_{B,x}) \geq 0$ for all $x \in X$ which completes the proof.
\qed
\end{proof}
Results that strengthen the concavity of a function can be extremely useful. For example in optimization theory the concept of a strict or even strongly concave function turns out to be important and powerful~\cite{ref:nesterov-book-04,boyd_book}. For this reason we believe that Corollary~\ref{cor_concCondEntropy} may be of interest.
\subsection{Joint convexity of relative entropy} \index{joint convexity relative entropy}
As discussed in Proposition~\ref{prop_relEntProperties}, the relative entropy is jointly convex in its two arguments. As we show next, Theorem~\ref{thm_strengthened_mono} implies a strengthened version of this convexity property.
\begin{corollary} \label{cor_jointConvRelEnt}
Let $\mu$ be a probability measure on a measurable space $(X,\Sigma)$, $(\rho_{A,x})_{x \in X}$ be a sequence of density operator on $A$ with $\rho_A= \int_X \mu(\di x) \rho_{A,x}$ and $(\sigma_{A,x})_{x \in X}$ be a sequence of nonnegative operators on $A$ with $\sigma_A = \int_X \mu(\di x) \sigma_{A,x}$. Then
\begin{align}
\int_X \mu(\di x) D(\rho_{A,x}\| \sigma_{A,x}) - D(\rho_A \| \sigma_A)
\geq \MD\big( \rho_{XA} \| \bar \cT_{A\to XA}(\rho_A) \big) \geq 0\, ,
\end{align}
where $\rho_{XA}:=\int_X \mu(\di x) \proj{x}_X \otimes \rho_{A,x}$, $\sigma_{XA}:=\int_X \mu(\di x) \proj{x}_X \otimes \sigma_{A,x}$,  and $\bar \cT_{A\to XA}(\cdot) :=  \bar \cT_{\sigma_{AX},\tr_X}(\cdot)$ defined in~\eqref{eq:rotated_petz}. 
\end{corollary}
\begin{proof}
Proposition~\ref{prop_MeasRelEnt} shows that
\begin{align}
\int_X \mu(\di x) D(\rho_{A,x}\| \sigma_{A,x}) - D(\rho_A \| \sigma_A) 
&= D(\rho_{XA} \| \sigma_{XA}) - D(\rho_A \| \sigma_A) \\
&\geq  \MD\big( \rho_{XA} \| \bar \cT_{A\to XA}(\rho_A) \big) \\
&\geq 0 \, ,
\end{align}
where the penultimate step uses Theorem~\ref{thm_strengthened_mono}. The final step follows from Proposition~\ref{prop_MeasRelEnt} together with the fact that the recovery map $\bar \cT_{A\to XA}$ is trace-preserving and completely positive. \qed
\end{proof}

\section{Background and further reading}

Quantum Markov chains were introduced in~\cite{accardi83} and their properties were studied carefully~\cite{Pet86,Pet03,HJPW04}. This raised the question how to characterize states with a small conditional mutual information. In~\cite{ILW08} (see~\cite{CSW12} for a simplified argument), it was realized that such states are not necessarily close to any Markov chain. This fact has been taken as an indication that the characterization of states with a small conditional mutual information may be difficult. Subsequently, it has been realized that a more appropriate measure instead of the distance to a Markov chain is to consider how well~\eqref{eq_Markov2} is satisfied~\cite{WL12,Zha12,Kim13,BSW14}. This was made precise by the breakthrough result of Fawzi and Renner~\cite{FR14}. This result generated a sequence of papers~\cite{BHOS14,TB15,SFR15,wilde15,STH15,JRSWW15,SBT16} which finally led to Theorems~\ref{thm_FR} and~\ref{thm_strengthened_mono} which were conjectured in~\cite{WL12}. 

A lower bound that is different to Theorem~\ref{thm_FR} has been obtained by~\cite{brandao11,christandl11}, where it was shown that
\begin{align}
I(A:C|B)_{\rho} 
&\geq \frac{1}{8 \ln 2} \max_{\sigma_{AC} \text{ separable }} \norm{\rho_{AC}- \sigma_{AC}}^2_{\mathrm{LOCC}} \\
&\geq \frac{1}{8 \sqrt{153} \ln 2} \max_{\sigma_{AC} \text{ separable }} \norm{\rho_{AC}- \sigma_{AC}}^2_{2} \, ,
\end{align}
where $\norm{\cdot}_{\mathrm{LOCC}}$ is the so-called LOCC norm.

Theorem~\ref{thm_FR} already found various applications that we do not discuss in the book. To name a few, it has been used to solve problems in thermodynamics~\cite{woods15,kato16} where for example it was shown that approximate quantum Markov chains are approximately thermal~\cite{kato16}. This means that for any $\rho_{ABC}$ such that $I(A:C|B)_{\rho} \leq \eps$ there exists a local Hamiltonian $H=h_{AB}+h_{BC}$, where $h_{AB}$ and $h_{BC}$ only act on $A\otimes B$ and $B\otimes C$, respectively, such that
\begin{align}
D\left(\rho_{ABC} \Big \| \frac{\ee^{-H}}{\tr\, \ee^{-H}} \right) \leq 3\eps \, .
\end{align}

Theorem~\ref{thm_FR} is also potentially useful in computational physics as it implies that systems satisfying a certain locality assumption can be represented efficiently. More precisely, consider a one-dimensional system consisting of $n$ subsystems $S_1, \ldots, S_n$ that feature a certain locality assumption in the sense that for all $k\in [n]$ we have
\begin{align} \label{eq_locAss}
I(S_1,\ldots,S_{k-2}:S_{k}|S_{k-1})_{\rho} \leq \eps \, .
\end{align}
Theorem~\ref{thm_FR} implies that the state $\rho_{S_1,\ldots,S_n}$ describing such a system can be represented efficiently as we can sequentially build it up. To see this let us start with the marginal $\rho_{S_1 S_2}$. Theorem~\ref{thm_FR} implies that there exists a recovery map $\bar \cT_{S_2 \to S_2 S_3 }$ such that
\begin{align}
\rho_{S_1 S_2 S_3} \approx \bar \cT_{S_2 \to S_2 S_3 }(\rho_{S_1 S_2}) \, .
\end{align}
By Theorem~\ref{thm_FR} there exists a recovery map $\bar \cT_{S_3 \to S_3 S_4 }$ such that
\begin{align}
\rho_{S_1 S_2 S_3 S_4} \approx \bar \cT_{S_3 \to S_3 S_4 }(\rho_{S_1 S_2 S_3}) \, .
\end{align}
By continuing like this we can reconstruct the full state $\rho_{S_1, \ldots, S_n}$. All we need to store in order to represent  $\rho_{S_1, \ldots, S_n}$ is a sequence of recovery maps that only takes linear space. To summarize, one-dimensional systems that satisfy the locality assumption~\eqref{eq_locAss} can be efficiently represented by a finite sequence of recovery maps given by Theorem~\ref{thm_FR}.

Theorem~\ref{thm_FR} has been successfully applied in other areas such as high energy physics~\cite{Czech2015,Ding2016,eisert_17}, solid state physics~\cite{brand16,swingle16,zanoci16}, quantum error correction~\cite{preskill17,hayden17}, quantum information theory~\cite{LiWin14,kim16,berta_16,buscemi16,alva16,lemm17}, and foundations of quantum mechanics~\cite{lloyd16}.

We note that Theorem~\ref{thm_FR} has been extended to separable Hilbert spaces~\cite{JRSWW15} (with the caveat that the measured relative entropy is replaced with min-relative entropy). It is an open question if Theorem~\ref{thm_FR} or Theorem~\ref{thm_strengthened_mono} remain valid in the more general algebraic setting. For this purpose the interested reader may have a look at Araki's Gibbs conditions~\cite{araki1974} (see also~\cite{bach00}) and the Tomita-Takesaki theory~\cite{bratteli2012operator}.

\appendix
\chapter{A large conditional mutual information does not imply bad recovery} \label{cha_constructionExample}

\abstract{In this appendix we construct a state $\rho_{ABC}$ that has a large conditional mutual information $I(A:C|B)_{\rho}$, however there exists a recovery map $\cR_{B \to BC}$ that approximately reconstructs $\rho_{ABC}$ from $\rho_{AB}$. More precisely, this example justifies~\eqref{eq_counterExample1} and~\eqref{eq_counterExample2}. }
\vspace{8mm}

 Since the example is purely classical we also use classical notation (i.e., we will speak for example about a distribution instead of a density operator).
Let $\cX = \{1,2,\ldots,2^n\}$ for $n \in \N$, $p,q \in [0,1]$ such that $p+q \leq 1$, and consider  two independent random variables $E_Z$ and $E_Y$ on $\{0,1\}$ and $\{0,1,2\}$, respectively, such that 
\begin{align}
\PP(E_Z=0)=p+q, \quad \PP(E_Y=0)=p, \quad \text{and} \quad \PP(E_Y=1)=q\, .
\end{align}
Let $X \sim \cU(\cX)$, where $\cU(\cX)$ denotes the uniform distribution on $\cX$ and define two random variables by
\begin{align}
Z:=\left \lbrace \begin{array}{l l}
X & \text{if} \quad E_Z=0 \\
U_Z & \text{otherwise} 
\end{array}
  \right.
  \qquad \text{and} \qquad
  Y:=\left \lbrace \begin{array}{l l}
X & \text{if} \quad E_Y=0 \\
Z & \text{if} \quad E_Y=1 \\
U_Y & \text{otherwise} \,  ,
\end{array}
  \right.
\end{align}
where $U_Y \sim \cU(\cX)$ and $U_Z \sim \cU(\cX)$ are independent. This defines a tripartite distribution $P_{XYZ}$. A simple calculation reveals that
\begin{align}
H(X|Y E_Y E_Z) &= p H(X|X E_Z) + q H(X|Z E_Z) + (1-p-q) H(X|U_Y E_Z) \\
&= q \big( (p+q) H(X|X) +(1-p-q) H(X|U_Z) \big) + (1-p-q) H(X) \\
&= n (1-p-q) (1+q) \, . \label{eq_ddss}
\end{align}
Similarly we find
\begin{align}
H(X|Y Z E_Y E_Z) &=  q(1-p-q) H(X|U_Z) +(1-p-q)(1-p-q) H(X|U_Y) \\
&= n(1-p-q) (1-p) \, .\label{eq_ddss2}
\end{align}
We thus obtain
\begin{align} \label{eq_exCMI}
I(X:Z|Y)_P &= H(X|Y) - H(X|YZ)  \\
& \geq  H(X|Y E_Y E_Z) - H(X|YZ E_Y E_Z) - I(X:E_Y E_Z|YZ) \\
& \geq n (1-p-q) (p+q) - \log 6 \, . \label{eq_CMI_comp}
\end{align}

We next define a recovery map $\cR_{Y \to Y'Z'}$ that creates a tuple of random variables $(Y',Z')$ out of $Y$ such that
\begin{align*}
(Y',Z'):= (p^2+q+p q) (Y,Y) + \frac{1}{2}\big(1-p^2-q-p q\big) (Y,U) + \frac{1}{2}\big(1-p^2-q-p q\big) (U',Y) \, ,
\end{align*}
where $U, U'$ are independent uniformly distributed on $\cX$. Let 
\begin{align}
Q_{XY'Z'}:=\cR_{Y \to Y'Z'}(P_{XY}) 
\end{align}
denote the distribution that is generated when applying the recovery map (described above) to $P_{XY}$. In the following we will assume that $n$ is sufficiently large.
It can be verified easily that $Q_{Y'Z'}=P_{YZ}$. 
Since $P_{XYZ}$ and $Q_{XY'Z'}$ are classical distributions we have $D_{\max}(P_{XYZ} \| Q_{XY'Z'}) = \max_{x,y,z} \log \frac{P_{XYZ}(x,y,z)}{Q_{XY'Z'}(x,y,z)}$. We note that $\PP(X=Y)=p+pq+q^2$ according to the distribution $P_{XY}$ and hence
\begin{align}
D_{\max}(P_{XYZ}\| Q_{XY'Z'}) 
&= \max \left \lbrace \log \frac{(p+q)^2}{\PP(X=Y)(p^2+q+pq)}, \log \frac{(1-p-q)q}{\PP(X\ne Y)(p^2+q+pq)}, \right. \nonumber \\
&\hspace{12.5mm} \left.   \log \frac{(p+q)(1-p-q)}{\PP(X=Y)\frac{1}{2}(1-p^2-q-pq)}, \log \frac{(1-p-q)p}{\PP(X=Y)\frac{1}{2}(1-p^2-q-pq)}, \right. \nonumber \\
&\hspace{12.5mm} \left.  \log \frac{(1-p-q)^2}{\PP(X\ne Y)(1-p^2-q-pq)} \right \rbrace \label{eq_Dmax_comp1}
\end{align}
and
\begin{align}
D_{\max}(Q_{XY'Z'}\| P_{XYZ}) 
&= \max \left \lbrace \log \frac{\PP(X=Y)(p^2+q+pq)}{(p+q)^2}, \log \frac{\PP(X\ne Y)(p^2+q+pq)}{(1-p-q)q}, \right. \nonumber \\
&\hspace{12.8mm} \left. \log \frac{\PP(X=Y)\frac{1}{2}(1-p^2-q-pq)}{(p+q)(1-p-q)},  \log \frac{\PP(X=Y)\frac{1}{2}(1-p^2-q-pq)}{(1-p-q)p}, \right. \nonumber \\
&\hspace{12.8mm} \left. \log \frac{\PP(X\ne Y)(1-p^2-q-pq)}{(1-p-q)^2} \right \rbrace \, .\label{eq_Dmax_comp2}
\end{align}

For $\kappa<\infty$, $p=\frac{1}{2}$, $q=0$, and $n$ sufficiently large we find by combining~\eqref{eq_CMI_comp} with~\eqref{eq_Dmax_comp1}
\begin{align} \label{eq_resEx1}
\kappa \, D_{\max}\big(P_{XYZ}\| \cR_{Y \to YZ}(P_{XY}) \big) = \kappa <\frac{n}{4} - \log 6 \leq I(X:Z|Y)_P \, ,
\end{align}
which justifies~\eqref{eq_counterExample1}.
For $\kappa<\infty$, $p=q=\frac{1}{4}$, and $n$ sufficiently large~\eqref{eq_CMI_comp} and~\eqref{eq_Dmax_comp2} imply 
\begin{align}
\kappa\, D_{\max}\big( \cR_{Y \to YZ}(P_{XY}) \| P_{XYZ} \big)  = \kappa \log \frac{15}{8} < \frac{n}{4} - \log 6 \leq I(X:Z|Y)_P  \, ,
\end{align}
justifying~\eqref{eq_counterExample2}.

These examples show that there exist classical tripartite distributions $P_{XYZ}$ with a large conditional mutual information $I(X:Y|Z)_P$ and a recovery map $\cR_{Y \to YZ}$ such that $\cR_{Y \to YZ}(P_{XY})$ is close to $P_{XYZ}$ and $\cR_{Y \to YZ}(P_{Y})=P_{YZ}$. The closeness is measured with respect to the max-relative entropy.

\chapter{Example showing the optimality of the $\Lambda_{\max}$-term}\label{app_Ex_tight}

\abstract{In this appendix, we construct a classical example showing that~\eqref{eq_CMI_UB} is essentially tight in the sense that it is no longer valid if the max-relative entropy in the definition of $\Lambda_{\max}(\rho_{AB} \| \cR_{B\to B})$ is replaced with $\Lambda_{\alpha}(\rho_{AB} \| \cR_{B\to B})$ for any $\alpha <\infty$. In other words, for any $\alpha <\infty$ we construct a density operator $\rho_{ABC} \in \St(A\otimes B \otimes C)$ and a recovery map $\cR_{B \to BC} \in \TPCP(B,B\otimes C)$ that satisfy~\eqref{eq_TSquantum_intro}.}
\vspace{8mm}

Our construction is purely classical which is the reason that we switch to the classical notation.
 Let $\cS=\{0,\ldots,2^{n}-1 \}$ and consider a tripartite distribution $Q_{XYZ}$ defined via the random variables $X\sim \cU(\cS)$ and $X=Y=Z$. Let $Q'_{XYZ}$ be the distribution defined via the random variables $X \sim \cU(\cS)$, $Y \sim \cU(\cS)$ where $X$ and $Y$ are independent, $\cU(\cS)$ denotes the uniform distribution on $\cS$ and $Z= (X+Y)\!\!\!\mod 2^n$. For $p\in [0,1]$ we define a binary random variable $E$ such that $\PP(E=0)=p$. Consider the distribution
\begin{align}
P_{XYZ}=\left \lbrace \begin{array}{l l}
Q_{XYZ} & \text{if } E=0 \\
Q'_{XYZ}& \text{if } E=1  \, .
\end{array} \right .
\end{align}
We next define two recovery maps $\tilde \cR_{Y\to Y' Z'}$ and $\bar \cR_{Y \to Y'Z'}$ that create the tuples $(Y',Z')$ out of $Y$ such that 
\begin{align}
(Y',Z')= (Y,Y) \qquad \text{and} \qquad (Y',Z')=\big(U,(Y-U)\!\!\!\!\! \mod 2^n \big) \, ,
\end{align}
where $U \sim  \cU(\cS)$, respectively. We then define another recovery map as 
\begin{align}
\cR_{Y \to Y' Z'}:=p \tilde \cR_{Y \to Y' Z'} +{(1-p)} \bar \cR_{Y \to Y' Z'} \, .
\end{align}
We note that the recovery map satisfies $\cR_{Y\to Y'Z'}(P_{Y})=P_{YZ}$.
A simple calculation shows that
\begin{align}
H(X|YE)_P = p H(X|Y)_Q + (1-p)H(X|Y)_{Q'} = (1-p) n
\end{align}
and
\begin{align}
H(X|YZE)_P = p H(X|YZ)_Q + (1-p)H(X|YZ)_{Q'} = 0 \, .
\end{align}
We thus find
\begin{align}
I(X:Z|Y)_{P} &= H(X|Y) - H(X|YZ) \\
& \geq H(X|YE) - H(X|YZE) - I(X:E|YZ) \\ 
&\geq (1-p)n -h(p)  \, . \label{eq_CMI1}
\end{align}
The distribution $\cR_{Y\to Y'Z'}(P_{XY})$ generated by applying the recovery map to $P_{XY}$ can be decomposed as
\begin{align}
\cR_{Y\to Y'Z'}(P_{XY})= p \left(p \tilde S_{XYZ} + (1-p) \bar S_{XYZ} \right) + (1-p) \left(p \tilde S'_{XYZ} + (1-p) \bar S'_{XYZ} \right) \, ,
\end{align}
where $\tilde S_{XYZ} =\tilde \cR_{Y \to Y' Z'}(Q_{XY})$, $\bar S_{XYZ}=\bar \cR_{Y \to Y' Z'}(Q_{XY})$, $\tilde S'_{XYZ} =\tilde \cR_{Y \to Y' Z'}(Q'_{XY})$, and $\bar S'_{XYZ}=\bar \cR_{Y \to Y' Z'}(Q'_{XY})$.
The joint convexity of the relative entropy~\cite[Theorem~2.7.2]{cover} then implies
\begin{align}
&D\big(P_{XYZ} \| \cR_{Y\to Y'Z'}(P_{XY}) \big) \nonumber \\
& \hspace{10mm} \leq p D\big( Q_{XYZ} \| p \tilde S_{XYZ} + (1-p) \bar S_{XYZ} \big) + (1-p) D\big(Q'_{XYZ} \|  p \tilde S'_{XYZ} + (1-p) \bar S'_{XYZ} \big)
\end{align}
A simple calculation shows that
\begin{align}
D\big( Q_{XYZ} \| p \tilde S_{XYZ} + (1-p) \bar S_{XYZ} \big)
&= \sum_{x=y=z} Q_{XYZ}(x,y,z)\log \frac{Q_{XYZ}(x,y,z)}{p \tilde S_{XYZ}(x,y,z) + (1-p) \bar S_{XYZ}(x,y,z)} \nonumber \\
&\leq \frac{2^{-n}}{p 2^{-n}} = \log \frac{1}{p}
\end{align}
and
\begin{align}
&D\big( Q'_{XYZ} \| p \tilde S'_{XYZ} + (1-p) \bar S'_{XYZ} \big) \nonumber \\
&\hspace{15mm}=\sum_{x,y,z=x+y \!\!\!\!\!\mod 2^n } \!\!\!\!\!Q'_{XYZ}(x,y,z)\log\frac{Q'_{XYZ}(x,y,z)}{p \tilde S'_{XYZ}(x,y,z) + (1-p) \bar S'_{XYZ}(x,y,z)}  \\
&\hspace{15mm}\leq \frac{2^{-2n}}{p 2^{-2n}} = \log \frac{1}{p} \, .
\end{align}
We thus have
\begin{align} \label{eq_relEnt3}
D\big(P_{XYZ} \| \cR_{Y\to Y'Z'}(P_{XY}) \big)\leq \log\frac{1}{p} \, .
\end{align}
We note that the recovery map $\cR_{Y \to Y'} = \tr_{Z'} \circ \cR_{Y \to Y' Z'}$ leaves the uniform distribution $Q'_{XY}$ invariant, i.e., $\cR_{Y \to Y'} (Q'_{XY})= Q'_{XY}$. As a result we find
\begin{align} \label{eq_relEnt2}
\Lambda_{\alpha}(P_{XY} \| \cR_{Y \to Y'})
&\leq D_{\alpha}(P_{XY}\| Q'_{XY}) \\
&=\frac{1}{\alpha -1} \log \big(2^{-n} (1-p)^{\alpha}(2^n -1) + 2^{-n}(1-p+p 2^{n})^\alpha \big) \, ,
\end{align}
where the final step follows by definition of the $\alpha$-R\'enyi relative entropy and a straightforward calculation.

Recall that we need to prove~\eqref{eq_TSquantum_intro}, which in the classical notation reads as
\begin{align} \label{eq_toshow}
D\big(P_{XYZ} \| \cR_{Y\to Y'Z'}(P_{XY}) \big) + \Lambda_{\alpha}(P_{XY} \| \cR_{Y \to Y'})  <I(X:Z|Y)_P  \, ,
\end{align}
for all $\alpha <\infty$.
As mentioned in~\eqref{eq_mono}, the $\alpha$-R\'enyi relative entropy is monotone in $\alpha$ which shows that it suffices to prove~\eqref{eq_toshow} for all $\alpha  \in  (\alpha_0,\infty)$, where $\alpha_0\geq 0$ can be arbitrarily large. 

Combining~\eqref{eq_relEnt3} and~\eqref{eq_relEnt2} shows that for any $\alpha \in (\alpha_0,\infty)$ where $\alpha_0$ is sufficiently large, $p=\alpha^{-2}$, and $n=\alpha$
\begin{align*}
D\big(P_{XYZ} \| \cR_{Y\to Y'Z'}(P_{XY}) \big) + \Lambda_{\alpha}(P_{XY} \| \cR_{Y \to Y'})
&\leq 2 \log \alpha + \frac{1}{\alpha-1} \log \left( 1 +2^{-\alpha} (1+ \alpha^{-2} 2^{\alpha})^\alpha \right)\, ,
\end{align*}
where we used that $(1-\alpha^{-2})^\alpha(2^\alpha-1) \leq 2^{\alpha}$ for $\alpha \geq 1$.
Using the simple inequality $\log(1+x)\leq \log x + \frac{2}{x}$ for $x\geq1$ gives
\begin{align}
&D\big(P_{XYZ} \| \cR_{Y\to Y'Z'}(P_{XY}) \big) + \Lambda_{\alpha}(P_{XY} \| \cR_{Y \to Y'})  \nonumber \\
&\hspace{20mm} \leq 2 \log \alpha -\frac{\alpha}{\alpha -1} + \frac{\alpha}{\alpha -1} \log\left(1+\frac{2^\alpha}{\alpha^2} \right) + \frac{2}{\alpha -1} 2^\alpha\left(1+\frac{2^\alpha}{\alpha^2} \right)^{-\alpha} \\
&\hspace{20mm} \leq 2 \log \alpha -\frac{\alpha}{\alpha -1} + \frac{\alpha}{\alpha -1} \log\left(1+\frac{2^\alpha}{\alpha^2} \right) + 2^{-\alpha}  \, ,
\end{align}
where the final step is valid since $\alpha$ is assumed to be sufficiently large. 
Using once more $\log{(1+x)}\leq \log x + \frac{2}{x}$ for $x\geq1$ gives
\begin{align}
&D\big(P_{XYZ} \| \cR_{Y\to Y'Z'}(P_{XY}) \big) +  \Lambda_{\alpha}(P_{XY} \| \cR_{Y \to Y'}) \nonumber \\
&\hspace{30mm}\leq 2 \log \alpha +\frac{\alpha}{\alpha -1} \left( \alpha -2 \log \alpha -1 + \frac{2\alpha^2}{2^\alpha} \right)+ 2^{-\alpha}   \\
&\hspace{30mm}= \alpha -\frac{2}{\alpha-1} \log \alpha + 2^{-\alpha} \mathrm{poly}(\alpha) \, ,
\end{align}
where $\mathrm{poly}(\alpha)$ denotes an arbitrary polynomial in $\alpha$.
As a result, we obtain for a sufficiently large $\alpha$
\begin{align}
D\big(P_{XYZ} \| \cR_{Y\to Y'Z'}(P_{XY}) \big) + \Lambda_{\alpha}(P_{XY} \| \cR_{Y \to Y'}) 
&< \alpha - \frac{2}{\alpha} \label{eq_strict}\\
&\leq \alpha - \alpha^{-1} - h(\alpha^{-2}) \label{eq_binEnt}\\
&\leq I(X:Z|Y)_P \label{eq_finals} \, .
\end{align}
The two steps~\eqref{eq_strict} and~\eqref{eq_binEnt} are both valid because $\alpha$ is sufficiently large. The final step uses~\eqref{eq_CMI1}.



\chapter{Solutions to exercises}
\label{app_solutions} 

\abstract{In this appendix we give solutions to the exercises stated throughout the book. The exercises are chosen such that they can be solved without major difficulties. They serve the purpose of a verification possibility for the reader to check if she has understood the presented subject.}
\vspace{8mm}

\subsubsection*{Solution to Exercise~\ref{ex_charMarkov}}
We view statement~\eqref{eq_cl_def} as the definition of a (classical) Markov chain. It thus remains to show that~\eqref{eq_cl_recovery} and~\eqref{eq_cl_CMI} are both equivalent to~\eqref{eq_cl_def}.
Bayes' theorem ensures that $P_{XYZ}=P_{XZ|Y} P_Y$ and $P_{XY}=P_{X|Y}P_{Y}$. As a result we find that
\begin{align}
P_{XZ|Y} = P_{X|Y} P_{Z|Y} \quad \iff \quad P_{XYZ}=P_{XY} P_{Z|Y} \, ,
\end{align}
 which shows that~\eqref{eq_cl_def} is equivalent to~\eqref{eq_cl_recovery}. By definition of the relative entropy and the conditional mutual information we have
 \begin{align}
 I(X:Z|Y)_P = D(P_{XZ|Y}\| P_{X|Y} P_{Z|Y})\, .
 \end{align}
Recalling that $D(P\|Q) = 0$ if and only if $P=Q$ shows that~\eqref{eq_cl_def} is equivalent to~\eqref{eq_cl_CMI}. 

\subsubsection*{Solution to Exercise~\ref{ex_varFormulaCMI}}
This solution follows the arguments presented in~\cite{Horodecki2005,ILW08}.
A simple calculation shows that
\begin{align} \label{eq_clCMIEx}
I(X:Z|Y)_P = \sum_{x,y,z} P_{XYZ}(x,y,z) \log \frac{P_{Z|XY}(z|xy)}{P_{Z|Y}(z|y)} \, .
\end{align}
The distribution $P_{XYZ}$ can be decomposed as $P_{XYZ}=P_Y P_{Z|Y} P_{X|YZ}$ and any Markov chain $Q_{XYZ}$ can be written as $Q_{XYZ}=Q_Y Q_{Z|Y} Q_{X|Y}$. We thus find
\begin{align}
D(P_{XYZ}\|Q_{XYZ}) 
&= \sum_{x,y,z} P_{XYZ}(x,y,z) \left( \log \frac{P_Y(y)}{Q_Y(y)} + \log \frac{P_{Z|Y}(z|y)}{Q_{Z|Y}(z|y)} + \log \frac{P_{X|YZ}(x|yz)}{Q_{X|Y}(x|y)} \right) \\
&=D(P_Y\|Q_Y) + D(P_{Z|Y}\| Q_{Z|Y}) + \sum_{x,y,z} P_{XYZ}(x,y,z) \log \frac{P_{X|YZ}(x|yz)}{Q_{X|Y}(x|y)} \\
&= D(P_Y\|Q_Y) + D(P_{Z|Y}\| Q_{Z|Y}) + D(P_{X|Y}\|Q_{X|Y}) +  I(X:Z|Y)_P \, ,
\end{align}
where the final step uses~\eqref{eq_clCMIEx}. Since the relative entropy is nonnegative and zero if and only if the two arguments coincide this proves the assertion.

\subsubsection*{Solution to Exercise~\ref{ex_pNorm}}
That the Schatten $p$-norm satisfies the nonnegativity and absolute homogeneity property is obvious from its definition. It thus remains to prove the triangle inequality. The Schatten $p$-norm can be written as the $\ell_p$-norm of the singular values, i.e., for $L \in \LL(A)$ we have
\begin{align} \label{eq_SchattenSVD}
\norm{L}_p = \left( \sum_{k=1}^{\dim(L)} \sigma_k(L)^p \right)^{\frac{1}{p}} \, ,
\end{align}
where $(\sigma_k(L))_{k=1}^{\dim(L)}$ denote the singular values of $L$. The Minkowski inequality (see, e.g.,~\cite[Theorem~III.1]{simon_book}) then implies the triangle inequality for Schatten norms.

The identiy~\eqref{eq_SchattenSVD} shows that $\norm{L}_p=\norm{L^\dagger}_p$ as singular values are invariant under conjugate transposition. The singular value decomposition ensures that there exist unitaries $U,V\in \U(A)$ such that $L=U \Lambda V^\dagger$, where $\Lambda$ is a diagonal matrix containing the singular values of $L$. Using the fact that Schatten norms are unitarily invariant gives
\begin{align}
\norm{L L^\dagger}_p
=\norm{L^\dagger L}_p 
=\norm{\Lambda \Lambda}_p
=\left(\sum_{k=1}^{\dim(L)} \sigma_k(L)^{2p} \right)^{\frac{2}{2p}} 
=\norm{L}_{2p}^2 \, .
\end{align}

The fact that Schatten $p$-norms are monotone in $p$ follows directly from the monotonicity of $\ell_p$-norms via~\eqref{eq_SchattenSVD}. To see this let $0 \ne x \in \C^{d}$ and $1\leq p \leq q$ and define $y:=\frac{x}{\norm{x}_p}$. Since $|y_k|\leq 1$ and $\norm{y}_p=1$ we find
\begin{align}
\norm{y}_q 
= \left(\sum_{k=1}^d |y_k|^q \right)^{\frac{1}{q}} 
\leq \left(\sum_{k=1}^d |y_k|^p \right)^{\frac{1}{q}} 
= \norm{y}_p^{\frac{p}{q}}
=1 \, .
\end{align}
As a result we have
\begin{align}
\norm{x}_q = \norm{\norm{x}_p y}_q = \norm{x}_p \norm{y}_q \leq \norm{x}_p \, .
\end{align}

If $(\sigma_k(L_1))_{k=1}^{\dim(L_1)}$ and $(\sigma_k(L_2))_{k=1}^{\dim(L_2)}$ denote the singular values of $L_1$ and $L_2$, respectively, then the $\dim(L_1) \dim(L_2)$ singular values of $L_1 \otimes L_2$ are given by all possible multiplications of a singular values of $L_1$ with a singular values of $L_2$. This directly implies that Schatten norms are multiplicative under tensor products.
\subsubsection*{Solution to Exercise~\ref{ex_CPvsP}}
Consider the transpose map $\cT: \LL(A) \to \LL(A)$ that is given by $\cT \, : \, X \mapsto X^{\mathrm{T}}$, where $X^{\mathrm{T}}$ denotes the transpose of $X$ with respect to some fixed basis. The transpose map is clearly positive, since for any state $\ket{\psi}$ we have
\begin{align}
\bra{\psi} X^{\mathrm{T}} \ket{\psi} 
=\bra{\psi} \bar{X}^{\dagger} \ket{\psi} 
= \overline{\bra{\psi} \bar{X} \ket{\psi} }
= \bra{\bar \psi} X \ket{\bar \psi} 
\geq 0 \, .
\end{align}
The transpose map is however not completely positive. To see this it suffices to consider a two-dimensional system, i.e., $\dim(A)=2$. For the computational basis $\{\ket{0},\ket{1} \}$ and the maximally entangled state
\begin{align}
\ket{\phi}_{AB}= \frac{1}{\sqrt{2}}\left(\ket{00}_{AB}+\ket{11}_{AB} \right) 
\end{align}
we find that
\begin{align}
(\cT_A \otimes \cI_B)(\proj{\phi}_{AB}) = \frac{1}{2} \left(\ket{00} \bra{00}_{AB}+\ket{10} \bra{10}_{AB}+\ket{01} \bra{01}_{AB}+\ket{11} \bra{11}_{AB} \right) \, ,
\end{align}
which is not a positive operator as it has eigenvalues $\pm \frac{1}{2}$.

\subsubsection*{Solution to Exercise~\ref{ex_Kraus}}
The finite sequence $(E_k)_{k\in[r]}$ of Kraus operators is not uniquely determined by $\cE$. It can be shown~\cite[Theorem~2.1]{Wolf_skript} that two sets of Kraus operators $(E_k)$ and $(E'_k)$ represent the same map $\cE$ if and only if there is a unitary $U$ such that $E_k = \sum_{j} U_{kj} E'_j$ (where the smaller set is padded with zeros). A proof of this statement can be found in~\cite[Theorem~2.1]{Wolf_skript}.

\subsubsection*{Solution to Exercise~\ref{ex_convexOpt}}
Since the two optimization problems in Lemma~\ref{lem_varFormulaRelEnt} are equivalent it suffices to show that one of them is a convex optimization problem. We do so for the first optimization problem. Since every Hermitian operator can be written as the logarithm of a nonnegative operator we can rewrite~\eqref{eq_varFormulaRelEnt_1} as
\begin{align} \label{eq_convOpt}
D(\rho\| \sigma) 
=\sup_{H \in \Her(A)}\left \lbrace\tr\, \rho H - \log \tr\, \ee^{\log \sigma + H} \right \rbrace\, . 
\end{align}
The set of Hermitian operators is clearly convex. Furthermore, the function $H \mapsto \log \tr\, \ee^{\log \sigma + H}$ is convex on the set of Hermitian operators.
To see this, we recall the variational formula given in~\eqref{eq_step1} which shows that for any $t \in [0,1]$ and $H_1,H_1 \in \Her(A)$ we have
\begin{align}
t \log \tr\, \ee^{\log \sigma + H_1} +(1-t) \log \tr\, \ee^{\log \sigma + H_2}
&\geq \max_{\rho \in \St(A)} \big\{ \tr\, \big(t H_1 + (1-t) H_2 \big) \rho - D(\rho \| \sigma) \big\} \\
&= \log \tr\, \ee^{\log \sigma + (t H_1 + (1-t) H_2)} \, .
\end{align}
This shows that $H \mapsto \log \tr\, \ee^{\log \sigma + H}$ is a convex function and hence~\eqref{eq_convOpt} is a convex optimization problem.

\subsubsection*{Solution to Exercise~\ref{ex_tropp}}
We first prove~\eqref{eq_varTrace}. Klein's inequality for $f(t) = t \log t$ (which is strictly convex for $t>0$) implies that 
\begin{align}
\tr\, B \geq \tr\, X - \tr\, X \log X + \tr\, X \log B \, ,
\end{align}
where equality holds if and only if $X=B$. This already proves~\eqref{eq_varTrace}. Applying~\eqref{eq_varTrace} for $B=\ee^{H + \log \sigma}$ gives~\eqref{eq_varTropp}. 

\subsubsection*{Solution to Exercise~\ref{ex_pinchingFT}}
A simple calculation shows that for any $\kappa >0$, $\mu_{\kappa}$ is a probability distribution on $\R$, i.e., $\mu_{\kappa}(t) \geq 0$ for all $t \in \R$ and $\int_{-\infty}^{\infty} \mu_{\kappa}(\di t) =1$. Furthermore
\begin{align}
\hat \mu_{\kappa}(\omega):= \int_{-\infty}^{\infty} \mu_{\kappa}(\di t) \ee^{-\ci \omega t} = \frac{3}{\kappa} (\mathrm{tri}_\kappa \star \mathrm{tri}_\kappa)(\omega) \, .
\end{align}
This then straightforwardly implies the five properties mentioned in Section~\ref{sec_pinching}.

\subsubsection*{Solution to Exercise~\ref{ex_pinchingMapTPCP}}
By the operator-sum representation of quantum channels (see Proposition~\ref{prop_operatorSum}) the pinching map defined by~\eqref{eq_pinchingDef} is trace-preserving and completely positive since
\begin{align}
\sum_{\lambda \in \spec(H)} \Pi_{\lambda} \Pi_{\lambda}  = \sum_{\lambda \in \spec(H)} \Pi_{\lambda} = \id_A \, .
\end{align}
This also shows that the pinching maps is unital, i.e., $\cP_{H}(\id_A) = \id_A$.

\subsubsection*{Solution to Exercise~\ref{exercise_tensorProduct}}
Let $(\ell_{k\ell})$ denote the entries of the operator $L_1$ (if we view it as a matrix). By definition of the tensor product we find
\begin{align}
\tr\, L_1 \otimes L_2 = \sum_{k} \tr\, \ell_{kk} L_2 = \sum_k \ell_{kk} \, \tr\, L_2 = (\tr\, L_1)(\tr\, L_2) \, ,
\end{align}
which proves the first identity.

Every nonnegative operator can be diagonalized, i.e., there exist unitaries $U_1 \in \U(A)$ and $U_2 \in \U(B)$ such that $C_1 = U_1 \Lambda_1 U_1^\dagger$ and $C_2 = U_2 \Lambda_2 U_2^\dagger$ for diagonal matrices $\Lambda_1$ and $\Lambda_2$ with nonnegative entries. We then find
\begin{align}
\log C_1 \otimes C_2 
&=  \log (U_1 \otimes U_2) (\Lambda_1 \otimes \Lambda_2) (U_1^\dagger \otimes U_2^\dagger)\\
&= (U_1 \otimes U_2) \big( \log \Lambda_1 \otimes \Lambda_2 \big) (U_1^\dagger \otimes U_2^\dagger) \\
&=(U_1 \otimes U_2) \big( (\log \Lambda_1) \otimes \id_B + \id_A \otimes (\log \Lambda_2) \big) (U_1^\dagger \otimes U_2^\dagger) \\
& = (\log C_1) \otimes \id_B + \id_A \otimes (\log C_2) \, ,
\end{align}
which proves the second identity.

The third identity follows from a known relation between the tensor product and the direct sum, i.e.,
\begin{align}
\exp(L_1) \otimes \exp(L_2) 
= \exp(L_1 \oplus L_2)
= \exp(L_1 \otimes \id_B + \id_A \otimes L_2) \, .
\end{align}

\subsubsection*{Solution to Exercise~\ref{ex_3GT}}
Let $B_1,B_2,B_3 \in \Pos(A)$ be such that $B_k:=\log H_k$ for $k\in \{1,2,3\}$. Essentially the same steps as in the proof presented in Section~\ref{sec_specPinchGT} show that
\begin{align}
&\log \tr \exp(\log B_1 + \log B_2 + \log B_3) \nonumber \\
&\hspace{20mm}= \frac{1}{m} \log \tr \exp( \log B_1^{\otimes m} + \log B_2^{\otimes m} + \log B_3^{\otimes m}) \\
&\hspace{20mm}\leq \frac{1}{m} \log \tr \exp\big( \log B_1^{\otimes m} + \log B_2^{\otimes m} + \log \cP_{B_2^{\otimes m}}( B_3^{\otimes m}) \big) + \frac{\log \poly(m)}{m}Ê\\
&\hspace{20mm} = \frac{1}{m} \log \tr \exp\big( \log B_1^{\otimes m} + \log B_2^{\frac{1}{2}\otimes m} \cP_{B_2^{\otimes m}}( B_3^{\otimes m}) B_2^{\frac{1}{2}\otimes m}\big)+ \frac{\log \poly(m)}{m} \, ,
\end{align}
where the first step uses Exercise~\ref{exercise_tensorProduct}. The inequality step follows from the pinching inequality (see Lemma~\ref{lem_propertiesPinching}), together with the fact that the logarithm is operator monotone (see Table~\ref{table_opConvex}) and $H \mapsto \tr \, \ee^{ H}$ is monotone (see Proposition~\ref{prop_traceFunctions}). Furthermore we use the observation presented in Remark~\ref{remark_types}, i.e., that the number of distinct eigenvalues of $B_2^{\otimes m}$ grows polynomially in $m$. The final step uses that $\cP_C(B)$ always commutes with $C$ (see Lemma~\ref{lem_propertiesPinching}).

Repeating the same arguments gives
\begin{align}
&\log \tr \exp(\log B_1 + \log B_2 + \log B_3) \nonumber \\
&\hspace{10mm} \leq \frac{1}{m} \log \tr \exp\big( \log B_1^{\otimes m} + \log \cP_{B_1^{\otimes m}}( B_2^{\frac{1}{2}\otimes m} \cP_{B_2^{\otimes m}}( B_3^{\otimes m}) B_2^{\frac{1}{2}\otimes m} )\big)+ \frac{\log \poly(m)}{m} \\
&\hspace{10mm}= \frac{1}{m} \log \tr \, B_1^{\otimes m} \cP_{B_1^{\otimes m}}( B_2^{\frac{1}{2}\otimes m} \cP_{B_2^{\otimes m}}( B_3^{\otimes m}) B_2^{\frac{1}{2}\otimes m} )  + \frac{\log \poly(m)}{m} \\
&\hspace{10mm}= \frac{1}{m} \log \tr \, B_1^{\otimes m}  B_2^{\frac{1}{2}\otimes m} \cP_{B_2^{\otimes m}}( B_3^{\otimes m}) B_2^{\frac{1}{2}\otimes m}    + \frac{\log \poly(m)}{m} \, ,
\end{align}
where the final step uses Lemma~\ref{lem_propertiesPinching}. The integral representation for pinching maps given by Lemma~\ref{lem_IntegralPinching} shows that
\begin{align}
&\log \tr \exp(\log B_1 + \log B_2 + \log B_3) \nonumber \\
&\hspace{10mm}= \frac{1}{m} \log \int_{-\infty}^{\infty} \mu_{\Delta_{B_2^{\otimes m}}}(\di t) \tr \, B_1^{\otimes m}  B_2^{\frac{1}{2}\otimes m} \ee^{\ci t B_2^{\otimes m}} B_3^{\otimes m}\ee^{-\ci t B_2^{\otimes m}} B_2^{\frac{1}{2}\otimes m}    + \frac{\log \poly(m)}{m}  \\
&\hspace{10mm}\leq \log \sup_{t \in \R} \tr \, B_1 B_2^{\frac{1+\ci t}{2}} B_3 B_2^{\frac{1-\ci t}{2}} +  \frac{\log \poly(m)}{m}  \, ,
\end{align}
where the final step uses Exercise~\ref{exercise_tensorProduct} and that for any $B \in \Pos(A)$ and any $t \in \R$ there exists a $s \in \R$ such that $\ee^{\ci t \log B} = \ee^{\ci s B}$. Considering the limit $m \to \infty$ finally gives
\begin{align}
\tr \exp(\log B_1 + \log B_2 + \log B_3) \leq \sup_{t \in \R} \tr \, B_1 B_2^{\frac{1+\ci t}{2}} B_3 B_2^{\frac{1-\ci t}{2}} \, ,
\end{align}
which proves the desired inequality.

\subsubsection*{Solution to Exercise~\ref{ex_logdet}}
Every positive definite matrix can be diagonalized, i.e., there exists a unitary $U \in \U(A)$ such that $B=U \Lambda U^\dagger$ where $\Lambda$ is a diagonal matrix containing the eigenvalues $(\lambda_k)_k$ of $B$. We thus have
\begin{align}
\tr \log B
= \tr\, U (\log \Lambda) U^\dagger
= \tr\, \log \Lambda
= \sum_k \log \lambda_k
=\log \prod_k \lambda_k
= \log \det B \, .
\end{align}

\subsubsection*{Solution to Exercise~\ref{ex_DPI_SSA}}
If we choose $\rho = \rho_{ABC}$, $\sigma = \id_A \otimes \rho_{BC}$ and $\cE(\cdot) = \tr_C(\cdot)$, Proposition~\ref{prop_PetzDPI} simplifies to the statement that the following are equivalent
\begin{enumerate}
\item $I(A:C|B)_{\rho}=0$
\item $\exists$ $\cR_{B \to BC}$ such that $\cR_{B \to BC}(\rho_{AB})=\rho_{ABC}$ and $\cR_{B\to BC}(\rho_B)=\rho_{BC}$
\end{enumerate}
In particular the recovery map can be chosen to be the rotated Petz recovery map given in~\eqref{eq_PetzRecMap2}.
This is exactly the statement of Theorem~\ref{thm_PetzCMI}.

\subsubsection*{Solution to Exercise~\ref{ex_redMainRes}}
This exercise is so simple that it does not require a solution. If we evaluate Theorem~\ref{thm_strengthened_mono} for $\rho=\rho_{ABC}$, $\sigma = \id_A \otimes \rho_{BC}$, and $\cE = \tr_C$ we immediately obtain Theorem~\ref{thm_FR}.

\subsubsection*{Solution to Exercise~\ref{ex_propRecMap}}
The recovery map $\bar \cT_{\sigma,\cE}$ is clearly completely positive. It is also trace-non-increasing as for any $t \in \R$
\begin{align}
\tr  \cT^{[t]}_{\sigma,\cE}(X)
&= \tr \, \sigma \cE^\dagger\big( \cE(\sigma)^{-\frac{1+\ci t}{2}} X \cE(\sigma)^{-\frac{1-\ci t}{2}}   \big)\\
&= \tr\, \cE(\sigma)  \cE(\sigma)^{-\frac{1}{2}} X  \cE(\sigma)^{-\frac{1}{2}}\\
&= \tr\, \Pi_{\cE(\sigma)}X  \\
&\leq \tr\, X \, ,
\end{align}
where the final inequality step is an identity in case $\cE(\sigma)$ has full support.
The recovery map $\bar \cT_{\sigma,\cE}$ clearly is explicit, universal and perfectly recovers $\sigma$ from $\cE(\sigma)$.  

For $\cE=\cI$ we find
\begin{align}
\bar \cT_{\sigma,\cI}(\cdot) = \int_{-\infty}^{\infty} \di t \beta_0( t) \Pi_{\sigma}(\cdot) \Pi_{\sigma} = \Pi_{\sigma}(\cdot)\Pi_{\sigma} \, ,
\end{align}
which proves the normalization property.

Finally for $\omega \in \Poss(R)$ we have
\begin{align}
&\bar \cT_{\sigma \otimes \omega, \cE\otimes \cI_R }(\cdot) \nonumber \\
 &= \int_{-\infty}^{\infty} \di t \beta_0( t) \sigma^{\frac{1+\ci t}{2}} \otimes \omega^{\frac{1+\ci t}{2}} (\cE \otimes \cI_R)^\dagger\big( \cE(\sigma)^{-\frac{1+\ci t}{2}} \otimes \omega^{-\frac{1+\ci t}{2}} (\cdot)\cE(\sigma)^{-\frac{1-\ci t}{2}} \otimes \omega^{-\frac{1-\ci t}{2}}  \big) \sigma^{\frac{1-\ci t}{2}} \otimes \omega^{\frac{1-\ci t}{2}} \nonumber\\
 &= (\bar \cT_{\sigma,\cE} \otimes \cI_R )(\cdot) \, ,
\end{align}
which proves the last property and thus completes the exercise.

\backmatter



\fussy
\bibliographystyle{arxiv_no_month}
\bibliography{bibliofile}


\printindex

\end{document}